\documentclass[12pt]{wlscirep}
\usepackage[utf8]{inputenc}
\usepackage{float}
\usepackage[T1]{fontenc}
\usepackage[draft]{changes}
\usepackage{csquotes}
\usepackage{amsmath}
\usepackage{xspace}
\usepackage{bm}
\usepackage{amssymb}
\usepackage[ruled,vlined]{algorithm2e}
\usepackage{todonotes}
\usepackage{amsthm}

\newcommand{\supercite}[1]{\textsuperscript{\cite{#1}}}

\newcommand{\ourModel}{PhyE2E\xspace}

\newcommand{\ourModelDCM}{PhyE2E(D\&C+MCTS)\xspace}
\newtheorem{definition}{Definition}
\newtheorem{theorem}{Theorem}
\newtheorem{lemma}[theorem]{Lemma}

\newcommand{\appendixref}[1]{%
    \ifthenelse{\equal{#1}{fig:Divide-and-Conquer-performance}}{S1}{%
    \ifthenelse{\equal{#1}{tab:feynman_dataset}}{S1}{%
    \ifthenelse{\equal{#1}{tab:constants_phy1}}{S2}{%
    \ifthenelse{\equal{#1}{tab:constants_phy2}}{S3}{%
    \ifthenelse{\equal{#1}{tab:constants_phy3}}{S4}{%
    \ifthenelse{\equal{#1}{tab:constants_phy4}}{S5}{%
    \ifthenelse{\equal{#1}{tab:constants_phy5}}{S6}{%
    \ifthenelse{\equal{#1}{tab:formulas_phy1_otherYear}}{S7}{%
    \ifthenelse{\equal{#1}{tab:constants_phy1_otherYear}}{S8}{%
    \ifthenelse{\equal{#1}{tab:phy1_baseline}}{S9}{%
    \ifthenelse{\equal{#1}{tab:phy1_consts}}  {S10}{%
    \ifthenelse{\equal{#1}{tab:phy2_baseline}}{S11}{%
    \ifthenelse{\equal{#1}{tab:phy2_consts}}  {S12}{%
    \ifthenelse{\equal{#1}{tab:phy3_baseline}}{S13}{%
    \ifthenelse{\equal{#1}{tab:phy3_consts}}  {S14}{%
    \ifthenelse{\equal{#1}{tab:phy4_baseline}}{S15}{%
    \ifthenelse{\equal{#1}{tab:phy4_consts}}  {S16}{%
    \ifthenelse{\equal{#1}{tab:phy5_baseline}}{S17}{%
    \ifthenelse{\equal{#1}{tab:phy5_consts}}  {S18}{%
    \ifthenelse{\equal{#1}{sec:baselines}}{3}{%
    \ifthenelse{\equal{#1}{sec:proof-decomposition}}{5.1}{%
    \ifthenelse{\equal{#1}{sec:proof-aggregation}}  {5.2}{%
    \ifthenelse{\equal{#1}{sec:algorism}}{4}{%
    \ifthenelse{\equal{#1}{algorism1}}{1}{%
    ???}}}}}}}}}}}}}}}}}}}}}}}}}

\newcommand{\defeq}{\overset{\text{\tiny def}}{=}}
\begin{document}

\title{A Neural Symbolic Model for Space Physics}

\author[1,$\dag$]{Jie Ying}
\author[2,$\dag$]{Haowei Lin}
\author[3,$\dag$]{Chao Yue}
\author[4]{Yajie Chen}
\author[5]{Chao Xiao}
\author[5]{Quanqi Shi}
\author[2]{Yitao Liang}
\author[6,7,*]{Shing-Tung Yau}
\author[6,7,8,*]{Yuan Zhou}
\author[9,10,*]{Jianzhu Ma}

\affil[1]{Qiuzhen College, Tsinghua University, Beijing, China.}
\affil[2]{Institute for Artificial Intelligence, Peking University, Beijing, China.}
\affil[3]{School of Earth and Space Sciences, Peking University, Beijing, China}
\affil[4]{Max-Planck Institute for Solar System Research, Göttingen, Germany.}
\affil[5]{Institute of Space Sciences, Shandong University, Weihai, China}
\affil[6]{Yau Mathematical Sciences Center, Tsinghua University, Beijing, China.}
\affil[7]{Beijing Institute of Mathematical Sciences and Applications, Beijing, China.}
\affil[8]{Department of Mathematical Sciences, Tsinghua University, Beijing, China.}
\affil[9]{Department of Electronic Engineering, Tsinghua University, Beijing, China.}
\affil[10]{Institute for AI Industry Research, Tsinghua University, Beijing, China.}
\affil[$\dag$]{Equal contribution.}
\affil[*]{Correspondence should be addressed to: \href{mailto:styau@tsinghua.edu.cn}{styau@tsinghua.edu.cn}, \href{mailto:yuan-zhou@tsinghua.edu.cn}{yuan-zhou@tsinghua.edu.cn}, \href{mailto:majianzhu@tsinghua.edu.cn}{majianzhu@tsinghua.edu.cn}.}

\begin{abstract}
Symbolic regression, a key problem in discovering physics formulas from observational data, faces persistent challenges in scalability and interpretability. We introduce \ourModel, an AI framework designed to discover physically meaningful symbolic expressions. \ourModel decomposes the symbolic regression problem into sub-problems via second-order neural network derivatives, and employs a transformer architecture to translate data into symbolic formulas in an end-to-end manner. The generated expressions are further refined via Monte-Carlo Tree Search and Genetic Programming. We leverage a large language model to synthesize extensive expressions resembling real physics, and train the model to recover these formulas directly from data. Comprehensive evaluations demonstrate that \ourModel outperforms existing state-of-the-art approaches, delivering superior symbolic accuracy, fitting precision, and unit consistency. We deployed \ourModel to five critical applications in space physics. The AI-derived formulas exhibit excellent agreement with empirical data from satellites and astronomical telescopes. We improved NASA's 1993 formula for solar activity and provided an explicit symbolic explanation of the long-term solar cycle. We also found that the decay of near-Earth plasma pressure is proportional to $r^2$ to Earth, with subsequent mathematical derivations validated by independent satellite observations. Furthermore, we found symbolic formulas relating solar EUV emission lines to temperature, electron density and magnetic field variations. The formulas obtained are consistent with properties previously hypothesized by physicists.
\end{abstract}

\maketitle

\section{Introduction}\label{sec3}
The primary distinction between physics and other data sciences is the pursuit of the discovery of fundamental laws behind the world using concise symbolic formulas.  
The discovery of physics formulas is a lengthy process of trial and error.
For instance, after about 40 attempts to match Mars data with various elliptical shapes, Kepler discovered that Mars' orbit was elliptical and proposed the three laws of Kepler\supercite{cranmer2023interpretablePHD}. 
Similarly, through meticulous experimentation of electric and magnetic phenomena, Faraday unveiled the laws of electromagnetic induction. He demonstrated that the intricate relationship between electricity and magnetism could be articulated through elegant, fundamental principles derived from empirical data\supercite{pearce1963faraday}.

Discovering laws in physics or other specialized fields often demands years of domain expertise. Numerous methods have been specifically developed and proposed to aid in the discovery of physics formulas and to enhance the understanding of the underlying principles\supercite{brunton2016discovering, de2024ai, ahmadi2024ai}.
Since AI has achieved tremendous success in multiple domains, a naturally arising question is whether we can leverage AI to automatically extract physical laws from experimental observation data to better understand our physical world. 
This task is known as symbolic regression (SR) within the AI field and has received widespread attention in recent years.
Genetic algorithms (GAs) were first adopted to address SR problems by evolving a population of candidate symbolic expressions (typically represented as trees) to minimize a fitness function, which measures how well the expression fits the data\supercite{schmidt2010age, schmidt2009distilling, la2019probabilistic, lalearning, virgolin2021improving, mccormick2019gplearn, de2021interaction, arnaldo2014multiple, kommenda2020parameter, virgolin2019linear}. 
Monte Carlo Tree Search (MCTS) predicts the expression by exploring a search tree of possible expressions, simulating paths down the tree by randomly selecting mathematical operations to extend expressions. For each path, it generates complete expressions and assesses fitness based on how well they approximate the target function\supercite{kamienny2023deep, lu2021incorporating, xu2023rsrm, xie2024efficient, sun2022symbolic}. 
Compared to GAs, MCTS offers more dynamic and fine-grained control over reward signals. Rewards can be customized to reflect the quality of partial solutions during tree exploration, facilitating a more efficient search for the optimal solution.
Recent advances in deep neural networks have paved the way for the development of end-to-end methodologies.
Unlike Genetic Algorithms (GAs) and Monte Carlo Tree Search (MCTS), which are symbolic search techniques that frequently struggle with the expansive search space, the end-to-end approach streamlines the process by eliminating the need for iterative searching and refinement.
It predicts mathematical expressions as sequences of symbols (operators, variables, constants) from data, allowing symbolic expressions to be generated in a single forward pass through the neural network, which significantly improves speed, especially for large datasets and complex functions\supercite{valipour2021symbolicgpt, chen2024bootstrapping, xing2021automated, kamienny2022end, biggio2021neural, shojaee2023transformer}. 
End-to-end approaches require large datasets for training, yet the number of physical formulas discovered by humanity remains relatively limited. There are still many unknown areas in the physical world waiting for humanity to discover and formulate new laws and equations.
Therefore, current data-driven end-to-end algorithms are limited in reasoning about simple mathematical equations, with only a few preliminary works successfully applied to real physical data\supercite{makke2025inferring, 10.1093/pnasnexus/pgae467}.
Additionally, symbols in physical law have physical unit systems, and the entire expression should ensure the correctness of these units, which is not yet considered by current computational methods. 

To address these limitations, we propose a new framework, named \ourModel, 
to achieve accurate symbolic regression for space physics. The \ourModel framework includes the following key components.
First, we fine-tuned a large language model (LLM) with existing physics formulas, enabling it to generate a diverse array of formulas that align with the statistical distribution of physics formulas. 
By harnessing the common knowledge acquired by LLMs from the internet, the generative model efficiently learns the underlying distribution from a small set of seed formulas, thereby overcoming the challenge of data sparsity in the training process.
Second, we trained an end-to-end formula regression model based on the transformer model, which directly converts data matrices into symbolic formulas encoded in Polish representation\supercite{shojaee2023transformer, kamienny2022end, biggio2021neural}. With the aid of the LLM, our end-to-end model is trained with a large volume of formulas that ``look physical'' and adhere to consistent unit systems. 
Third, to reduce the search complexity, we developed a formula-splitting technique that is able to group variables without nonlinear (or logarithmically nonlinear) relationships, producing a series of simpler sub-formulas for a nested and more simplified symbolic regression task. 
This technique uses an oracle neural network to fit the data and then analyzes its second-order derivatives to uncover relationships among the input variables. 
Finally, we leveraged the state-of-the-art GAs and MCTS methods to further refine the predicted formulas.
To evaluate our performance, we conducted comparisons on both an LLM-synthesized dataset and another real-world physics dataset AI-Feynman against GA-based, Transformer-based, and NN-based methods. Experimental results indicate that our method outperforms all others in terms of data fitting accuracy, the correctness of the mathematical form of the formulas, and unit accuracy.
We further applied our model to various important applications in space physics, including predicting sunspot numbers, plasma sheet pressure, solar rotational angular velocity, emission line contribution functions from the Sun, and lunar tidal signals in the plasma layer.
Compared to formulas proposed by physicists, the formulas derived from \ourModel exhibit better generalizability, a more concise mathematical form, more precise physical units, and more importantly, provide physically meaningful insights and explanations for instrumental observations.

\section{Results}\label{sec4}

\paragraph{Overview of \ourModel.} 
\ourModel comprises an end-to-end transformer-based physical model designed to take observed data points as input and predict both the operators and the physical units involved in a formula directly (Fig.~\ref{fig1}). 
A set of 264,000 synthetic formulas were generated by an LLM (OpenLLaMA2-3B) that had been fine-tuned using real physics formulas from the Feynman Symbolic Regression Database (FSReD)\supercite{udrescu2020ai} (Fig.~\ref{fig1} top). 
We randomly selected 180,000 formulas from the synthetic dataset for training, reserving the remaining data for testing.
The inference process of \ourModel involves three stages: a variable interaction detection method to decompose the target formula into sub-formulas (Fig.~\ref{fig1} bottom left), end-to-end prediction of each sub-formula using the trained model (Fig.~\ref{fig1} middle), and GA and MCTS approaches to refine the predicted formulas (Fig.~\ref{fig1} bottom right). 
More specifically, we first fit the data using a standard multi-layer neural network, and used its Hessian matrix to determine the nonlinear interactions between all pairs of variables (Sec.~\ref{methods:D&CDivide} in Methods).
The core of our model leverages a transformer architecture, synthesizing formulas as a prefix sequence with attached physical units (Sec.~\ref{methods:PhyE2E} in Methods). This approach incorporates both observed data points and physical prior knowledge about the target formula. 
In the final stage, we utilized GA and MCTS approaches to minimize the root-mean-square error (RMSE) of the predicted formula (Sec.~\ref{sec:MCTS-GP-refinement} in Methods). This was achieved by using a grammar pool of context-free rules that incorporates basic operators like exp, as well as the most promising sub-formulas, which were automatically constructed within the search tree by \ourModel.

\paragraph{Performance on the synthetic and AI Feynman datasets.}  

We divided the synthetic dataset containing 264,000 formulas into training, validation, and test sets with a ratio of 80\%, 10\%, and 10\%.
We ensure that all validation and test formulas are unseen during training by removing formulas that are either identical, or become completely equivalent to the formulas in the training set after simple mathematical transformations (Sec.~\ref{methods:dataDividing} in Methods). 
We compared \ourModel with 15 state-of-the-art symbolic regression baseline models, including 4 GA-based models (PySR\supercite{cranmer2023interpretable}, GP-GOMEA\supercite{virgolin2021improving}, Operon\supercite{10.1145/3377929.3398099}, and GPLearn\supercite{mccormick2019gplearn}), 3 Transformer-based models (TPSR\supercite{shojaee2023transformer}, EndToEnd\supercite{kamienny2022end}, NeSymReS\supercite{biggio2021neural}), 2 LLM-based models (LaSR\supercite{grayeli2024symbolic}, LLM-SR\supercite{shojaee2024llm}), and 6 NN-based models (uDSR\supercite{landajuela2022unified}, PhySO\supercite{tenachi2023deep}, AIFeynman\supercite{udrescu2020ai, udrescu2020ai2}, ParFam\supercite{scholl2023parfam}, KAN\supercite{liu2024kan, liu2024kan_v2}, BSR\supercite{jin2019bayesian}).
The technical details of running these models are provided in Supplementary Notes \appendixref{sec:baselines}.
The detailed comparisons with PySR under different configurations are provided in Supplementary Notes \appendixref{sec:comparison_pysr}.
We also included two variants of \ourModel in our comparison, including versions without using the formula decomposition module (D\&C) and the MCTS refinement module (MCTS).
We evaluated the performance of all the models on 6 metrics\supercite{la2021contemporary} including symbolic accuracy, average accuracy($R^2>0.99$), unit accuracy, complexity, relative complexity to the ground truth formulas and elapsed times (Sec.~\ref{methods:metrics} in Methods).

First, we evaluated whether the physics formulas synthesized by the LLM were consistent with the real physics formulas from the Feynman Symbolic Regression Database (FSReD).
It can be observed that the formulas generated by the LLM closely match the distribution of real physics formulas in terms of the number of variables, formula complexity, depth, and operator types measured by the Jensen-Shannon divergence ($D_{\mathrm{JS}}$) (Sec.~\ref{methods:metrics} in Methods, Fig.~\ref{fig2}a). 
Then, we studied the symbolic accuracy of the formulas generated by the model, that is, whether the mathematical forms of the formulas correspond to the true formulas used to generate the data. 
\ourModelDCM exceeds the second-best model PySR by 26.48\%, outperforms the best NN-based model, uDSR, by 31.75\%, the best LLM-based model, LaSR, by 37.89\%, and the best Transformer-based model, TPSR, by 39.83\% (Fig. \ref{fig2}b).
The generated formulas by \ourModel also demonstrate a more powerful ability to fit data compared to other methods.
\ourModelDCM outperforms all the state-of-the-art approaches by at least 20.00\% in terms of Avg.Acc.($R^2>0.99$).
Regarding the accuracy of (physical) units, we observe that \ourModel achieves 99.27\% of accuracy, leading all the other approaches. The performance drops to 93.30\% by including the D\&C and MCTS modules. 
This decline is due to the absence of unit constraints during the D\&C and MCTS refinement stage.
The most compatible baseline of unit accuracy is PhySO, which reaches a comparable 89.70\% with a strictly unit constraint during its search process\supercite{tenachi2023deep}.

The data-fitting capabilities of the formulas can be enhanced by increasing their complexity.
According to Occam's Razor, complex formulas tend to have weaker generalizations compared to simpler ones and lose their interpretability in a physical context. Therefore, in addition to studying the formulas' ability to fit data, we also focus on the complexity of the formulas generated by different models.
Our model produces formulas with lower model complexity compared to the formulas generated by other models.
For 42.17\% of the test formulas, their depth is less than 3, suggesting that they predominantly represent linear relationships. 
For these low complexity formulas, our model successfully recovers the mathematical form of 98.02\% of them (Fig.~\ref{fig2}d). 
The relative complexity to the ground truth formula of \ourModel is 2, which is 33.27\% better than the second-best model PySR. By introducing the D\&C and MCTS modules, \ourModel increases the complexity of formulas by 6.79\%, while still maintaining a similar complexity with PySR.
Only a handful of other baseline models demonstrate the capability to predict formulas of appropriate complexity (e.g., PySR, GP-GOMEA, AIFeynman, PhySO), primarily because of their constraints on complexity or the inclusion of physical units.
Others either generate formulas with high complexity (complexity>50) or formulas deviate far away from the target formula (relative complexity > 10), making it difficult to interpret in practice (Fig.~\ref{fig2}b).

Next, we evaluated the performance of different models on the formulas collected from the Feynman Symbolic Regression Database (FSReD), referred to as the Feynman Dataset for brevity.
The Feynman dataset contains 100 real-world physics formulas sampled from the seminal Feynman Lectures on Physics\supercite{feynman1963a, feynman1963b, feynman1963c} covering core physics topics like classical mechanics, electromagnetism, and quantum mechanics.
Although our model was not directly trained on formulas from the Feynman dataset, our training dataset was generated by LLM based on formulas from the Feynman dataset.
To inhibit data leakage, we removed the formulas in the training dataset that were either identical, or became completely equivalent after simple mathematical transformations compared to those in the Feynman dataset. A comprehensive list of the test formulas is provided in Supplementary Table \appendixref{tab:feynman_dataset}.

First, we observe that all computational methods, including ours, show similar performance on the Feynman dataset and on our synthetic data, which demonstrates that the distribution of formulas generated by the LLM is essentially consistent with those from the Feynman dataset.
In terms of symbolic accuracy, \ourModelDCM exceeds the best classical model, PySR, by 10.09\%, and the best NN-based model, uDSR, by 18.49\%, the best Transformer-based model, TPSR, by 21.77\%, the best LLM-based model, LaSR, by 29.35\%. 
Although PySR outperforms standard \ourModel in numerical precision ($R^2>0.99$) with an accuracy of 84.96\%, it is still surpassed by \ourModel with the D\&C and MCTS modules.
The complexity of the models generated by \ourModel remains low. The relative complexity to the ground truth formula of \ourModel is 2.98, 3.67\% higher than the best model PySR.
The relative complexity for \ourModel is lower than 5.75 and the complexity is lower than 16.85, which is essentially the best among all the baseline models (Fig.~\ref{fig2}c).
To further investigate the relationship between performance and formula complexity, we calculated the symbolic accuracy as a function of the number of variables, complexity, and the number of unary and binary operations. 
We found that our method had a significant advantage over other methods for formulas with high complexity. When the complexity is larger than 20, our model outperformed the second-best methods 67.11\% and 32.73\% on the two datasets, respectively (Fig.~\ref{fig2}d).  
We divided both datasets into three subsets of varying difficulty based on the similarity of mathematical formulas compared to those in training datasets (0.95-1.0 as easy, 0.80-0.95 as medium, 0.00-0.80 as hard), and then systematically calculated the symbolic accuracy of each method. 
We found that our method had a significant advantage on the medium and hard datasets (Sec.~\ref{methods:dataDividing} in Methods, Fig.~\ref{fig2}d).
 
Among all the modules, the divide-and-conquer (D\&C) module plays a crucial role in simplifying the search space in our framework. 
Consider a formula $f=m\sqrt{B_1^2+B_2^2+B_3^2}$ from the Feynman dataset, our D\&C module first determined that the three variables $B_1$, $B_2$ and $B_3$ under the square root do not interact, which indicates that the operators between them can only be addition or subtraction (Supplementary Fig.~\appendixref{fig:Divide-and-Conquer-performance}a). 
Therefore, the original symbolic regression task decomposes into three parts $g_i=m\sqrt{B_i^2+C_i}$ ($i={1,2,3}$), each of which is processed individually by the end-to-end module. The predicted formulas for $g_1$, $g_2$ and $g_3$ are later aggregated into one complete formula (Supplementary Fig.~\appendixref{fig:Divide-and-Conquer-performance}a).
The D\&C module reduces the complexity of a formula by splitting a set of variables that have no nonlinear interactions locally within a certain formula. 
This also explains why we find that our method has a considerable advantage over other methods as the complexity of the formulas increases.

However, the risk associated with this methodology is that if the decomposition is incorrect, some of the segments will contain incorrect variables. Therefore, we carefully studied the performance of different D\&C strategies.
We implemented 4 different strategies, including Single Pattern (detects the interaction of additive and multiplicative), Multi-Pattern (identifies all interactions such as additive terms under sine functions), Fixed Threshold and Adaptive Threshold (Sec.~\ref{methods:D&CDivide} in Methods).
We evaluated the performance of different strategies by assessing how many formulas could be correctly decomposed, partially correctly decomposed, and completely incorrectly decomposed on the test set of the synthetic dataset.
We found that the Adaptive Threshold strategy provided a flexible approach for interaction detection, resulting in a substantial improvement in complete accuracy by 50.61\%.
The Multi-Pattern strategy facilitates diverse types of interactions and effectively reduces the proportion of absolutely wrong formulas from 6.50\% to 5.03\% (Supplementary Fig.~\appendixref{fig:Divide-and-Conquer-performance}b).
Overall, adopting the Multi-Pattern and Adaptive Threshold strategies results in the highest number of accurate formula decompositions, which was also the strategy we used in our framework. 

To discover physics formulas rather than purely mathematics formulas, another important technique is the consistent units of physics quantities\supercite{tenachi2023deep, udrescu2020ai2}. We further retrain two additional models using the same architecture but one without units decoding strategy and another without any physical priors, thereby excluding the unit decoding strategy as well (Sec.~\ref{methods:PhyE2E} in Methods). These three models are evaluated using three accuracy metrics on the synthetic dataset (Supplementary Fig.~\appendixref{fig:Divide-and-Conquer-performance}d). We found that the unit prior plays an important role especially when dealing with a small amount of input data. In an extreme case with only five input data points, the symbolic accuracy of PhyE2E improved to 39.83\%, compared to 26.06\% without units decoding strategy and 8.97\% without any physical priors. Another observation is that the units accuracy is notably enhanced by the incorporation of physical priors. The unit accuracy improved by 25.90\% compared to the PhyE2E without units decoding strategy, and by 56.70\% compared to the PhyE2E without any physical priors in the five-data-point case. This improvement is also observed in the case with 50 data points, where the unit accuracy increased by 4.47\% and 12.23\%, respectively.

\paragraph{Performance of sunspot number prediction.} 

Next, we applied our trained model to multiple applications in space physics.  
Our goal is to find formulas that are more accurate in prediction and simpler in mathematical form than existing physics formulas.
We directly applied our trained model to these real-world applications instead of performing any fine-tuning operations.
We started by studying the pattern of changes in the sunspot number (SSN) over time. The Sun serves as the primary source of energy for the entire Earth system.
Predicting sunspot numbers is essential for forecasting space weather events that can impact both satellite operations and terrestrial communication systems. 
These forecasts are also pivotal in climate studies, aiding in modeling the impact of solar variability on the Earth's climate. 
The accurate prediction of sunspot activity is also crucial for managing the effects of geomagnetic storms on the technological infrastructure, which can lead to significant disruptions and damage.
Although the 11-year solar cycle is determined through direct observations of sunspot numbers for the past four centuries, scientists want to know whether there is a longer cycle in solar activity which could influence the climate of Earth (Fig.~\ref{fig:SSN}a).
Therefore, in this task, in addition to predicting the sunspot number, we also focus on how to derive the long-term cyclical variations of solar activity from the predicted physics formulas.

We first collected the SSN data from the Sunspot Index and Long-term Solar Observations (SILSO) over the last 400 years\supercite{sidc}.
The most widely adopted formula is the one proposed by Hathaway et al.\supercite{hathaway1994shape}, which is still being used in recent studies to analyze data from the last 30 years\supercite{upton2023solar}.
This formula modeled the sunspot numbers $R(t)$ using different sets of parameters for different cycles (Fig.~\ref{fig:SSN}b).
Although it can accurately fit the data for each cycle, it cannot be generalized from one cycle to another, making it incapable of predicting future SSNs or revealing the longer cycle of solar activity. 
To summarize a symbolic formula for SSN, we selected the SSN data from year 1855 to 1976 which containing 11 cycles and 1,450 data points as input of our model. 
The main components of the denominator in our formula are a squared sine term and a squared cosine term, while the numerator consists of a squared sine term (Fig.~\ref{fig:SSN}c top).
The Pearson correlation between the predicted SSN and the measured ones for the next four complete cycles (year 1976 to 2019) reaches 0.72. 
For the upcoming cycle (2019-present), \ourModel predicts the peak value to be 177.40 and occurs on October 10, 2024 (Fig.~\ref{fig:SSN}c top).
It is worth noting that no data were used to fine-tune or retrain our model. 
We did not use all the data to generate the formula because we needed to examine the generalizability of the formula obtained by the model. The generalizability of the model and the generalizability of the formulas predicted by our model should be evaluated separately.
We tried generating formulas with different amounts of data, and the resulting formula forms remained largely consistent (Supplementary Tables \appendixref{tab:formulas_phy1_otherYear}, \appendixref{tab:constants_phy1_otherYear}).

We compared the formula generated by \ourModel with those generated by other SR models. 
Among all the models, only our model can be directly applied to the data for formula inference without any retraining or fine-tuning. 
Therefore, we retrained all the other models to be compared on data from year 1855 to 1976 and tested their performance on data after year 1976 to fairly compare with our model. 
We first examined the formulas generated by different SR models. We observed that, except for AIF and GP-GOMEA, all other methods produced formulas containing trigonometric functions capable of generating periodic outputs. Formulas from BSR, EndToEnd, KAN, NeSymReS, PhySO, and TPSR all yield identical values for each cycle, which clearly contradicts our observations and common sense. 
The formulas generated by GPLearn, Operon, ParFam, uDSR, and our method are capable of generating periodic outputs with different values for each cycle (Fig.~\ref{fig:SSN}c bottom). 
We further examined the formulas generated by PySR under different operator sets and different constraints variants (Supplementary Tables \appendixref{tab:phy1_pysr_baseline}, \appendixref{tab:phy1_pysr_consts}), and take the best PySR model with the highest Avg-R score into comparison.
However, the five other models generated excessively complex formulas, making it impractical to parse their physical meaning or ascertain the long-term cyclical patterns of solar activity (Supplementary Tables \appendixref{tab:phy1_baseline}, \appendixref{tab:phy1_consts}).
Next, we quantified the performance of the formulas generated by these models on the test sets from year 1976 to 2019 using two metrics: 1) the correlation between the formula's predictions and the measured data within each individual cycle, then averages these correlations across multiple cycles (avg-R); 2) the correlation across multi-cycles (multi-R).
Our simpler formulas significantly outperformed these complex formulas for both metrics. Specifically, the avg-R and multi-R of our method outperform the second-best methods by at least 76.01\% and 58.57\%, respectively (Fig.~\ref{fig:SSN}d).

To further validate the accuracy of the formula we obtained, we focused on the SSN data before the 1700s. Due to technological limitations, there are no SSN data directly observed from telescopes prior to 1749\supercite{sidc}.
Therefore, we collected solar modulation levels reconstructed from atmospheric $^{14}C$ concentrations from the annual rings of thousand-year-old trees as an approximation of the SSN measurements\supercite{brehm2021eleven}. 
We adopted the same smoothing strategy as reported in Brehm et al., 2021 and found that the Pearson correlation between solar modulation from tree rings and the SSN measured by the telescopes is 0.886 after the 1700s, which verifies their close relationship (Fig. \ref{fig:SSN}f).
Then, we compared the SSN generated by our formula and solar modulation data before the 1700s. 
Their Pearson correlation is 0.561 before 1300, 0.653 during 1300 to 1700 and 0.501 after 1700 (Fig. \ref{fig:SSN}e), which further demonstrates the predictive capacity of our formulas on previously unseen data.
Lastly, since there are three sine/cosine functions in our formula, it is natural for us to derive three cycles from these trigonometric functions. The shortest cycle is 10.91 years, which aligns with observations of solar activity widely accepted by the research community. The second longest cycle is 59.27 years, which coincides exactly with the 60-year cycle of the ancient Chinese astronomical calendar system of Heavenly Stems and Earthly Branches. The longest cycle is 204.93 years, which we speculate is the cycle of the solar system operating within a larger planetary system (Fig. \ref{fig:SSN}e). 
Although this result requires further confirmation through additional astronomical observational data, it is the first conjecture directly derived from symbolic formulas. The constants of our formula are in Supplementary Table \appendixref{tab:constants_phy1}.

\paragraph{Performance of plasma pressure prediction.} 

Plasma pressure is a macroscopic parameter that plays an important role in plasma dynamics and the generation of electric currents. 
Increasing plasma pressure gradients in the radial direction causes the stretching of magnetic field lines and enhances perpendicular currents flowing azimuthally. The azimuthal plasma pressure gradient generates field-aligned currents, resulting in the bending of magnetic field lines (Fig.~\ref{fig:plasma_rotation}a).
Wang et al.~proposed a formula that describes how equatorial plasma pressure varies with its position relative to Earth using Geotail and the NASA mission Time History of Events and Macroscale Interactions During Substorms (THEMIS) data. 
However, their formula involves exponential terms and 9 constants for night-side equatorial isotropic plasma pressure, making it difficult  to derive meaningful physical interpretations \supercite{wang2013empirical,yue2013empirical}(Fig.~\ref{fig:plasma_rotation}b).
To derive a simpler formula, we divided the same equatorial plasma pressure data according to the azimuthal angle, using the data from the near-side of the Earth as input for our model and the data from the far-side to assess the performance of our formula.
As we increased the number of input data points, feeding the \ourModel model with progressively farther data from Earth, the average mean square error (MSE) decreased rapidly (Fig.~\ref{fig:plasma_rotation}c, Supplementary Table \appendixref{tab:phy2_baseline}, \appendixref{tab:phy2_consts}). 
\ourModel achieves an average MSE of $7.04\times10^{-3}$ with only 10\% of the data provided, demonstrating strong generalization capabilities with the small dataset as input. 
As more data was provided, the accuracy of our model continued to improve, eventually reaching $6.63\times10^{-4}$, resulting in more precise predictions for the far side of Earth regions (Fig.~\ref{fig:plasma_rotation}e).
\ourModel also outperforms all other baseline models in terms of fitting accuracy (MSE) and model complexity, delivering the most accurate and simplest prediction formula, while other models are unable to provide accurate predictions in certain areas.
The formula by Wang et al.~cannot accurately predict the plasma pressure for the far side of Earth regions, while the EndToEnd method fails to provide accurate predictions for the near side of Earth regions (Figs.~\ref{fig:plasma_rotation}d).
Note that this problem involves two variables, so there are two possible scenarios: one where $r$ and $\theta$ can be decomposed into two sub-formulas, and another where decomposition is not possible. We generated one formula for each scenario, compared their MSE on the training dataset, and selected the formula with the lower MSE as the final prediction. In this case, the formula derived using the decomposition method outperformed the alternative.
The formula we predicted has a more concise mathematical form compared to the formula proposed by Wang et al.~(Fig.~\ref{fig:plasma_rotation}b). 
It reveals to us that the decay of the near-Earth plasma pressure is proportional to the square of the distance $r$ to the Earth's center, whereas in the formula proposed by Wang et al., the plasma pressure has an exponential relationship with $r$. 
More importantly, we can derive certain physical facts that align with the observational data from this new formula.
Specifically, if the plasma pressure decays with the square of $r$ and it is also known that the magnetic pressure decays with the sixth power of $r$. Then, according to the formula $\text{plasma beta} = \text{magnetic pressure (Pth)} / \text{plasma pressure (Pb)}$,
we can infer that plasma beta increases with the fourth power of $r$, which can be confirmed by observational data from another independent study\supercite{lui1992radial}. 
Among the formulas obtained by the other methods, only the EndToEnd approach\supercite{kamienny2022end} produced a formula that is inversely proportional to the square of $r$. However, its mathematical form is more complex compared to our formula (Supplementary Table \appendixref{tab:phy2_baseline}).
The constants of our formula are in the Supplementary Table \appendixref{tab:constants_phy2}.

\paragraph{Performance of solar rotational angular velocity prediction.} 

The Sun's magnetic field is generated by the plasma motion within its interior.
Angular velocity of solar rotation varies at different latitudes, and the magnetic field lines are stretched and twisted (Fig.~\ref{fig:plasma_rotation}f). 
Differential rotation is a significant factor in the solar cycle for the prediction of magnetic fields and sunspots. 
Understanding solar differential rotation helps to improve the prediction of solar activities, which is important for predicting and mitigating the impact of space weather events on satellites and human activities in space. 
Differential rotation also provides key insights into the structure and dynamics of the solar interior by comparing rotation speeds at different latitudes to those predicted by comprehensive numerical models\supercite{Hotta2021, Vasil2024}. 
One of the most widely adopted formulas decribing the relationship between the solar differential rotation and solar latitude was derived by Snodgrass et al.\supercite{snodgrass1983magnetic} 
In this work, they assumed that solar differential rotation was symmetric with respect to the equator. Such an assumption often fails especially during periods of high levels of solar activity. 
In addition, this formula was fitted by using the measurements at low latitudes of the Sun, but observations at high latitudes are still missing, which limits the suitability of the model near the solar poles. 
For this task, the challenge for other AI models is that the limited amount of data makes it impossible for them to train the model and predict the formulas. 

\ourModel derived a simpler and more accurate formula with a simple trigonometric term using the data from Snodgrass et al.~in Magnetic Rotation of the Solar Photosphere\supercite{snodgrass1983magnetic}.
The largest difference between this formula and the one generated by \ourModel lies in their difference in the trigonometric periodicity, leading to a steeper prediction for polar regions, rather than a flat one (Fig.~\ref{fig:plasma_rotation}g).
We further reduced the number of training data points and found that \ourModel could predict the same formula with only 14 data points as input, rapidly achieving an MSE of $1.31\times 10^{-4}$, which outperforms other models (Fig.~\ref{fig:plasma_rotation}g, Supplementary Table \appendixref{tab:phy3_baseline}, \appendixref{tab:phy3_consts}).
In contrast to the oscillations seen in other models, \ourModel exhibits exceptional consistency and robustness, providing stable predictions with varying amounts of input data. 
The constants of our model can be found at the Supplementary Table \appendixref{tab:constants_phy3}.
The predictions for all the baseline models are quite similar in non-polar regions, but they start to diverge in the polar regions (Fig.~\ref{fig:plasma_rotation}h).
Regarding predictions at high latitudes, Hotta et al.\supercite{Hotta2021} overcame the ``convective conundrum'' through the supercomputer Fugaku and successfully reproduced solar-like differential rotation.
We selected the simulation data from the north and south polar regions and compared them with the baseline models. 
The formula derived from \ourModel performed the best in both polar regions, with correlations of 0.9814 and 0.9740, outperforming all  baseline models including the formula proposed by Snodgrass et al\supercite{snodgrass1983magnetic}.
In addition, we applied our model to different heights in the solar atmosphere, using data from various spectral lines\supercite{rao2024height} in the photosphere (Si I and Fe I) and the chromosphere (H$\alpha$) (Fig.~\ref{fig:plasma_rotation}i).
Similar formulas are derived across all the spectral lines with remarkable consistency, suggesting that the differential rotation speed within the Sun follows a regular and predictable pattern(Fig.~\ref{fig:plasma_rotation}j).

\paragraph{Performance of contribution function of emission lines.}

Emission lines in the extreme ultraviolet spectrum of the Sun, such as Fe X lines, are often used to observe the solar corona (Hinode/EIS, Solar Orbiter/EUI) (Fig.~\ref{fig:contributionFunction_lunarTide}a). 
Predicting the contribution functions of these lines helps in plasma diagnostics such as temperatures, densities, and magnetic fields, which is essential to understand solar phenomena such as flares and coronal mass ejections. 
The EUV emission lines can also be formed in other types of astrophysical targets, including stellar coronae, galactic nuclei, and supernovae. 
Understanding the formation of the emission lines can provide insight into the physical conditions and processes of these targets.
Given its role as a critical component of fundamental atomic databases applicable to a wide range of studies, considerable efforts have been made to calculate the contribution functions of emission lines (CHIANTI\supercite{1997A&AS..125..149D, 2024ApJ97471D} and NIST atomic database\supercite{kramida2023}).
The contribution functions could be approximated by solving complex quantum mechanical equations involving detailed calculations of electron transitions, collisions, and radiative transfer, which is usually a computationally expensive process.
Therefore, a challenge in physics is whether we can accurately estimate the contribution function using easily observable data, including the temperature and electron density around the Sun.
Currently, there is no physics formula that accurately reveals how the temperature and electron density around the Sun influence the contribution levels of the emitted spectral lines.

To address this problem, we downloaded the Fe X 174 and Fe X 175 line data from the CHIANTI database\supercite{1997A&AS..125..149D, 2024ApJ97471D}, and uniformly sampled 2,500 data points in an electron density range of $10^8-10^{10}$ and a temperature range of $5\times10^5-5\times10^6$°C.
Data were segmented into low and high temperature regions using a cutoff of $2.8\times 10^{6}$°C. 
Instead of fine-tuning or retraining our model, we took the data from the low-temperature region as input to generate the formula and test the prediction performance of the formula in the high-temperature region. 
In the high-temperature region, \ourModel achieved a significantly lower MSE, with the magnitude of the MSE being two orders of magnitude smaller than the other baseline models (Fig.~\ref{fig:contributionFunction_lunarTide}b left).
Both the Fe X 174 and 175 emission lines are highly temperature-dependent, with a relatively smaller influence from the electron density. 
To address this issue, we further investigated the ratio of these two lines, which serves as a powerful diagnostic tool for probing the effects of electron density on the intensity.
Compared to the prediction of the two spectral lines individually, the prediction of the ratio of the two lines carries more physical significance and is also more challenging. Accurate predictions of the individual spectral lines do not guarantee that their ratio can be predicted accurately (Figs.~\ref{fig:contributionFunction_lunarTide}b,c,d).
In this task, our formula achieves an MSE of $3.10\times 10^{-3}$, which is three orders of magnitude lower than the second best method EndToEnd (Fig.~\ref{fig:contributionFunction_lunarTide}b, middle).

Next, we examined the physical significance of the formulas generated by different methods. First, the complexity of our formula is not the lowest, but it is the only formula that has the correct physical units among all the baseline methods (Figs.~\ref{fig:contributionFunction_lunarTide}b right, Supplementary Table \appendixref{tab:phy4_baseline}, \appendixref{tab:phy4_consts}).
In the solar corona, the processes of ionization and recombination, as well as collision and excitation, can be considered decoupled\supercite{2004pscibookA}. 
In our formula, the electron density and temperature also exist in a decoupled form.
The dependence of the electron density following the mathematical form $(n+c_1)/(n+c_2)$ is widely accepted by the space physics community\supercite{mason1994spectroscopic}. 
It captures the behavior where the intensity increases with electron density at low values but saturates at higher densities because of collisional de-excitation.
The temperature term of our formula is composed of the combination of a power-law term and an exponential term. The power-law term dominates at low temperatures, capturing the increase in intensity as more electrons gain sufficient energy to excite the ions; the exponential decay term dominates at high temperatures, reflecting the rapid decrease in intensity due to ionization to higher states. This combination of a power-law term and an exponential term was also adopted by Raymond et al.\supercite{1977ApJS35419R} The constants of our formula are in the Supplementary Table \appendixref{tab:constants_phy4}.
For the formulas generated by other methods, some have excessive complexity, making them difficult to interpret, while others have incorrect physical units or are overly simplistic. For instance, the formula produced by PhySO is simple, but does not include the electron density term (Supplementary Table \appendixref{tab:phy4_baseline}).

\paragraph{Performance of lunar tide signal of plasma layer.}

The Earth's magnetospheric electric fields, including corotation and convection electric fields, are crucial for understanding the behavior of charged particles and maintaining the stability of the magnetosphere (Fig.~\ref{fig:contributionFunction_lunarTide}e). 
These fields are responsible for the movement and energization of charged particles, which in turn affect space weather and the interaction between the solar wind and Earth's magnetic field. 
Prior to 2023, it was commonly accepted in the scientific community that the electric fields at a specific near-Earth location were solely influenced by the Earth's position relative to the Sun and the distance to the Earth's center.
A recent work indicated that due to the effects of lunar tidal forces, the electric fields were also related to the Earth's position relative to the Moon \supercite{xiao2023evidence}.
However, due to the complexity of the problem, physicists have been unable to provide a physical formula that links these three important factors: electric fields with the distance to the Earth's center, the relative positions of the Earth and the Moon, and the Earth's position relative to the Sun, which can be represented as Lunar Local Time (LLT), Magnetic Local Time (MLT) and L-shell, respectively.

To address this problem, we collected $\sim$20,000,000 data measured by the Van Allen Probes satellites between L values of 3–6 from January 2013 to May 2019 from the RBSP/EFW official website
(\url{http://www.space.umn.edu/rbspefw-data/}), and used \ourModel to generate a formula to predict Radial Electric Field, denoted as $E_r$, from LLT, MLT and L-shell values.
Due to the large volume and high redundancy of data, we divided the entire three-dimensional space near the Earth into $50\times50\times50$ grids, and then calculated the average value of the Radial Electric Field within each grid. 
We randomly sampled 80\% of the grids as input for our model and also used them as training data for other baseline models. 
The remaining data were adopted as test data to evaluate the performance of all models.
Based on the adaptive threshold for the decomposition of the formula, we found that there are no coupling relationships among these three variables. Therefore, we decomposed the original symbolic regression problem into three sub-problems, each containing only one variable. Then, we predicted each uni-variate function using the end-to-end model (Sec.~\ref{methods:PhyE2E} in Methods).
Without fine-tuning or re-training of the model, the formula generated by our model has an MSE lower than the second-best method by 53.37\%, with complexity reduced by 75.91\% (Fig.~\ref{fig:contributionFunction_lunarTide}f). 
We examined the effects of LLT and MLT on the Radial Electric Field separately and found that our formula provides a good approximation for the original data. The prediction of our formula is much smoother than the measured data, due to the data smoothing applied within each grid (Fig.~\ref{fig:contributionFunction_lunarTide}g).
Compared to the formulas generated by other methods, the formula produced by our model captures multiple physical principles, making it more physically meaningful. First, among all the models, only our model provides an asymmetric prediction between the dayside and nightside of the Earth, suggesting that the radial electric field ($E_r$) on the dayside decays more rapidly in the radial direction (L-shell), while on the nightside, the radial electric field($E_r$) decays faster in the non-radial direction (MLT).
Second, since the radial electric field (positive direction towards Earth) is derived from the calculation of the electric field's y-component, it exhibits periodic variation with Magnetic Local Time (MLT), and the period is 12 hours\supercite{zhang2024estimating}, which is consistent with the periodicity of the cosine function of MLT in our formula ($12.13$ hours). 
Third, our formula indicates that $E_r$ decays with the square of the L-shell, which is consistent with theoretical calculations. According to the ideal magneto-hydrodynamic (MHD), the corotational electric field $E$ could be derived as $E=-\Omega_E B_0/L_{shell}^2$, where $\Omega_E$ and $B_0$ are Earth's rotational angular velocity and Earth's surface magnetic field, respectively. The constants of our model could be found in Supplementary Table \appendixref{tab:constants_phy5}.
The complexity of our formula is not the lowest among all the models because the pattern of this physical application is complicated (Fig.~\ref{fig:contributionFunction_lunarTide}b, Supplementary Table \appendixref{tab:phy5_baseline}, \appendixref{tab:phy5_consts}). 
Among methods with lower complexity, only the formulas from BSR and PySR exhibit periodicity in LLT or MLT, and only PySR captures the inverse relationship between $E_r$ and L-shell. However, the BSR formula does not include the crucial L-shell variable, and the PySR formula lacks the LLT variable (Supplementary Table \appendixref{tab:phy5_baseline}).

\section{Discussion}\label{sec5}

Existing symbolic regression research primarily employs search methods based on Monte Carlo Tree Search (MCTS) and Reinforcement Learning (RL). These methods often struggle to accurately predict formulas with a large number of variables or complex operational relationships between variables. To discover the correct formulas within a limited time, most of the MCTS approaches require prior knowledge to achieve an initialization close to the true solution. In contrast, our method can decompose formulas without knowing their specific forms, significantly reducing the complexity of symbolic regression. For each decomposed sub-problem, we utilized an end-to-end approach, tokenizing the data and directly translating it into formula strings using a transformer. Among all SR methods, our approach offers a ready-to-use model and is the only one that does not require retraining or fine-tuning on physical data.

One characteristic of space physics is that there is no data noise on the planet scales or on the microscopic particle scale. However, if the model is to be generalized to other areas of physics, such as condensed matter physics and fluid dynamics, the effects of noise inherent in the data must be considered. 
Since the data is free from noise, the model must learn to deduce the operational relationships between physical variables based solely on the data provided.
Therefore, we can still leverage large language models using this method of simulating the generation of physics formulas for data augmentation. This principle is similar to that of AlphaZero\supercite{silver2018general}, which does not require human game records but can learn to play Go through AI-versus-AI games. This is because AI only needs to learn optimal strategies in various complex situations, rather than necessarily mimicking human players' thought processes and habits.
Currently, our current model cannot handle operations such as integration and differentiation, which means that a significant portion of physics formulas based on partial differential equations cannot be resolved. We believe that the data augmentation and formula decomposition techniques are still applicable in partial differential equations.


\section{Methods}\label{sec:methods}

We now detail the components of our system, starting with the generative model for synthetic physics formulas. 
Next, we present the core framework, which includes an end-to-end symbolic regression model integrated with a Divide-and-Conquer (D\&C) strategy for decomposing complex formulas into simpler sub-formulas. 
We then describe the process to refine the formula predicted by the end-to-end model, encompassing Monte Carlo Tree Search (MCTS) and Genetic Programming (GP) refinement. Finally, we discuss the details of how to construct test data and evaluation metrics.

\subsection{Generative model for synthetic physics formulas}
\label{sec:methods-generative-model}

To generate synthetic formulas resembling real physics formulas, we fine-tuned the pre-trained OpenLLAMA-2-3B language model\supercite{touvron2023llama} using the AI Feynman dataset\supercite{udrescu2020ai}, which is a benchmark collection of mathematical formulas representing real-world physical laws and relationships, consisting of 100 formulas. 
A two-stage fine-tuning strategy was devised to address the challenge of limited training data and to integrate prior knowledge of physical unit systems into the training process. 
In the first stage, the OpenLLAMA-2-3B model was fine-tuned on the AI Feynman training set. The fine-tuned model was then employed to generate 50,000 synthetic formulas. These formulas were evaluated for consistency with physical unit systems, resulting in approximately 8,000 formulas that adhered to unit consistency. The second stage involved reassigning weights to the 8,000 unit-consistent formulas generated in the first stage. 
This weighting was designed to ensure that the statistical distribution of the synthetic formulas, such as the number of variables, formula depth, and operator frequency, aligned with those of the real physics formulas in the AI Feynman dataset. Finally, the language model was further fine-tuned using the weighted distribution of the 8,000 formulas. 

Specifically, the fine-tuning was performed using the DeepSpeed ZeRO Stage 2 optimizer within the HuggingFace Transformer framework\supercite{maurya2024deep}, with a learning rate of \(3 \times 10^{-5}\). The training prompt followed the format: ``\texttt{Generate a physics formula: \{formula\}}''. 
Mathematical formulas in the prompt were represented in plain text using their natural mathematical forms. The notation of variables for physical quantities in the formulas adhere to the standard conventions used in the Feynman Dataset.
To generate synthetic formulas, the model was prompted with the same instruction. 
The hyperparameters of the generative model were configured to balance diversity and quality in the generated formulas: the temperature was set to 2.0,  efficient sampling was enabled (\(\texttt{do\_sample = True}\)), and the maximum length of the generated sequences was restricted to $64$ tokens. 
To evaluate the consistency of synthetic formulas with physical unit systems (in the first stage), the formula’s expression tree\supercite{lample2019deep} was constructed, and the units of each subtree were computed in a bottom-up manner, starting from the leaves and moving toward the root. During this process, the following checks were performed: (1) for any sub-tree in the form of $A + B$ or $A - B$, the units of $A$ and $B$ were required to be the same, (2) for any sub-tree in the form of $\sin(A)$, $\cos(A)$, or $\exp(A)$, the unit of $A$ was required to be null. For instance, the formula “acceleration $+$ velocity $/$ time” is valid while the formula “acceleration $+$ velocity” or  “$\sin($acceleration $+$ velocity $/$ time$)$” is invalid due to unit mismatch. No constraints were imposed on operations such as multiplication, division, or square root, although these operations produce derived units that may affect the validity of their parent expressions. For example, “sqrt(acceleration $/$ time)” yields a unit equivalent to that of velocity, so the formula “sqrt(acceleration $/$ time) $+$ velocity” satisfies the unit consistency requirement for addition. In contrast, “sqrt(acceleration $/$ time) + velocity $/$ time” fails this criterion.
To reassign weights to a set of formulas (in the second stage), a linear program (LP) was formulated. The LP variables represent the weights assigned to each formula, subject to the constraints that the weights must be non-negative and sum up to $1$. The objective was to minimize the total variation distances between the statistical distributions (e.g., number of variables, formula depth, operator frequency) of the weighted synthetic formulas and those of the AI Feynman dataset. These total variation distances were expressed as linear combinations of the LP variables. Additionally, a regularization term was included in the LP objective to ensure that the weighted distribution did not deviate excessively from the original uniform distribution.

\subsection{The divide-and-conquer strategy}
\label{methods:D&CDivide}

Many physics formulas exhibit intrinsic simplicity and symmetry, with variables often interconnected through straightforward addition or multiplication. Inspired by this observation, we proposed a divide-and-conquer strategy to decompose the target formula into a summation (or multiplication) of simpler sub-formulas by estimating inter-variable relationships. 
To achieve this, we first train an oracle neural network to fit the data, then use the hessian matrix to identify the inner nonlinear relationship between variables. These relationships guide the decomposition strategy, breaking the target formula into several simpler sub-formulas, which are then predicted independently and subsequently combined to reconstruct the target formula.
For instance, consider the mathematical formula $f(x_1, x_2, x_3, x_4)=x_1x_2+x_3\log(x_4)$. It is straightforward to verify the second derivatives between the group $\{x_1, x_2\}$ and the group $\{x_3, x_4\}$ are zero, i.e., $\partial^2f / \partial x_i\partial x_j=0, \forall i\in\{1, 2\},\forall j\in\{3,4\}$, which mathematically indicates that the target formula can be decomposed into two sub-formulas $f_1(x_1, x_2)=x_1x_2$ and $f_2(x_3, x_4)=x_3\log(x_4)$. Our divide-and-conquer strategy is based on a generalization of the underlying mathematical principle of the above example.
Below, we first introduce the inner-variable relationships, followed by the decomposition strategy for sub-formulas and the resampling technique for data points. Finally, we present the aggregation theorem for the back aggregation step.

\subsubsection{The oracle neural network and estimation of inter-variable relationships} \label{sec:estimate-inter-variable-relationships}

Given the data $D = \{(\bm{x}_i, y_i)\}_{i=1}^{N}$ with $N$ evaluations of the target formula $f(\bm{x})$ ($N = 200$ in our experiments), an oracle neural network $\tilde{f}_\theta(\bm{x})$ was first trained to approximate $f(\bm{x})$ at any input point $\bm{x}$. 
Here, $\theta$ are the parameters of the oracle neural network, which consists of $5$ hidden layers. The first $3$ layers each contained $128$ tanh neurons, while the last $2$ layers each contained $64$ tanh neurons. 
The network was trained for $400$ epochs, with an initial learning rate of $0.1$ that decayed tenfold every $100$ epochs. 
The following definition characterizes the inter-variable relationship to be estimated for decomposing the target formula.
\begin{definition}\label{def:sigma-separable}
Let $\sigma : \mathbb{R} \to \mathbb{R}$ be a uni-variate operator. Two features $i, j \in \{1, 2, \dots, n\}$ of a target formula $f: \mathbb{R}^n \to \mathbb{R}$ are said to be \emph{$\sigma$-separable} if there exist sub-formulas $f_1$ and $f_2$ such that $f$ can be expressed as:
\[
f(\bm{x}) \equiv \sigma(f_1(\bm{x}_{-i}) + f_2(\bm{x}_{-j})),
\]
where $\bm{x}_{-i}$ is the $(n-1)$-dimensional vector obtained by removing $x_i$ from $\bm{x}$, and $\bm{x}_{-j}$ is defined analogously.
\end{definition}

The uni-variate operator $\sigma$ cannot be generalized to multi-variate operator, as our method relies on the invertibility of $\sigma$.
When the operator $\sigma$ is invertible, the following lemma, whose proof is provided in Supplementary Sec.\mbox{~\appendixref{sec:proof-decomposition}}, provides an equivalent condition to check whether two features are $\sigma$-separable in a twice-differentiable formula.

\begin{lemma}\label{lem:sigma-separable}
Let the uni-variate operator $\sigma : \mathbb{R} \to \mathbb{R}$ and the target formula $f: \mathbb{R}^n \to \mathbb{R}$ be twice differentiable. 
Suppose $\sigma$ is strictly monotonic, then two features $i, j \in \{1, 2, \dots, n\}$ are $\sigma$-separable if and only if for all $\bm{x} \in \mathbb{R}^n$,
\begin{align}
\frac{\partial^2 \sigma^{-1} \circ f}{\partial x_i \partial x_j}(\bm x) =  (\sigma^{-1})''(f(\bm{x})) \cdot \frac{\partial f}{\partial x_i}(\bm{x}) \cdot \frac{\partial f}{\partial x_j}(\bm{x}) + (\sigma^{-1})'(f(\bm{x})) \cdot \frac{\partial^2 f}{\partial x_i \partial x_j}(\bm{x}) = 0.
\end{align}
\end{lemma}

\paragraph{Practical implementation.} In our experiment, we tried the uni-variate operators $\sigma \in\{$id, sqrt, inv, $\arcsin$, $\arccos$, $\log$, $\text{sqrt}\circ\log$, $\text{inv}\circ\log$, $\arcsin\circ\log$, $\arccos\circ\log\}$ to perform the divide-and-conquer strategy and predict the target formula via the end-to-end model (Sec.~\ref{methods:PhyE2E}). Note that the operators that involve logarithm effectively separate the target formula into the multiplication of two sub-formulas according to Definition~\ref{def:sigma-separable}. The prediction based on different uni-variate operators were collected and fed into the MCTS and GP module for further refinement (Sec.~\ref{sec:MCTS-GP-refinement}).

However, even with access to the oracle neural network $\tilde{f}_\theta$, we cannot directly verify the condition in Lemma~\ref{lem:sigma-separable} because it required evaluating all the points in $\bm{x} \in \mathbb{R}^n$ and the second-order derivative of the approximate function $\tilde{f}_\theta(\bm{x})$ tends to be noisy. In our algorithm, this condition was verified approximately by sampling a subset of points and employing a majority rule to mitigate the effects of approximation noise. Specifically, for the set of training data points $\{\bm{x}_1, \bm{x}_2, \dots, \bm{x}_N\}$ ($N = 200$), two features $i$ and $j$ are determined to be $\sigma$-separable if, for a threshold parameter $\epsilon > 0$, it holds that
\begin{align}
J_{i,j}(\sigma, \tilde{f}_\theta) \defeq \mathop{\mathrm{median}}_{1\leq k \leq N} \left\{\left|\frac{\partial^2 \sigma^{-1} \circ \tilde{f}_\theta}{\partial x_i \partial x_j}(\bm x _k)\right|\right\} \leq \epsilon .
\end{align}

The choice of the threshold parameter $\epsilon$ is crucial to the accurate identification of $\sigma$-separable pairs. We first employ the \emph{fixed thresholding} strategy, treating $\epsilon=\epsilon(\sigma)$ as a constant for each uni-variate operator $\sigma$. This fixed constant is determined using a data-driven approach for each $\sigma$. More specifically, we randomly sampled 1,000 target formulas $h_k$ and trained the corresponding oracle neural networks $\tilde{h}_k(\bm x_k|\theta_k)$. For each feature pair $(i, j)$, we identified their $\sigma$-separability by Lemma~\ref{lem:sigma-separable} and calculated the corresponding $J_{i,j}(\sigma, \tilde{h}_k)$. The threshold was then chosen to maximize the number of feature pairs among the 1,000 additional target formulas that were classified correctly in terms of $\sigma$-separability. 

On the other hand, a single constant value for a fixed $\epsilon$ may not work well for all target formulas $f$, as different functions can exhibit vastly different scales of derivatives. This scale can even vary if $f$ is multiplied by a large constant factor. To address this issue, the technique of \emph{adaptive thresholding} was also proposed to automatically select the threshold based on the derivative scale, thereby improving the numerical stability of the estimation of $\sigma$-separable pairs. More specifically, let $\epsilon_0 = \min_{1\leq i<j\leq n} J_{i,j} (\sigma, \tilde{f}_\theta)$. We then define $\epsilon_1 = \alpha_{\min} \epsilon_0$ and $\epsilon_2 = \alpha_{\max} \epsilon_0$ as the minimum and maximum thresholds, respectively. In our algorithm, we set $\alpha_{\min} = 2$ and $\alpha_{\max} = 10$. For each $\epsilon \in [\alpha_{\min}, \alpha_{\max}]$, the set of $\sigma$-separable pairs (hereafter referred to as the \emph{$\sigma$-separable set}) was estimated as
\begin{align}
Q_\epsilon (\sigma, \tilde{f}_\theta) \defeq \left\{(i, j) ~\Big|~ 1 \leq i < j \leq n, J_{i,j}(\sigma, \tilde{f}_\theta) \leq \epsilon \right\} .
\end{align}
Finally, the class of $\sigma$-separable sets was defined as
\begin{align}
\mathcal{Q} \defeq \{Q_\epsilon(\sigma, \tilde{f}_\theta) ~|~ \epsilon \in [\alpha_{\min}, \alpha_{\max}] \}.
\end{align}
It is straightforward to verify that $|\mathcal{Q}| \leq n(n-1)/2$. Therefore, applying the divide-and-conquer strategy based on each $Q \in \mathcal{Q}$ is computationally feasible. Each application of this strategy derived a prediction of the target formula. These formulas were evaluated and the best one was selected. In the following subsections, we will explain how the divide-and-conquer strategy operates for any estimated $\sigma$-separable set $Q$: it begins by inducing a set of feature sets, followed by predicting the sub-formulas for these feature sets, and ultimately integrates them into the final prediction of the target formula.

\subsubsection{The division of the target formula}

We have discussed the $\sigma$-separable relationship between pairs of features. The following definition about the global division of the target formula and features is crucial to our divide-and-conquer strategy.

\begin{definition}\label{def:sigma-partition}
Let $\sigma : \mathbb{R} \to \mathbb{R}$ be a uni-variate operator. For any feature subset $A_i$ of $A = \{1, 2, \dots, n\}$, denote by $\bm x_{A_i}$ the vector obtained by restricting  $\bm{x}$ to $A_i$, i.e., if $A_i = \{j_1, j_2, \dots, j_k\}$, then $\bm{x}_{A_i} = (x_{j_1}, x_{j_2}, \dots, x_{j_k})$. A class of subsets $\{A_i\}_{i=1}^m$ is said to be a \emph{$\sigma$-division} of $f$ if none of the subsets is contained in another subset and there exist $m$ sub-formulas $f_1, f_2, \dots, f_m$ such that $f$ can be expressed as:
\begin{equation}\label{eq:def-sigma-partition}
f(\bm{x}) \equiv \sigma(f_1(\bm{x}_{A_1}) + f_2(\bm{x}_{A_2}) + ... + f_m(\bm{x}_{A_m})),
\end{equation}
where each $f_i$ is a function of $\bm x_{A_i}$.
\end{definition}

It is quite straightforward to derive a $\sigma$-division from $\sigma$-separable pairs, detailed as follows. First, the $\sigma$-separable relationships between all features are identified as described in Sec.~\ref{sec:estimate-inter-variable-relationships}. Then, the algorithm starts with an initial (and trivial) $\sigma$-division $\mathcal{A}=\{\{1, 2, ..., n\}\}$ and iteratively refines it. At each iteration, the algorithm selects a $\sigma$-separable pair of features $j_1$ and $j_2$ that has not been considered. For each feature subset $A \in \mathcal{A}$ such that $\{j_1, j_2\} \subseteq A$, $\mathcal{A}$ is updated as follows:
\[
\mathcal{A} \leftarrow (\mathcal{A} - \{A\}) \cup \{A-\{j_1\}, A-\{j_2\}\}.
\]
This process is done after all $\sigma$-separable feature pairs have been considered. By the following Lemma~\ref{lem:sigma-partition}, whose proof can be found in Supplementary Sec.~\appendixref{sec:proof-decomposition}, this iterative procedure guarantees to yield a valid $\sigma$-division.

\begin{lemma}\label{lem:sigma-partition}
Let uni-variate operator $\sigma : \mathbb{R} \to \mathbb{R}$ be strictly monotonic and $f: \mathbb{R}^n \to \mathbb{R}$ be the target formula that is twice differentiable. Suppose we have accurately obtained the set of all $\sigma$-separable feature pairs of $f$. Then, for each iteration number $\ell \geq 1$, the division obtained by the above algorithm after the $\ell$-th iteration, denoted by $\{A_k^\ell\}_{k=1}^{m_\ell}$, is a $\sigma$-division of $f$.
\end{lemma}

Algorithm~\appendixref{algorism1} is provided in Supplementary Sec.~\appendixref{sec:algorism} to formally describe the above procedure to derive $\sigma$-separable feature pairs via the fixed and adaptive thresholding techniques, as well as to construct a $\sigma$-division based on the $\sigma$-separable pairs.

\subsubsection{Evaluation of surrogate sub-formulas and back aggregation for the target formula}
\label{sec:aggregation}

Once a $\sigma$-division $\{A_i\}_{i=1}^{m}$ is derived, it would be most natural to recursively perform symbolic regression to predict the sub-formulas $\{f_i\}$ and reconstruct the target formula $f$ according to Eq.~\eqref{eq:def-sigma-partition}. However, there are two challenges to this approach. First, we do not have evaluation data ($\{(\bm{x}, y = f_i(\bm{x})\}$) for each sub-formula $f_i$. Moreover, the sub-formula $f_i$ are not even unique. For example, if $f(x_1, x_2, x_3) = \sigma(f_1(x_1, x_2) + f_2(x_1, x_3))$, then it can also be expressed as $f(x_1, x_2, x_3) = \sigma((f_1(x_1, x_2) - x_1) + (f_2(x_1, x_3) + x_1))$, where $f_1' = f_1 - x_1$ and $f_2' = f_2 + x_1$ are also a set of valid sub-formulas for the $\sigma$-division $\{\{x_1, x_2\}, \{x_1, x_3\}\}$.

To address the above challenges, we turn to predict the \emph{surrogate sub-formulas} $\{g_i\}_{i=1}^m$, defined as follows. First, an arbitrary $\bm{z} \in \mathbb{R}^n$ is chosen, where in the experiment, $\bm{z}$ was sampled from the standard multivariate Gaussian distribution. Then, given the $\sigma$-division $\{A_i\}_{i=1}^{m}$, for each $i \in \{1, 2, \dots, m\}$, we define $g_i : \mathbb{R}^{A_i} \to \mathbb{R}$ as 
\begin{equation} \label{eq:def-surrogate-sub-formula}
g_i(\bm{x}_{A_i}) = f(\bm{x}_{A_i}, \bm{x}_{\overline{A_i}} = \bm{z}_{\overline{A_i}}),
\end{equation}
where $\overline{A_i} = [n] - A_i$.

Each surrogate sub-formula $g_i$ might be quite different from the corresponding sub-formula $f_i$. 
For instance, consider $f(x_1, x_2, x_3) = f_1(x_1, x_3) + f_2(x_2, x_3)$ where $f_1(x_1, x_3) =  \sin(x_1x_3)$ and $f_2(x_2, x_3) = \cos(x_2x_3)$. According to the definition, we have  $g_1(x_1, x_3)=\sin(x_1x_3)+\cos(c_2x_3)$, and $g_2(x_2, x_3)=\cos(x_2x_3)+\sin(c_1x_3)$. Note that each $g_i$ not only contains $f_i$, but also introduces extra non-trivial terms. The following theorem, the proof of which can be found in Supplementary Sec.~\appendixref{sec:proof-aggregation}, provides a way to eliminate the extra terms and aggregate the surrogate sub-formulas $\{g_i\}_{i=1}^{m}$ to reconstruct the target formula $f$.

\begin{theorem}\label{thm:aggregation}
Suppose the uni-variate operator $\sigma$ : $\mathbb{R} \to \mathbb{R}$ is strictly monotonic. Let  $\mathcal{A} = \{A_i\}_{i=1}^{m}$ be a $\sigma$-division of the target formula $f$ and $\{g_i\}_{i=1}^m$ be defined as in Eq.~\eqref{eq:def-surrogate-sub-formula} based on a vector $\bm{z}$. For each $I \subseteq [m]$,  denote
\[
\mathcal{A}_I = \cap_{i \in I} A_i.
\]
Then, it holds that
\begin{equation}\label{eq:thm-aggregation}
f(\bm{x}) = \sigma\left(\sum_{\emptyset \neq I \subseteq [m]} \frac{(-1)^{|I| - 1}}{|I|}\sum_{i \in I} \sigma^{-1} \circ g_i(\bm{x}_{\mathcal{A}_{I}}, \bm{x}_{A_i - \mathcal{A}_{I}}=\bm{z}_{A_i - \mathcal{A}_{I}})\right).
\end{equation}
\end{theorem}

\paragraph{Practical implementation.} For the $i$-th surrogate sub-formula $g_i$, the data $D^{(i)} = \{\bm{x}_k^{(i)}, y_k^{(i)}\}_{k=1}^{N}$ were constructed as follows:
\[
\bm x_k^{(i)} = (\bm x_k)_{A_i}, \qquad\qquad 
y_k^{(i)} = \tilde{f}_\theta(\bm x_{A_i} = \bm x_k^{(i)}, \bm x_{\overline{A_i}} = \bm c_{\overline{A_i}}),
\]
where $\bm{x}_k$ is the $k$-th data point in the original data $D$. Then, each $g_i$ where predicted by the end-to-end model (as will be described in Sec.~\ref{methods:PhyE2E}). Finally, the target formula $f$ was predicted by aggregating the surrogate sub-formulas according to Eq.~\eqref{eq:thm-aggregation}.

\subsection{The end-to-end model}\label{methods:PhyE2E}

\paragraph{Architecture.} Kamienny et al.\supercite{kamienny2022end} established a transformer-based end-to-end model for symbolic regression. 
We design a similar transformer that includes 4 encoder layers and 16 decoder layers, forming an asymmetric architecture\supercite{charton2021linear}. 
Each layer consists of 16 attention heads and an embedding dimension of 512. 
We remove positional embeddings to ensure permutation invariance of the input data.  

Physical priors (e.g. physical units) play a crucial role in governing the structure and plausibility of physics formulas\supercite{tenachi2023deep, bendinelli2023controllable}. In PhyE2E, we integrate four types of physical priors, including the physical units of input and output variables, formula complexity, candidate operators and candidate constants. Besides the evaluations at $N$ data points, we also append a set of $h$ candidate physical priors to the input of our model. 

To encode the observed data points, symbolic formulas, and physical priors, we construct a vocabulary that includes tokens for float numbers, operators, variables, and physical units. Each float number is decomposed into three tokens representing its sign, mantissa (between 0 and 9999), and exponent (from E-100 to E+100)\supercite{charton2021linear}. For physical units, we adopt five commonly used base units: Meter (m), Second (s), Kilogram (kg), Kelvin (K), and Volt (V)\supercite{udrescu2020ai}, which are slightly different from the International System of Units (SI). In this way, a physical unit can be encoded into five tokens. For dimensionless constants and formulas, we simply assign each base unit with value 0. Each data point is encoded into $(n+1)\times 3$ tokens, where each of the $n$ features and the evaluation is encoded into $3$ tokens. Each of the $h$ physical priors is represented using specialized vocabulary tokens that encode the variables, physical units, operators, and float numbers. Specifically, we use $1+5$ tokens to represent a variable with its physical unit, and $3+5$ tokens to represent a constant with its physical unit. All data point and physical prior sequences are padded with the \texttt{<pad>} token to ensure a uniform fixed length $L=33$. They are finally concatenated together to form the input of our model.

Additionally, to help the model better understand a consistent physical unit system, we propose a novel unit decoding strategy for the model’s output. Specifically, we incorporate the physical units of each operator and variable within the formula into the output sequence. 
While formulas are represented as prefix expressions composed of operators and variables, we enhance this representation by associating each operator and variable with its corresponding physical unit, so that our model outputs a sequence of \texttt{(operator/variable, physical unit)} pairs. Notably, we do not impose explicit unit consistency constraints during generation. Instead, by explicitly inferring the physical units of the formula in a step-by-step manner within the output sequence, the model is guided to learn the underlying physical principles that govern variable relationships automatically. This approach enables consistent and meaningful unit inference directly from data.

\paragraph{Training details.}
The training data consisted of the data evaluated by 200,000 synthetic physics formulas generated in Sec.~\ref{sec:methods-generative-model}. Each formula was evaluated at $N = 200$ points sampled from the standard multivariate Gaussian distribution, conditioned on the fact that the formula is properly defined at the point, which follows the previous work\supercite{kamienny2022end}.

We used the cross-entropy loss on the next-token prediction with the Adam optimizer, starting with a learning rate of $10^{-7}$, gradually increased to $2\times10^{-4}$ during the first 10,000 steps, and subsequently decayed proportional to the inverse square root of the step count\supercite{vaswani2017attention}.
A validation set comprising 20,000 synthetic formulas is held out, and our model is trained for 100 epochs, processing 500 million formulas in total, until the validation accuracy stabilizes. 
Each batch consists of 10,000 tokens, with formulas grouped by similar lengths to minimize padding overhead.
The training was performed on a single NVIDIA A100 GPU with 80GB of memory over about one day. 

\paragraph{Constant optimization.}

The end-to-end model only returned a ``formula skeleton'' where the constant parameters in the formula remained unoptimized. We adopted the BFGS algorithm\supercite{fletcher1987} to further refine the constants in the formula skeleton. Specifically, given a formula $f$, we denote by $\bm{c}$ the constant parameters in the formula. Based on the data  $\{(\bm x_k, y_k)\}_{k=1}^{N}$, the BFGS algorithm was invoked to find
\[
\arg\min_{\bm c}\sum_{k=1}^{N} [f(\bm x_k; \bm c) - y_k]^2 .
\]

\subsection{MCTS and GP refinement}
\label{sec:MCTS-GP-refinement}

\paragraph{MCTS.}

The standard Monte Carlo Tree Search (MCTS) process consists of four steps:\supercite{sun2022symbolic} 1) selection, where the best candidate formulas are identified based on their performance; 2) expansion, where new symbolic formulas (which may be incomplete) are generated by adding one more operator from the incomplete formula; 3) simulation, where the incomplete formulas are randomly simulated and then evaluated based on their predictive accuracy; and 4) back-propagation, where the results are propagated back to update rewards and visit times of each MCTS node. We employed MCTS to refine the target formula obtained by the end-to-end model. Specifically, random sub-trees were removed from the expression tree of the target formula, and MCTS was invoked with the resulting incomplete formula to search for a better target formula. 
Moreover, following the work of Sun et al.\supercite{sun2022symbolic}, we constructed a grammar pool to reduce the search space of MCTS. 
Specifically, this grammar pool consisted of basic operators ($+$, $-$, $\times$, $/$, $\sin(x)$, $\cos(x)$, $\exp(x)$) and the operators extracted from the predicted results of the end-to-end model.
During the expansion steps, MCTS was restricted to select operators only from the grammar pool to construct a complete symbolic formula.

\paragraph{Genetic programming.}

Our genetic programming (GP) refinement module followed the work of PySR\supercite{cranmer2023interpretable}, where we used the results from the end-to-end model and MCTS as the initial populations for the GP algorithm. 
During each round of evolution, each formula in the population undergoes random mutations, including appending, changing, or deleting specific operators in the formula. Crossover operations were also performed between individuals to combine their expressions and create new formulas. 
The entire GP algorithm optimizes the formulas for 40 rounds. The optimization process stops early if the MSE of one of the generated formulas achieves $10^{-6}$ during the process.
Finally, a Pareto front of the formulas was generated. The best formula from this set was selected based on the criterion of $\text{MSE}\times0.99^{\text{complexity}}$ to serve as the final solution, considering a trade-off between both accuracy and complexity.

\subsection{Details of test data}
\label{methods:dataDividing}
To ensure that no data leakage occurred, we strictly split the training and test datasets. While it is not straightforward to identify mathematically equivalent formulas via symbolic derivations, we chose to measure the similarity between two formulas based on numerical evaluations. Specifically, we define the similarity between two formulas, $f_i$ and $f_j$, with the same number of features, using the averaged $R^2$ score:
\[
\mathrm{sim}(f_i, f_j)\defeq 1 - \frac12 \left[\frac{\sum_{k=1}^{N'}(f_i(\bm{x}_k)-f_j(\bm{x}_k))^2}{\sum_{k=1}^{N'}(f_i(\bm{x}_k)-\hat{f_i}))^2} + \frac{\sum_{k=1}^{N'}(f_i(\bm{x}_k)-f_j(\bm{x}_k))^2}{\sum_{k=1}^{N'}(f_i(\bm{x}_k)-\hat{f_j}))^2}\right],
\]
where $\hat{f_i} \defeq [\sum_{k=1}^{N'}f_i(\bm{x}_k)]/N'$, $N'=40$, and the $\bm{x}_k$'s are independently sampled from the standard multivariate Gaussian distribution. The similarity is set to $0$ if there exists an $\bm{x}_k$ such that exactly one of $f_i(\bm{x}_k)$ and $f_j(\bm{x}_k)$ is undefined. The similarity between two formulas with a different number of features is also set to $0$. Observe that two mathematically equivalent formulas achieve the maximum possible similarity $1$.

We selected 3,000 formulas such that the pairwise similarity is less than $0.99$ to form our test set. For the training dataset, we only removed the symbolically identical formulas while not requiring the pairwise similarity to be away from $1$.
This is because formulas that are mathematically equivalent but with different symbolic forms may help the model understand the mathematical equivalences. 
Moreover, such formulas might originate from different physical scenarios and have different physical units, thus still represent distinct entities.
The formulas in the test set were further divided into three difficulty levels based on their maximum similarity with the formulas in the training dataset. 
Formulas with a similarity greater than $0.95$ were defined as the ones that were easy to predict, those with a similarity between $0.8$ and $0.95$ were considered as medium difficulty, and formulas with a similarity less than $0.8$ were considered as hard.

\subsection{Evaluation metrics}
\label{methods:metrics}

We evaluate the performance of our model and other baseline models using the following metrics.

\begin{itemize}
\item \textit{Symbolic accuracy.} This evaluation metric was introduced by Cava et al.\supercite{la2021contemporary} Let the target formula be $f$. The symbolic accuracy of a predicted formula $g$ is $1$ if there exists a constant such that $f - g \equiv c$ or $f = c \cdot g$ ($c \neq 0$ in the second case); otherwise, the symbolic accuracy is $0$.

\item \textit{Numerical precision.} 
This metric is used to evaluate the numerical precision of data points given by predicted formula. We utilize $R^2$-score for the predicted formula $g$ on testing data points $\{\bm x^{\mathrm{test}}_k, y_{k}^{\mathrm{test}}\}_{k=1}^{N}$ ($N = 200$) which are sampled from the same distribution as the training data points:
$$
R^2=1-\frac{\sum_{i=1}^{N}(y^{\mathrm{test}}_i - g(\bm x_i^{\mathrm{test}}))^2}{\sum_{i=1}^{N}(y^{\mathrm{test}}_i - \bar{y})^2}
$$
where $\bar{y}\defeq\frac{1}{N}\sum_{i=1}^{N}y_i^{\mathrm{test}}$.

\item \textit{Accuracy of physical units.} This metric aims to assess the capability of a model to generate formulas with consistent physical units. The accuracy is $1$ if all operations in the formula work with compatible physical units and the physical unit of the result is the same as the target formula; otherwise, the accuracy of physical units is $0$.

\item \textit{Formula complexity.} The complexity of a formula $f$ is defined to be the number of operators, variables and constant parameters in the formula. We also define the relative complexity of a prediction $g$ is the difference (in absolute value) of the complexity of $g$ and that of the target formula $f$.

\item \textit{Formula depth.} The depth of a formula $f$ is the maximum depth of its expression tree.
\end{itemize}

\section{Data availability}
AI Feynmen data can be downloaded from Feynman Symbolic Regression Database (\url{https://space.mit.edu/home/tegmark/aifeynman.html}). 
The formula generated by the LLM, including the training and test datasets can be downloaded\supercite{ying2025phye2e} from \url{https://figshare.com/articles/dataset/PhyE2E_datas/29615831/1}.
The sunspot number data could be downloaded from the Sunspot Index and Long-term Solar Observations website (\url{https://www.sidc.be/SILSO/datafiles}).
The plasma pressure data could be downloaded from Geotail and Time History of Events and Macroscale Interactions During Substorms (THEMIS) website (\url{https://themis.igpp.ucla.edu/overview_data.shtml}).
The solar rotational angular velocity data could be found at table presented by Snodgrass et al. 1983 in Magnetic Rotation of the Solar Photosphere\supercite{snodgrass1983magnetic}.
We collected the contribution function of emission lines data from the CHIANTI website (\url{http://www.chiantidatabase.org}).
Lunar tide signal data was downloaded from RBSP/EFW official website (\url{https://www.space.umn.edu/rbspefw-data/}).

\section{Code availability}
Codes for running \ourModel including both training and test modules are accessible at \url{https://github.com/Jie0618/PhysicsRegression} with a permanent version available\supercite{jie_ying_2025_16305086} via Zenodo at \url{https://doi.org/10.5281/zenodo.16305086}.
The pre-trained \ourModel can be downloaded at \url{https://figshare.com/articles/dataset/PhyE2E_datas/29615831/1}.

\bibliography{ref}

\begin{figure}[!htb]
\centering
\includegraphics[width=0.9\textwidth]{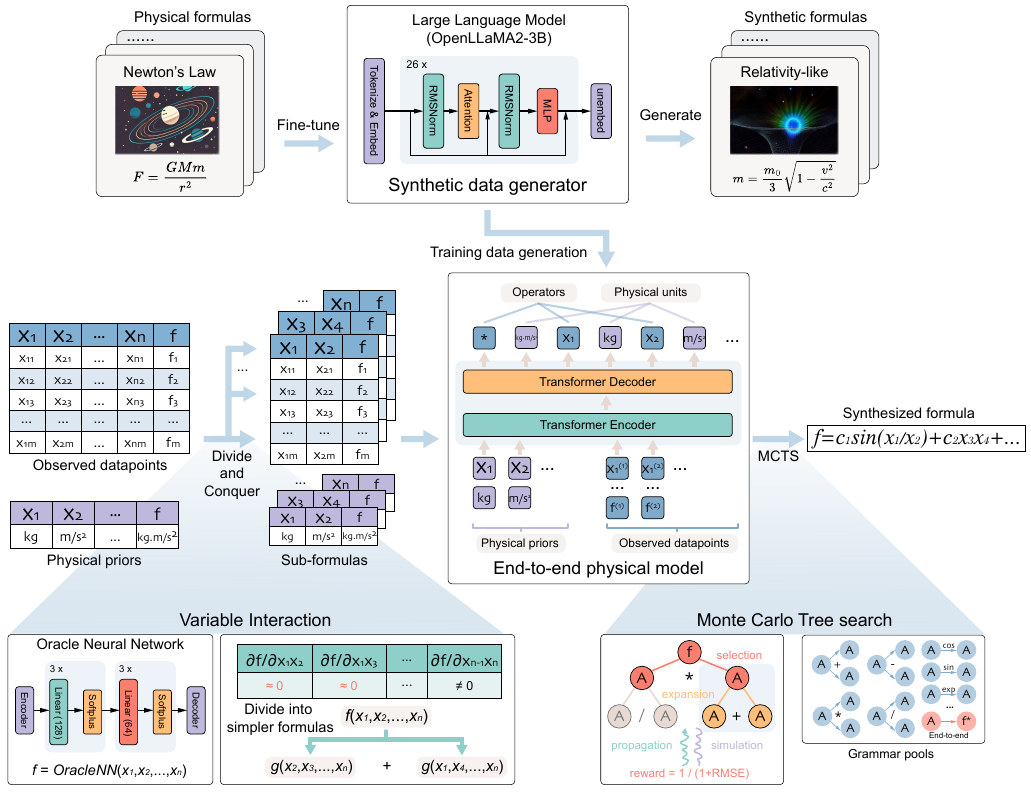}
\caption{\textbf{The overall \ourModel framework.} \textbf{Top.} The training dataset was augmented with a large-scale synthetic dataset generated by a large language model. \textbf{Middle.} A variable interaction technique was integrated to decompose the original symbolic regression problem into simpler sub-problems, referred to as Divide-and-Conquer(D\&C). An end-to-end model was trained to predict the target formula using observed data points and prior physical knowledge (referred to as ``physical priors''). \textbf{Bottom.} Monte Carlo Tree Search (MCTS) module was adopted to refine the generated formulas, using a context-free grammar pool that includes atomic formulas and the end-to-end generated formula.}
\label{fig1}
\end{figure}

\begin{figure}[!htb]
\centering
\includegraphics[width=1.0\textwidth]{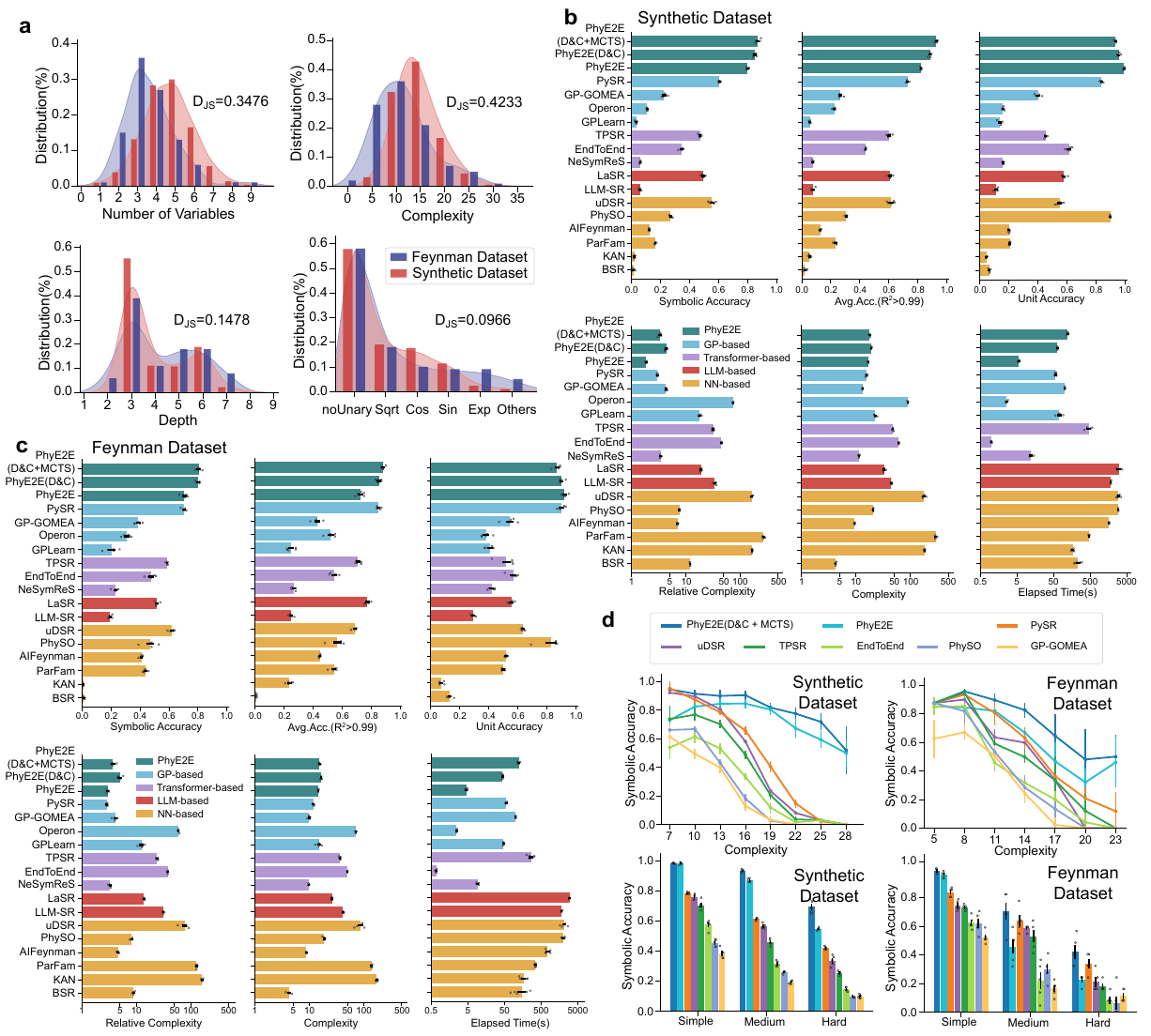}
\caption{\textbf{Performance on the synthetic and AI Feynman datasets.} 
\textbf{a,} Comparison between the formulas generated from LLaMa2 and the Feynman Dataset. The distance between the distributions of different properties of the two sets of formulas is measured using the Jensen-Shannon divergence (D$_{\text{JS}}$).
\textbf{b,c,} Evaluation results for Symbolic Regression methods on the test set of the synthetic dataset and AI Feynman dataset, respectively. Data are presented as mean values ± SEM (n=5 individual trials for each baselines).
\textbf{d,} Evaluation results on formulas with different complexity (upper panels) and different difficulties (bottom panels) on the synthetic and AI Feynman datasets. The bar plots represent mean values ± SEM (n=5 individual trials for each baselines).}\label{fig2}
\end{figure}

\begin{figure}[!htb]
\centering
\includegraphics[width=0.91\textwidth]{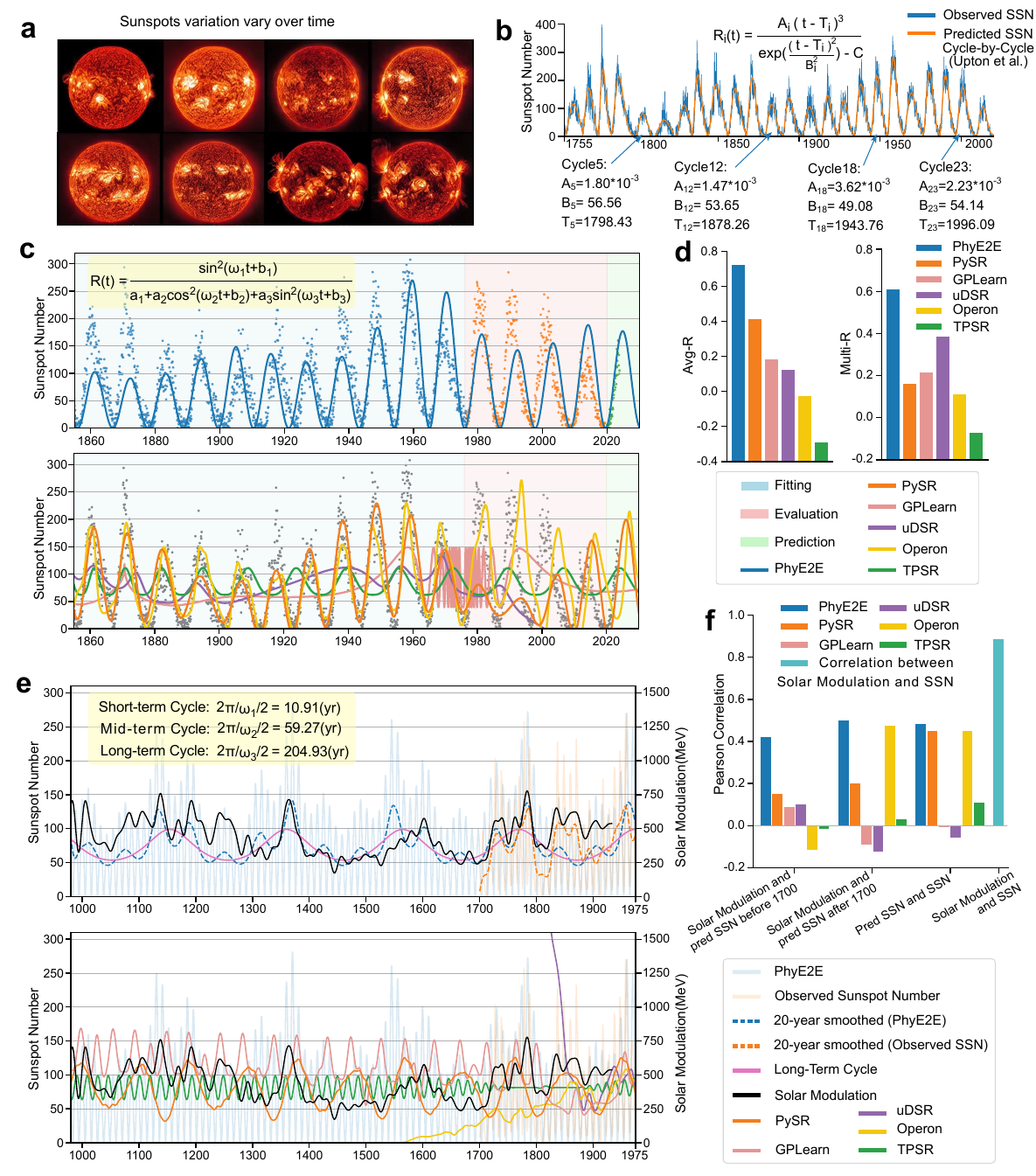}
\caption{\textbf{Performance of sunspot intensity prediction}
\textbf{a,} Sunspot variation over time.
\textbf{b,} Variations in SSN observed through telescopes from 1755 to 2020 and the formula derived by Hathaway et al., 1994.
\textbf{c,} The \ourModel formula and the variations in SSN yielded by the formula from 1855 to 1976 (top). The formulas generated by other baseline models and the variations in SNN yielded by these formulas from 1855 to 1976 (bottom).
\textbf{d,} Avg-R (left) and Multi-R (right) on the test data from 1976 to 2019 for different baseline models.
\textbf{e,} Solar modulation level and smoothed SSN from different baseline models over a longer time frame from 980 to 1976.
\textbf{f,} Pearson Correlation between the SSN observed by telescopes, SSN predicted by the generated formulas, and Solar Modulation level from 980 to 1932.}
\label{fig:SSN}
\end{figure}

\begin{figure}[!htb]
\centering
\includegraphics[width=1.0\textwidth]{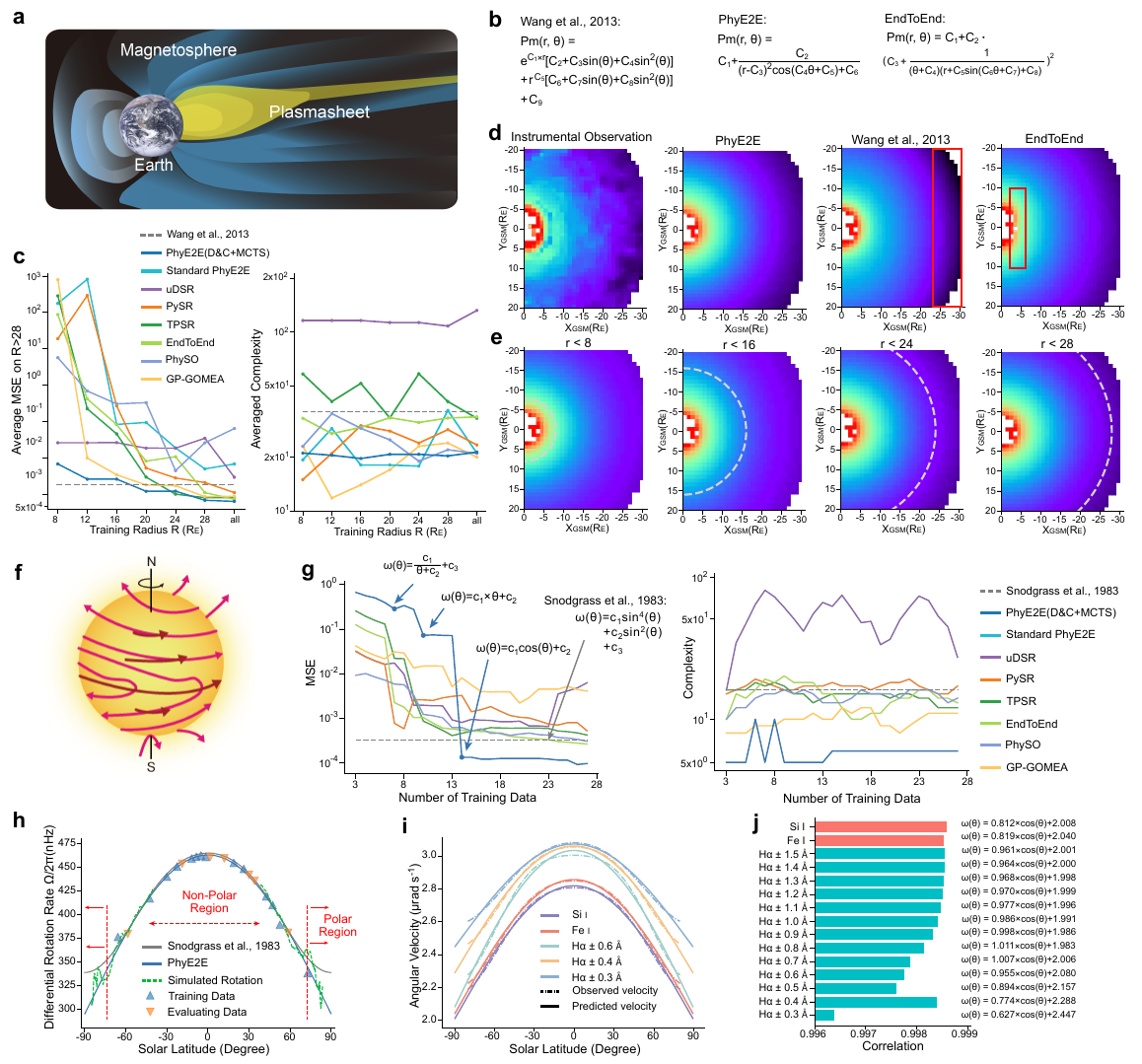}
\caption{\textbf{Performance of plasma sheet pressure prediction and solar differential rotation prediction}
\textbf{a,} The distribution of near-Earth magnetosphere and plasmasheet.
\textbf{b,} Symbolic formulas of Wang et al., 2013 and \ourModel.
\textbf{c,} Average Mean Square Error (left) and complexity (right) when utilizing data from different radius for models to be compared.
\textbf{d,} Instrumental observations and formula predictions for plasma sheet pressure using different models.
\textbf{e,} Predictions for plasma sheet pressure using data from different radius by \ourModel.
\textbf{f,} Solar rotation varies at different latitudes, making magnetic field lines stretched and twisted.
\textbf{g,} MSE and complexity from different models using different numbers of training data.
\textbf{h,} Predictions from Snodgrass et al., 1993 and \ourModel across all the latitudes.
\textbf{i,} Predictions of solar atmosphere, using data from various spectral lines in the photosphere and the chromosphere.
\textbf{j,} \ourModel predicts consistent formulas with high robustness across various spectral lines.
}
\label{fig:plasma_rotation}
\end{figure}

\begin{figure}[!htb]
\centering
\includegraphics[width=1.0\textwidth]{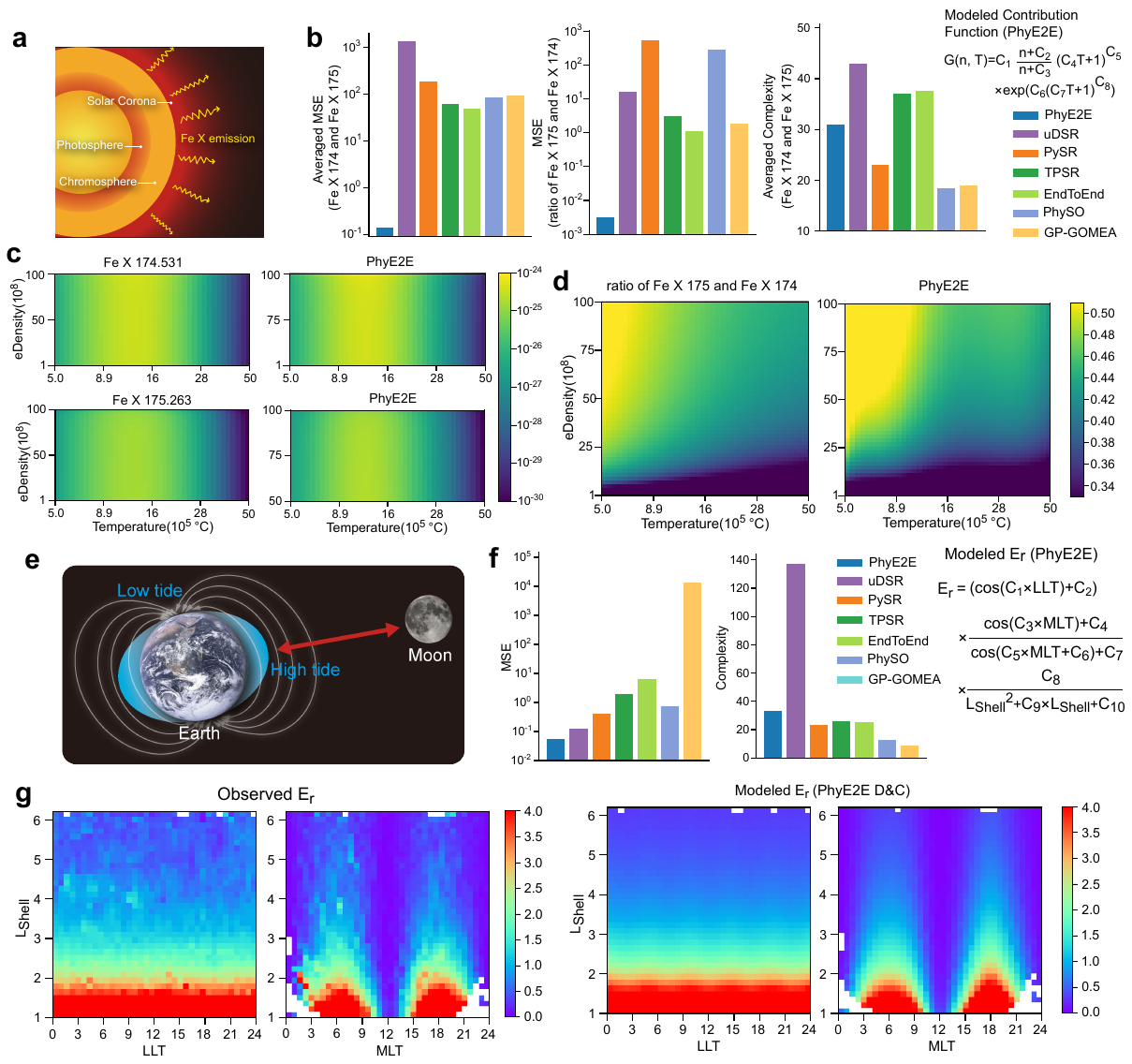}
\caption{\textbf{Performance of contribution function of emission lines predictions and lunar tide signal of plasma layer predictions}
\textbf{a,} Emission lines in the extreme ultraviolet spectrum of the Sun.
\textbf{b,} Average MSE for Fe X 174 and Fe X 175 (left), MSE of the ratio between the two emission lines (middle), and the complexity (right) of the formulas generated by different models to be compared. 
\textbf{c,} Instrumental measured contribution function and \ourModel predictions for Fe X 174 and Fe X 175.
\textbf{d,} Instrumental measured ratio of the two emission lines and \ourModel predictions.
\textbf{e,} Tidal radial electric fields influences the Earth's magnetospheric electric fields.
\textbf{f,} MSE and complexity for different models to be compared.
\textbf{g,} Instrumental measured radial electric field ($E_r$) (left) and \ourModel predictions for dayside and nightside of the Earth.}
\label{fig:contributionFunction_lunarTide}
\end{figure}

\renewcommand{\thefigure}{S\arabic{figure}}
\setcounter{figure}{0}

\begin{figure}[!htb]%
\centering
\includegraphics[width=1.0\textwidth]{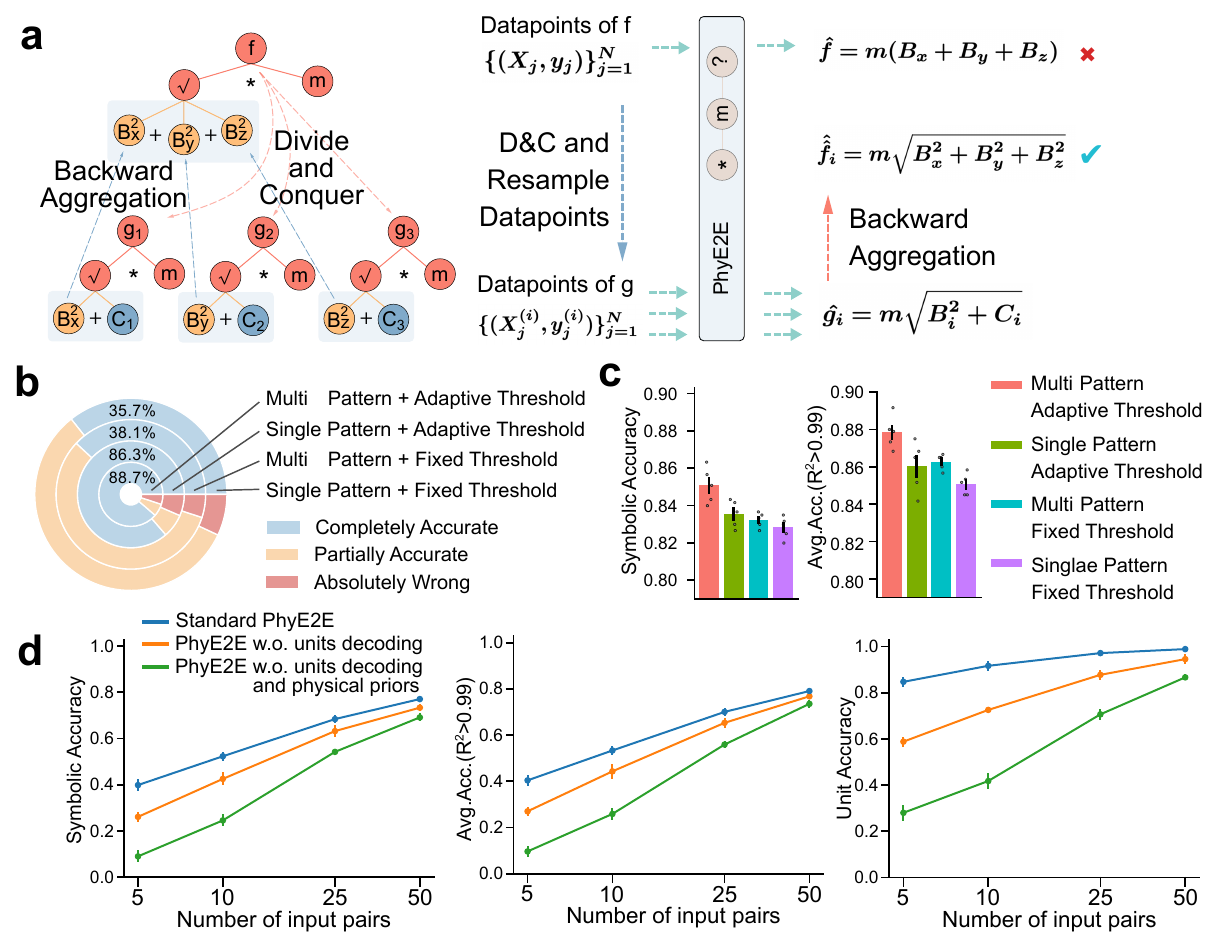}
\caption{\textbf{Performance on Synthetic and AI Feynman datasets.}
\textbf{a,} An example of D\&C procedure, including Divide-and-Conquer and Backward Aggregation step.
\textbf{b,} Decomposition accuracy for different D\&C strategies.
\textbf{c,} Symbolic accuracy and average accuracy of $R^2>0.99$ under different D\&C strategies. Data are presented as mean values ± SEM (n=5 individual trials for each configuration).
\textbf{d,} The accuracy performance on low data cases with different physical priors incorporated into PhyE2E. Data are presented as mean values ± SEM (n=5 individual trials for each configuration).}
\label{fig:Divide-and-Conquer-performance}
\end{figure}

\clearpage

\appendix
\section{Supplementary Figures}

\counterwithin{figure}{section}
\counterwithin{table}{section}

\renewcommand{\thefigure}{S\arabic{figure}}
\renewcommand{\thetable}{S\arabic{table}}

\setcounter{figure}{0}
\setcounter{table}{0}

\begin{figure}[!htb]%
\centering
\includegraphics[width=1.0\textwidth]{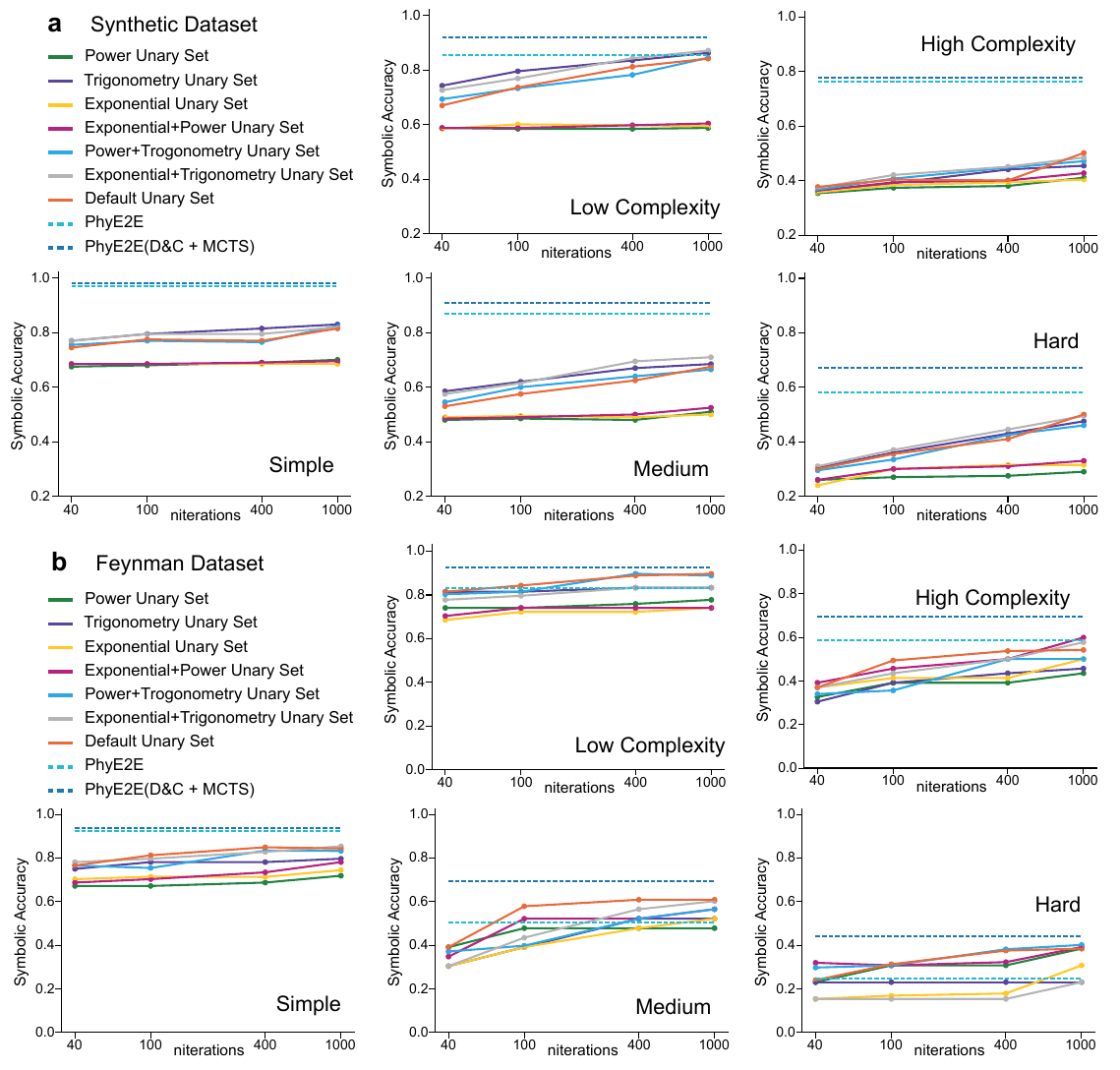}
\caption{\textbf{Detailed performance comparison between PhyE2E and PySR under different operator sets and different iterations.} 
\textbf{a,} Evaluation results of the synthetic dataset using various PySR configurations on formulas with different complexity levels (upper panels) and different difficulties (bottom panels). 
\textbf{b,} Evaluation results of the Feynman dataset using various PySR configurations on formulas with different complexity levels (upper panels) and different difficulties (bottom panels)}
\label{fig:fig-S2}
\end{figure}

\newpage

\begin{figure}[!htb]%
\centering
\includegraphics[width=1.0\textwidth]{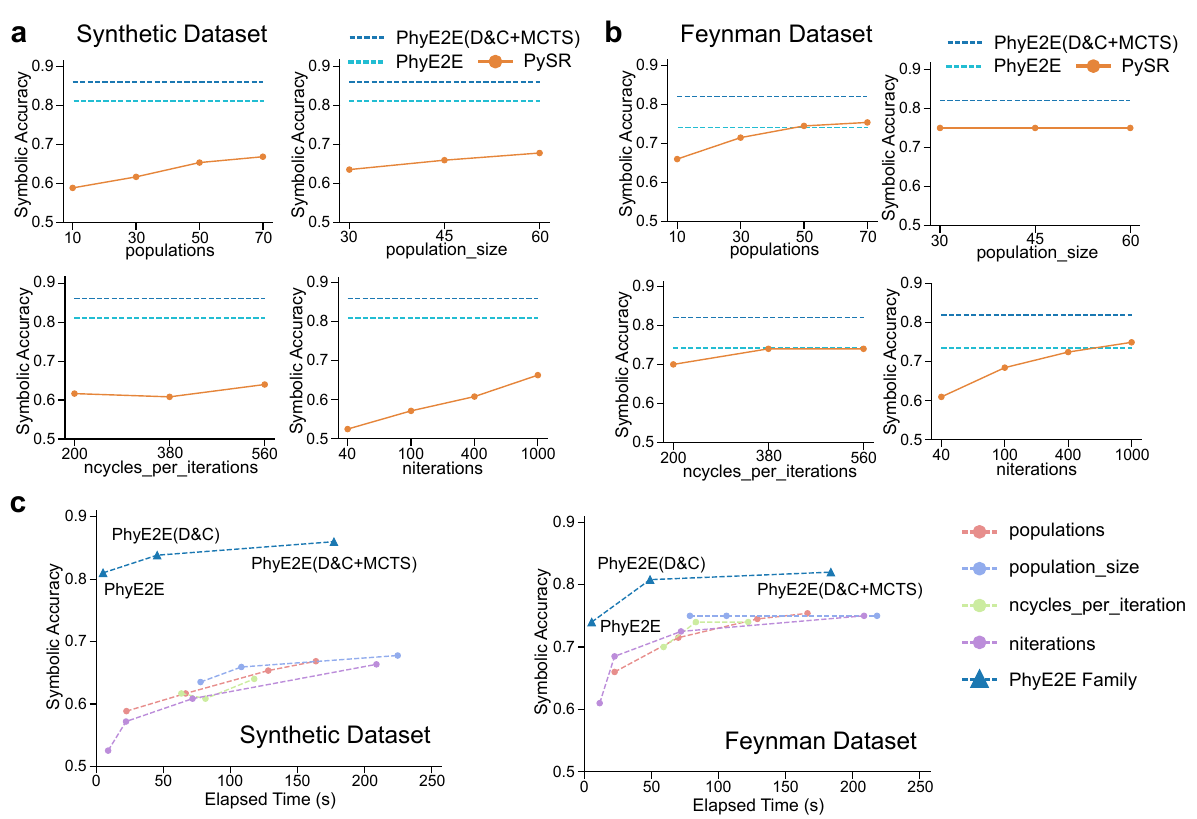}
\caption{\textbf{Detailed performance comparison between PhyE2E and PySR under different hyperparameter settings. } 
\textbf{a,} Evaluation results in terms of symbolic accuracy on the synthetic dataset, comparing PhyE2E and PySR across different search sizes. 
\textbf{b,} Evaluation results in terms of symbolic accuracy on the Feynman dataset, comparing PhyE2E and PySR across different search sizes.
\textbf{c,} Evaluation results in terms of symbolic accuracy with its corresponding elapsed times on the synthetic dataset (left panel) and the AI Feynman dataset (right panel), comparing PhyE2E and PySR across different search sizes.}
\label{fig:fig-S3}
\end{figure}

\newpage

\begin{figure}[!htb]%
\centering
\includegraphics[width=0.6\textwidth]{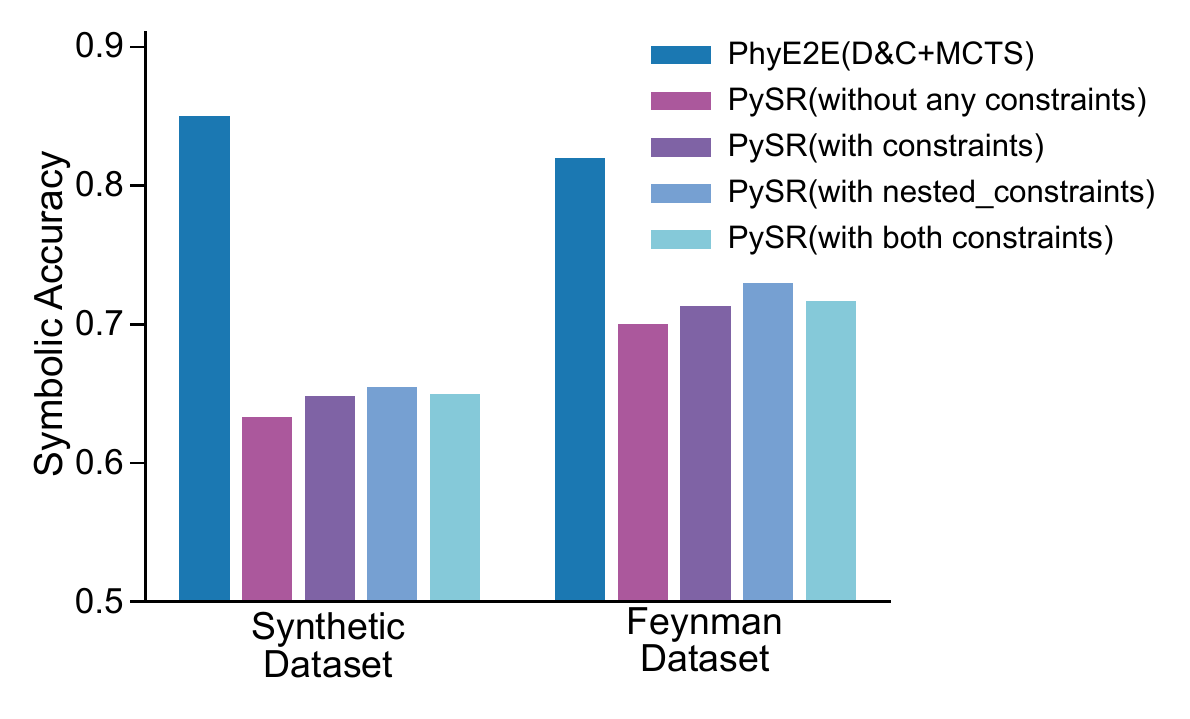}
\caption{\textbf{Evaluation results of four different PySR constraint configurations on the synthetic dataset (left) and the Feynman dataset (right).}}
\label{fig:fig-S4}
\end{figure}

\begin{figure}[!htb]%
\centering
\includegraphics[width=0.5\textwidth]{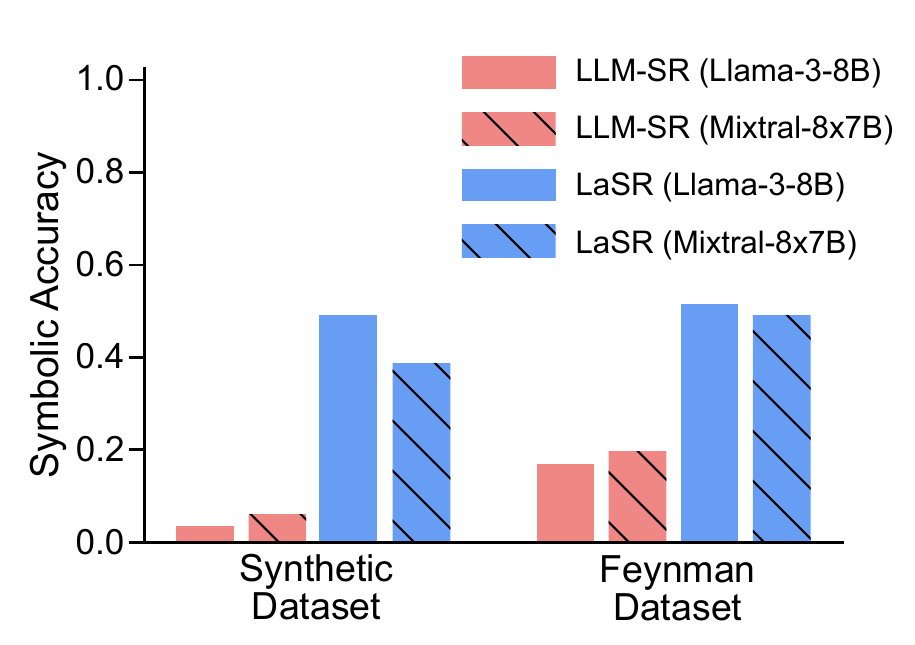}
\caption{\textbf{Evaluation results in terms of symbolic accuracy on the synthetic dataset (left) and the Feynman dataset (right), comparing LLM-based models using different LLM backbones including Llama-3-8B and Mixtral-8x7B.}}
\label{fig:fig-S5}
\end{figure}

\newpage

\section{Supplementary Tables}

\begin{table}[!htb]
\centering
\caption{Real-world physics formulas from Feynman Datasets}
\label{tab:feynman_dataset}
\includegraphics[width=0.93\textwidth]{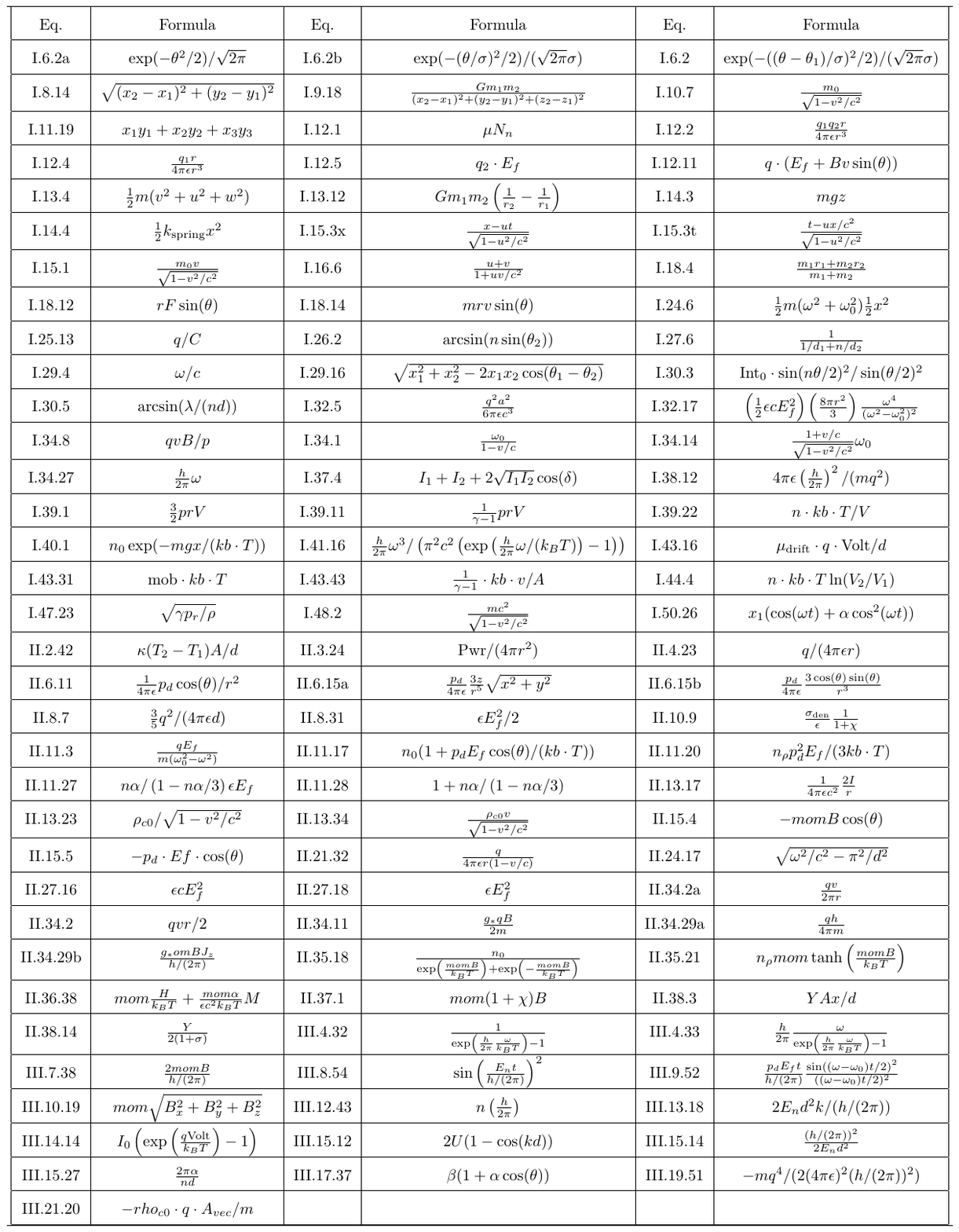}
\end{table}

\newpage

\begin{table}[!htb]
\centering
\caption{Constants of the derived formula from \ourModel for SSN prediction}
\label{tab:constants_phy1}
\includegraphics[width=0.8\textwidth]{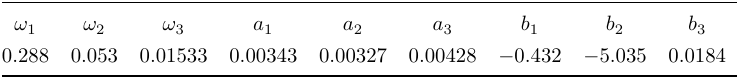}
\end{table}

\begin{table}[!htb]
\centering
\caption{Constants of derived formulas from \ourModel for plasma pressure prediction}
\label{tab:constants_phy2}
\includegraphics[width=0.7\textwidth]{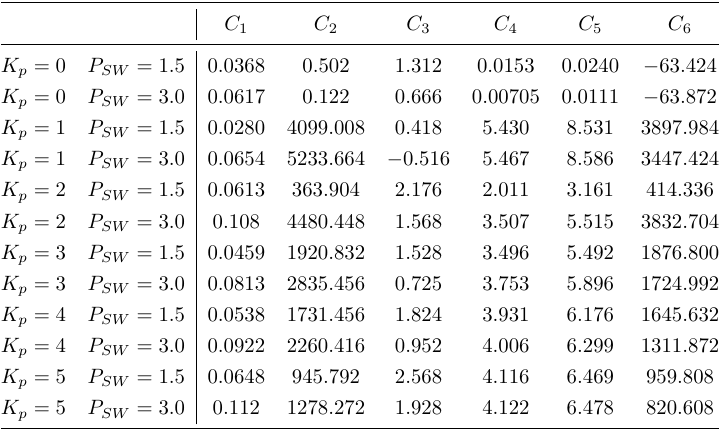}
\end{table}

\begin{table}[!htb]
\centering
\caption{Constants of derived formulas from \ourModel for differential rotation prediction}
\label{tab:constants_phy3}
\includegraphics[width=0.15\textwidth]{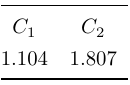}
\end{table}

\begin{table}[!htb]
\centering
\caption{Constants of derived formulas from \ourModel for the prediction of contribution functions}
\label{tab:constants_phy4}
\includegraphics[width=0.9\textwidth]{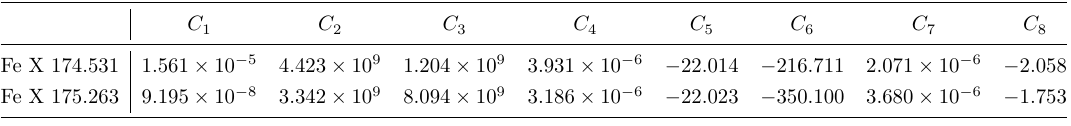}
\end{table}

\begin{table}[!htb]
\centering
\caption{Constants of derived formulas from \ourModel for the prediction of lunar tide signal}
\label{tab:constants_phy5}
\includegraphics[width=0.8\textwidth]{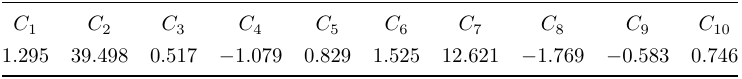}
\end{table}

\newpage

\begin{table}[!htb]
\centering
\caption{Formulas from \ourModel for SSN prediction on different time frame}
\label{tab:formulas_phy1_otherYear}
\includegraphics[width=1.0\textwidth]{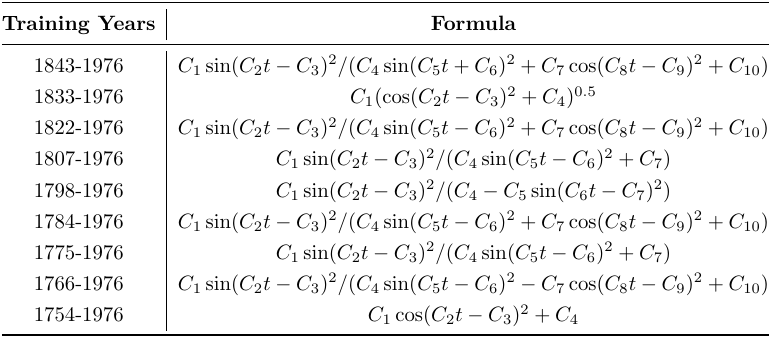}
\end{table}

\begin{table}[!htb]
\centering
\caption{Constants of derived formulas from \ourModel for SSN prediction on different time frame}
\label{tab:constants_phy1_otherYear}
\includegraphics[width=1.0\textwidth]{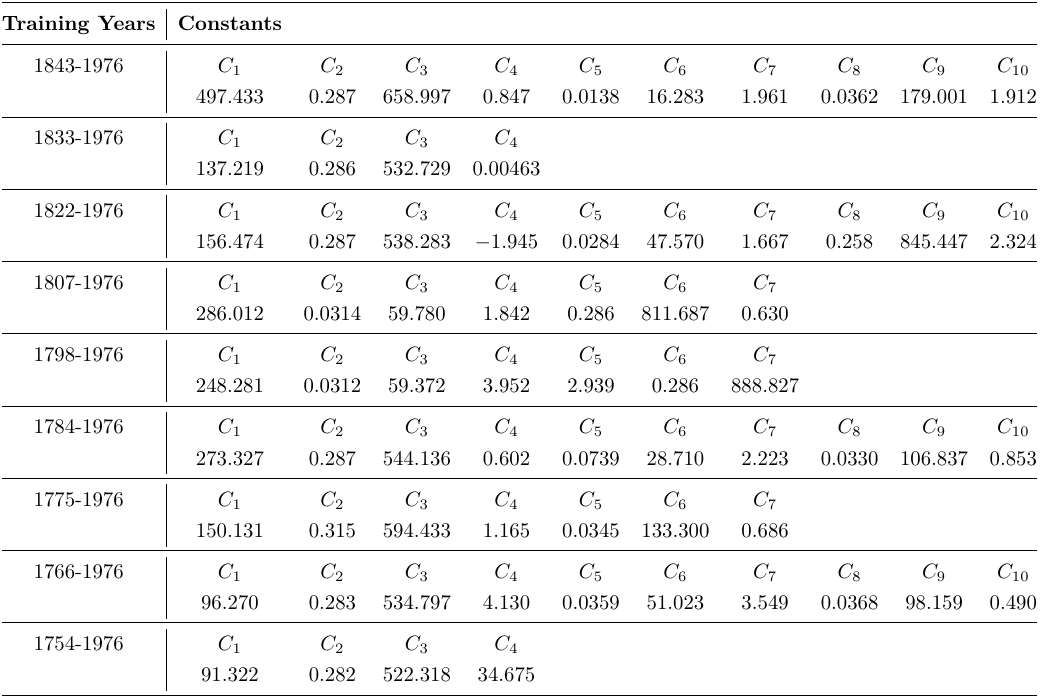}
\end{table}

\newpage

\begin{table}[!htb]
\centering
\caption{Derived Formulas from all the baseline models for SSN prediction and its short-term cycle}
\label{tab:phy1_baseline}
\includegraphics[width=0.95\textwidth]{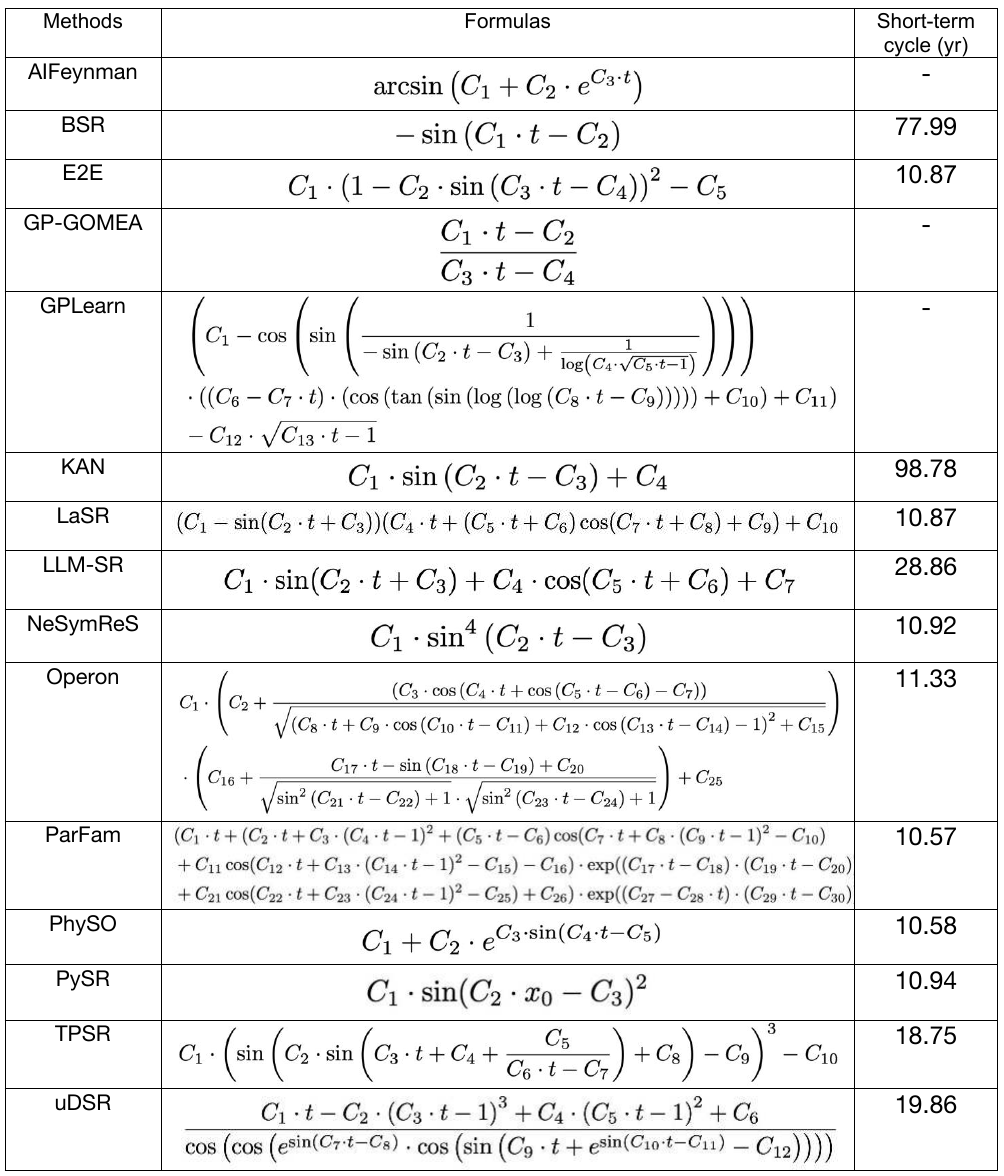}
\end{table}

\newpage

\begin{table}[!htb]
\centering
\caption{Constants of derived formulas from all the baseline models for SSN prediction}
\label{tab:phy1_consts}
\includegraphics[width=0.92\textwidth]{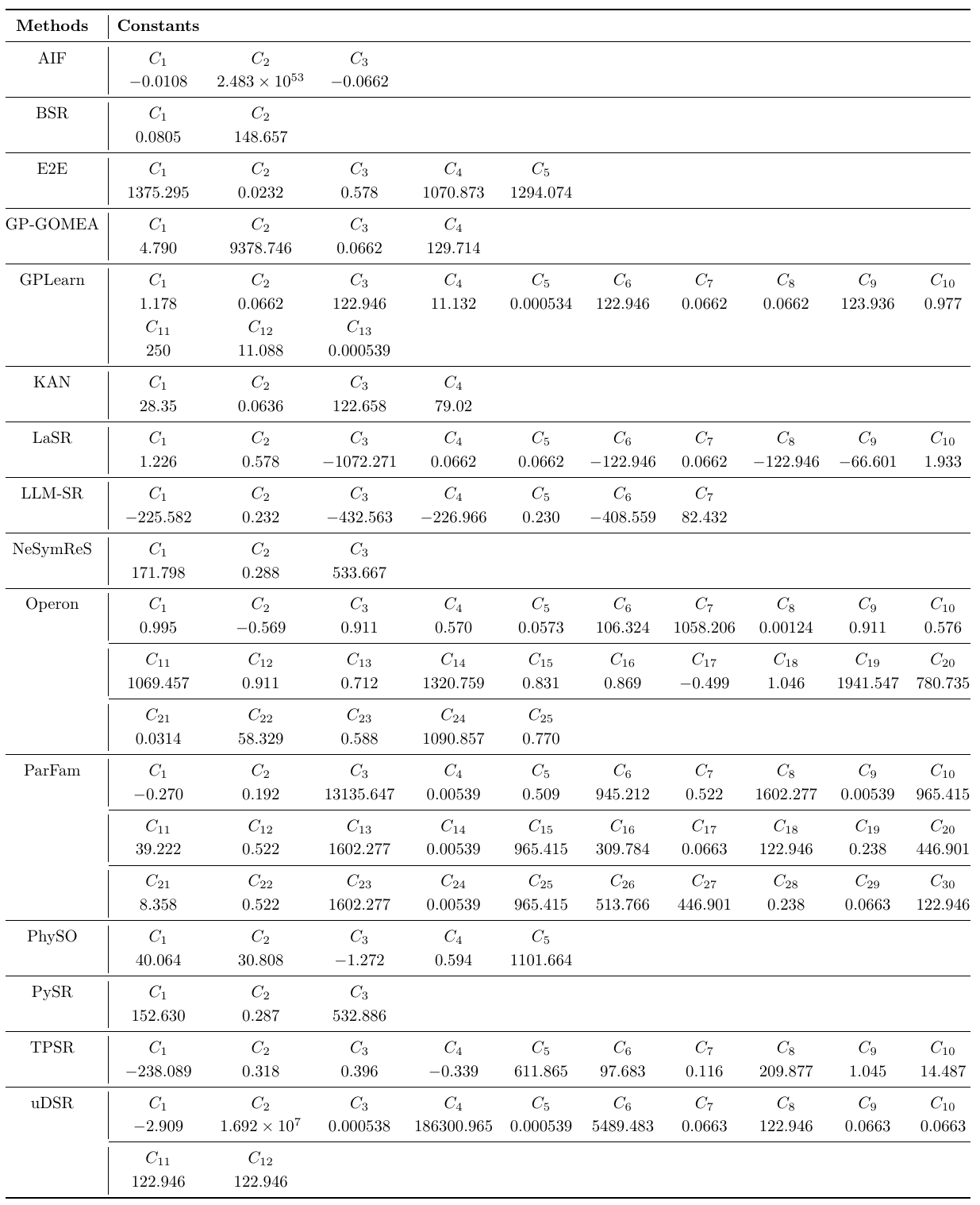}
\end{table}

\newpage

\begin{table}[!htb]
\centering
\caption{Derived formulas from PySR under different operator set and constraint configurations for SSN prediction}
\label{tab:phy1_pysr_formulas}
\includegraphics[width=1.0\textwidth]{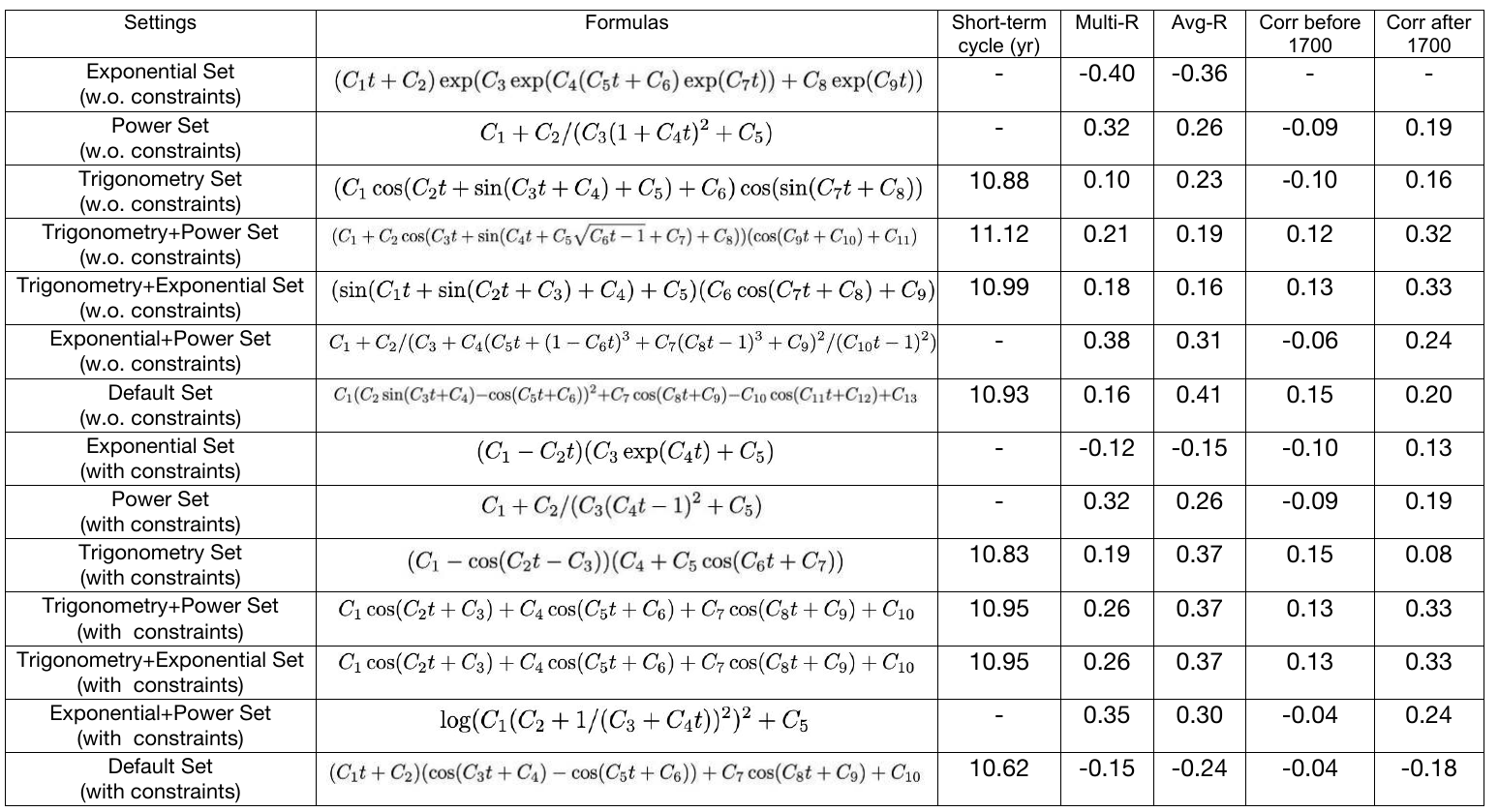}
\end{table}

\newpage

\begin{table}[!htb]
\centering
\caption{Constants of derived formulas from PySR under different operator set and constraint configurations for SSN prediction}
\label{tab:phy1_pysr_consts}
\includegraphics[width=1.0\textwidth]{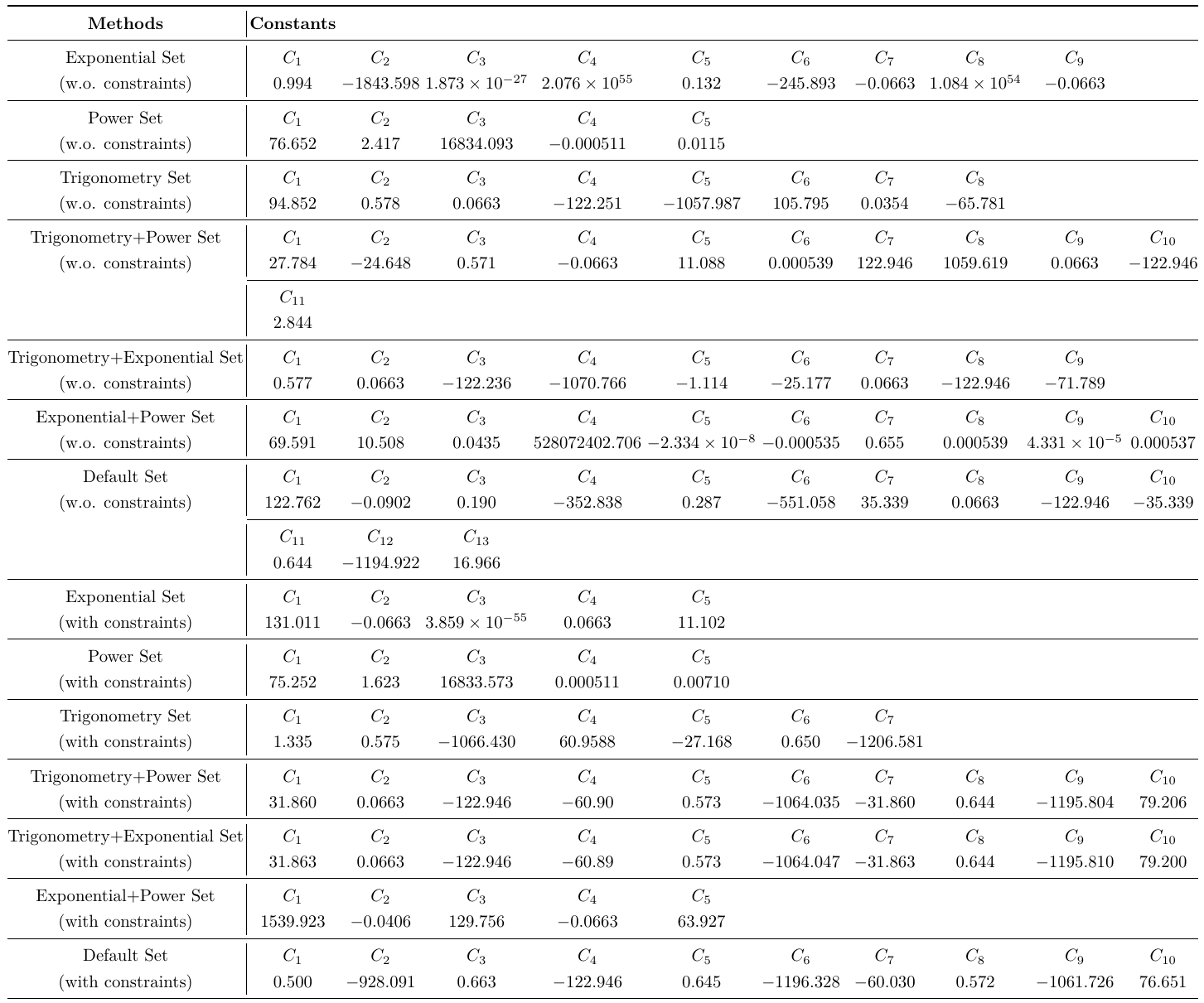}
\end{table}

\newpage

\begin{table}[!htb]
\centering
\caption{Derived formulas from all the baseline models for plasma pressure prediction}
\label{tab:phy2_baseline}
\includegraphics[width=0.85\textwidth]{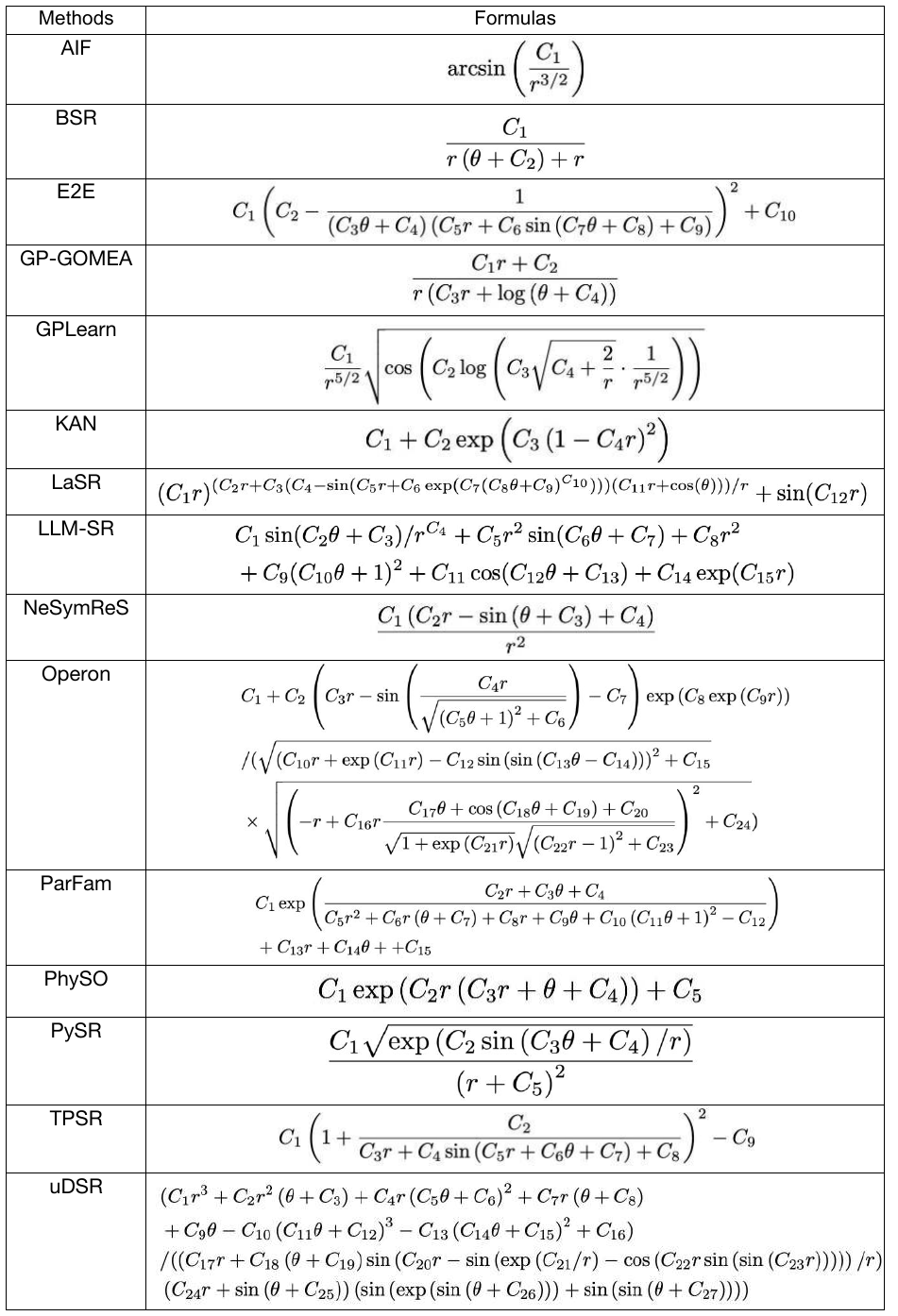}
\end{table}

\newpage

\begin{table}[!htb]
\centering
\caption{Constants of derived formulas from all the baseline models for plasma pressure prediction}
\label{tab:phy2_consts}
\includegraphics[width=0.85\textwidth]{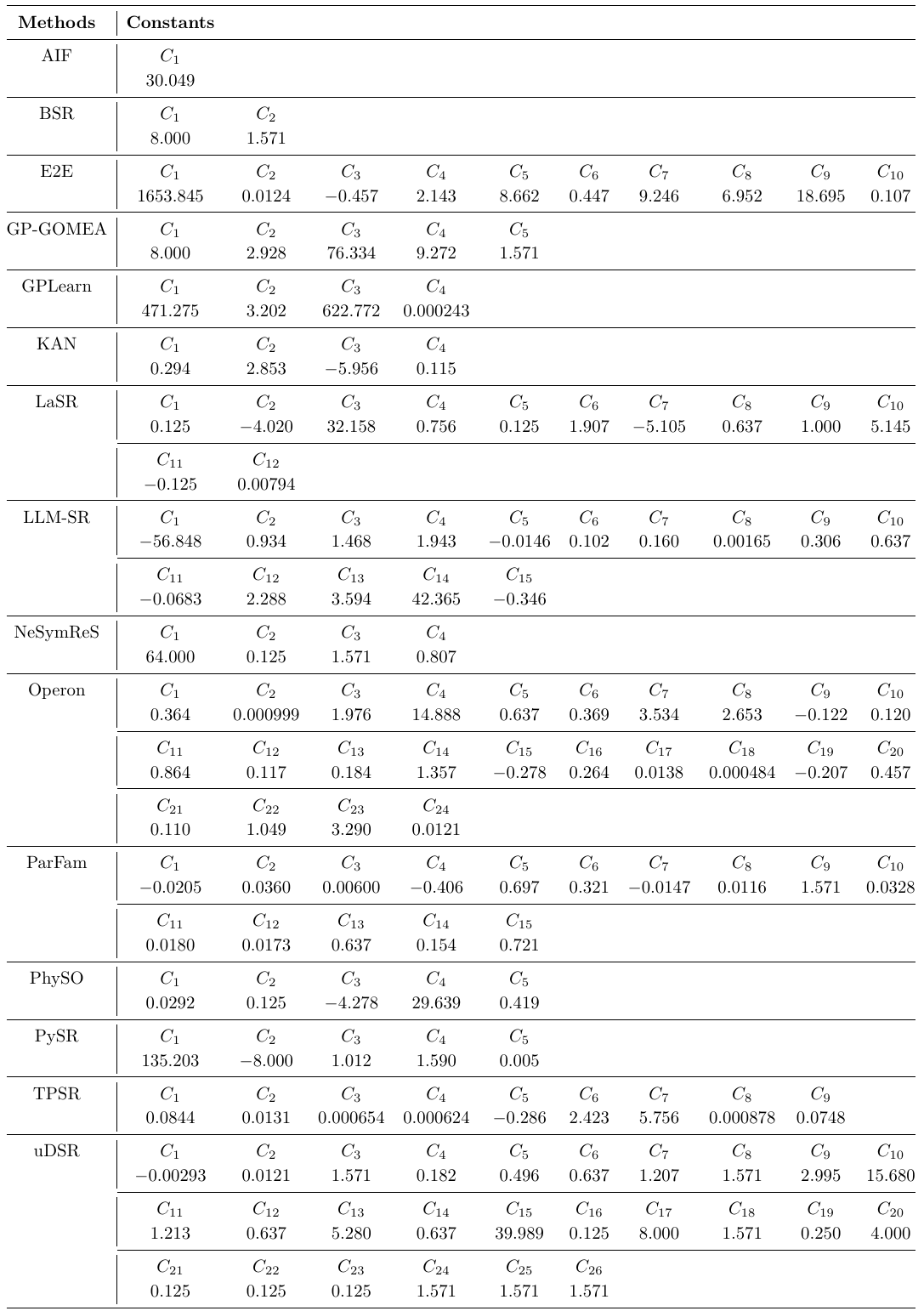}
\end{table}

\newpage

\begin{table}[!htb]
\centering
\caption{Derived formulas from all the baseline models for differential rotation prediction}
\label{tab:phy3_baseline}
\includegraphics[width=0.90\textwidth]{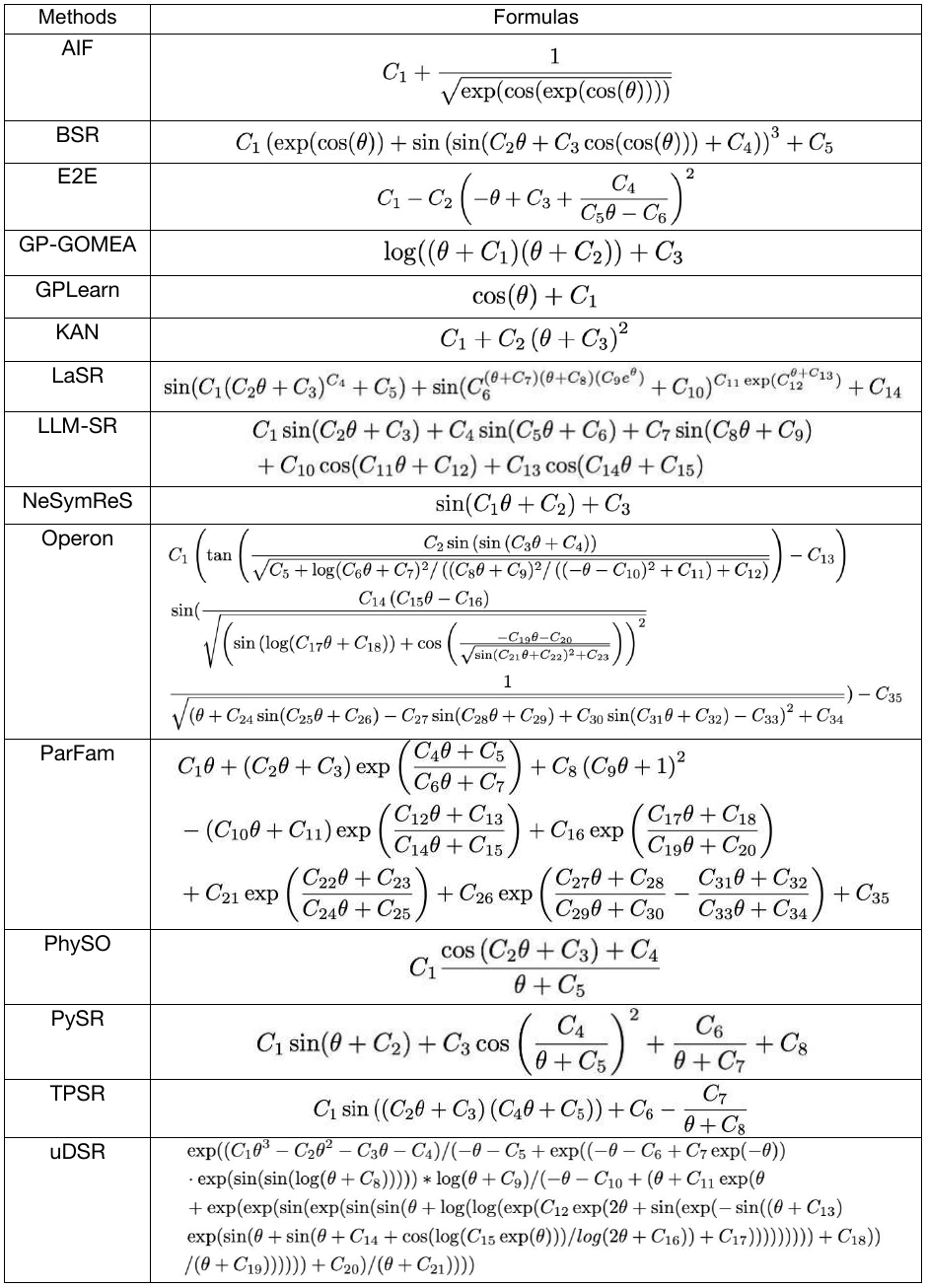}
\end{table}

\newpage

\begin{table}[!htb]
\centering
\caption{Constants of derived formulas from all the baseline models for differential rotation prediction}
\label{tab:phy3_consts}
\includegraphics[width=0.86\textwidth]{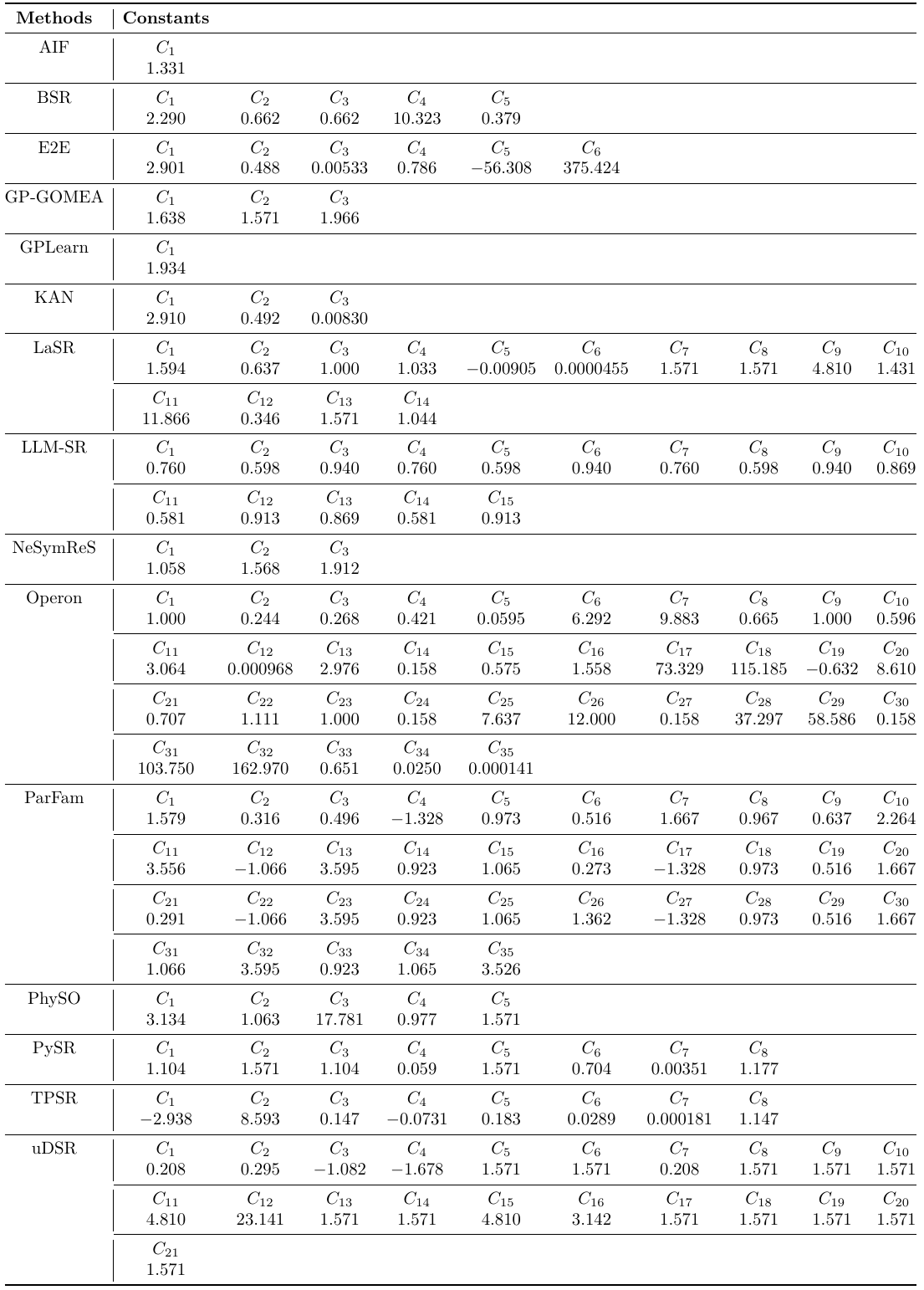}
\end{table}

\newpage

\begin{table}[!htb]
\centering
\caption{Derived formulas from all the baseline models for the prediction of contribution functions}
\label{tab:phy4_baseline}
\includegraphics[width=0.8\textwidth]{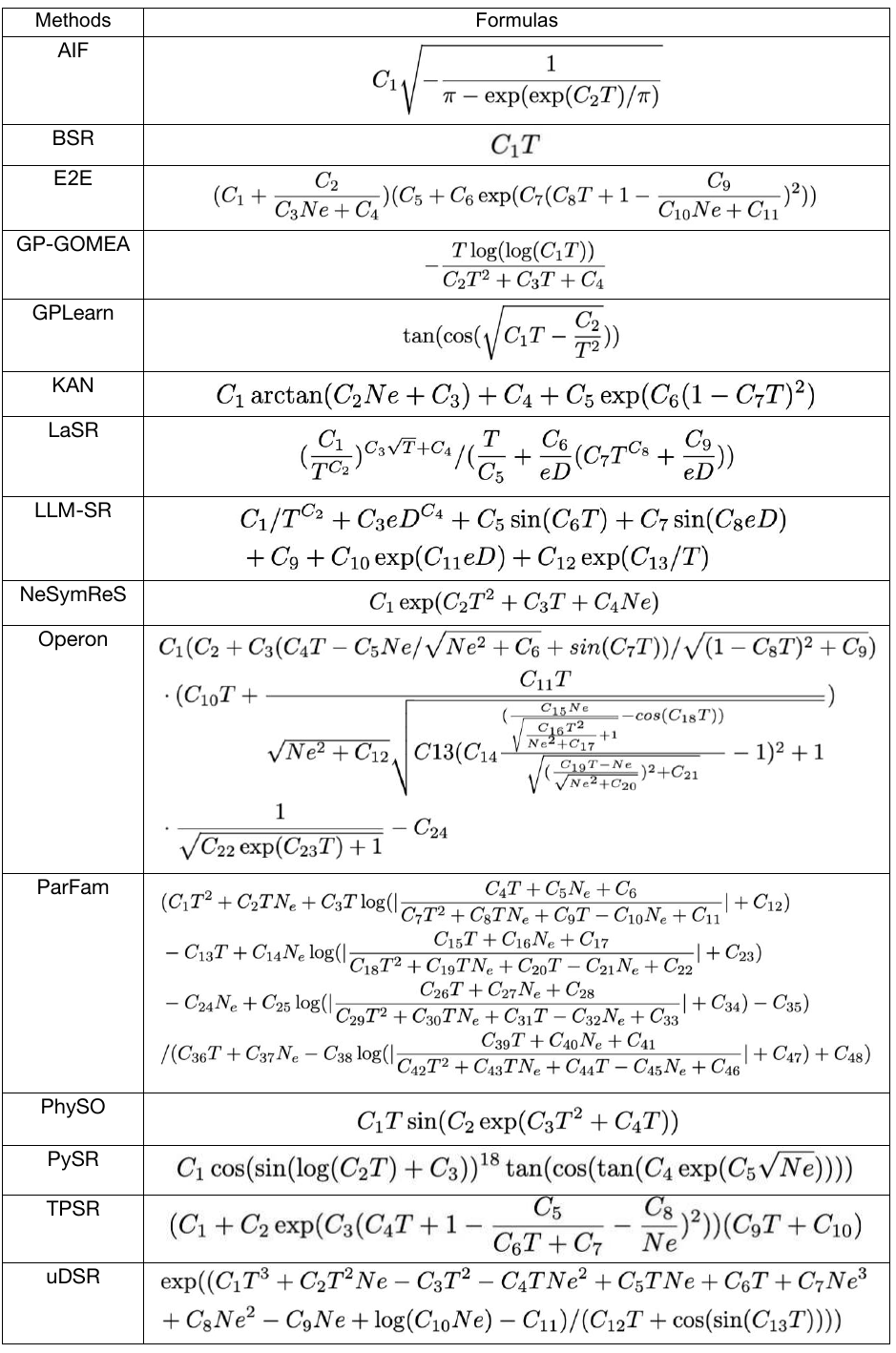}
\end{table}

\newpage

\begin{table}[!htb]
\centering
\caption{Constants of derived formulas from all the baseline models for the prediction of contribution functions}
\label{tab:phy4_consts}
\includegraphics[width=1.0\textwidth]{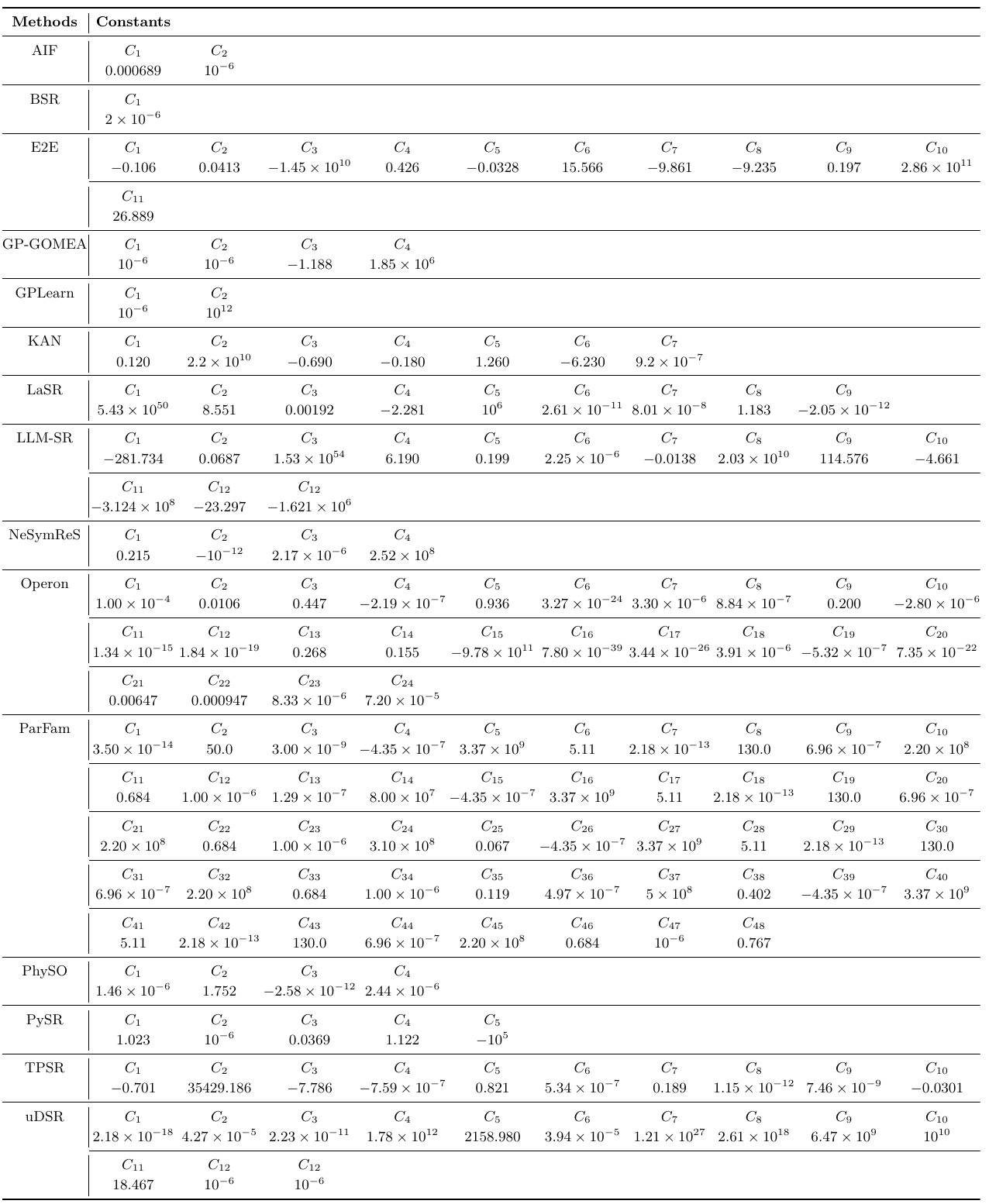}
\end{table}

\newpage

\begin{table}[!htb]
\centering
\caption{Derived formulas from all the baseline models for the prediction of lunar tide signal}
\label{tab:phy5_baselines}
\includegraphics[width=0.95\textwidth]{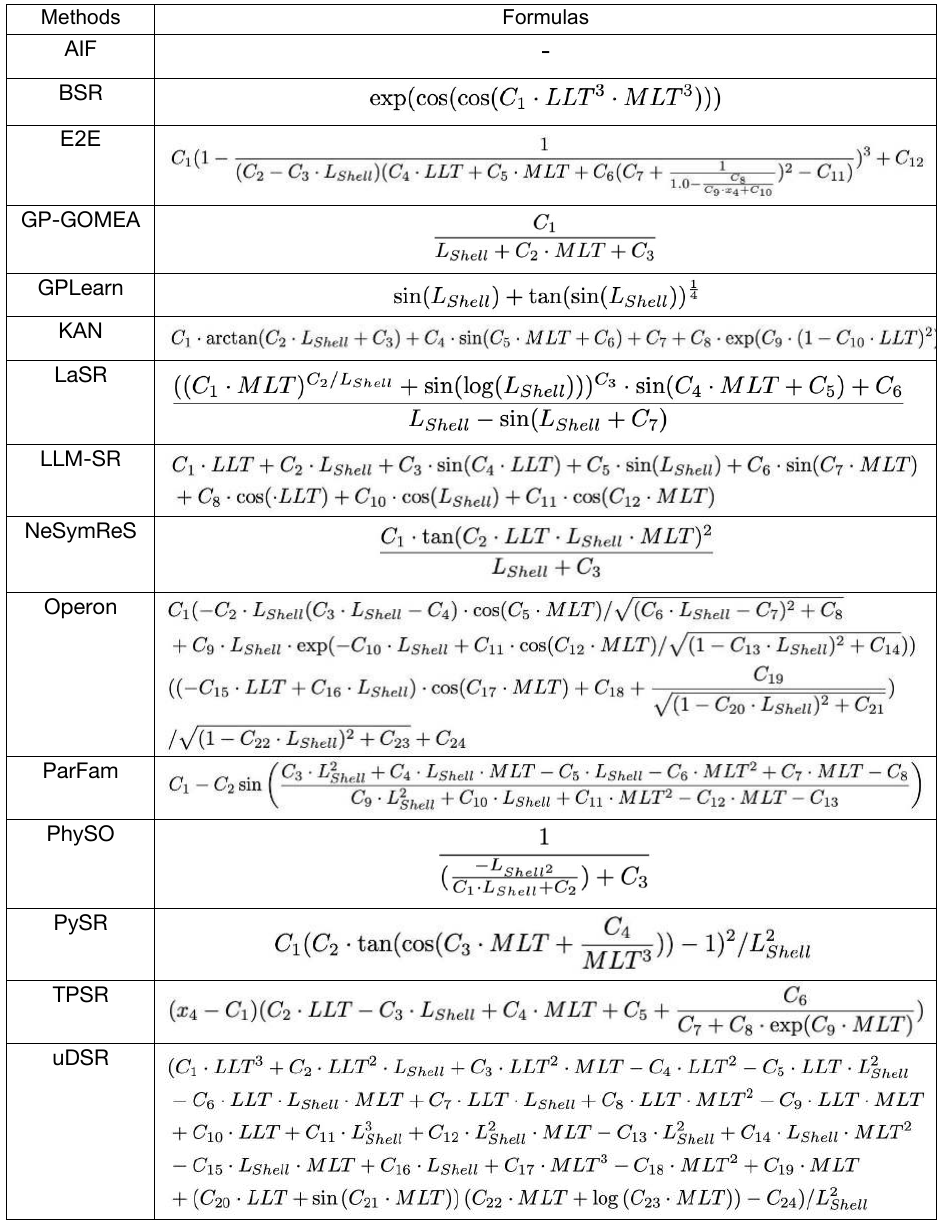}
\end{table}

\newpage

\begin{table}[!htb]
\centering
\caption{Constants of derived formulas from all the baseline models for the prediction of lunar tide signal}
\label{tab:phy5_consts}
\includegraphics[width=0.83\textwidth]{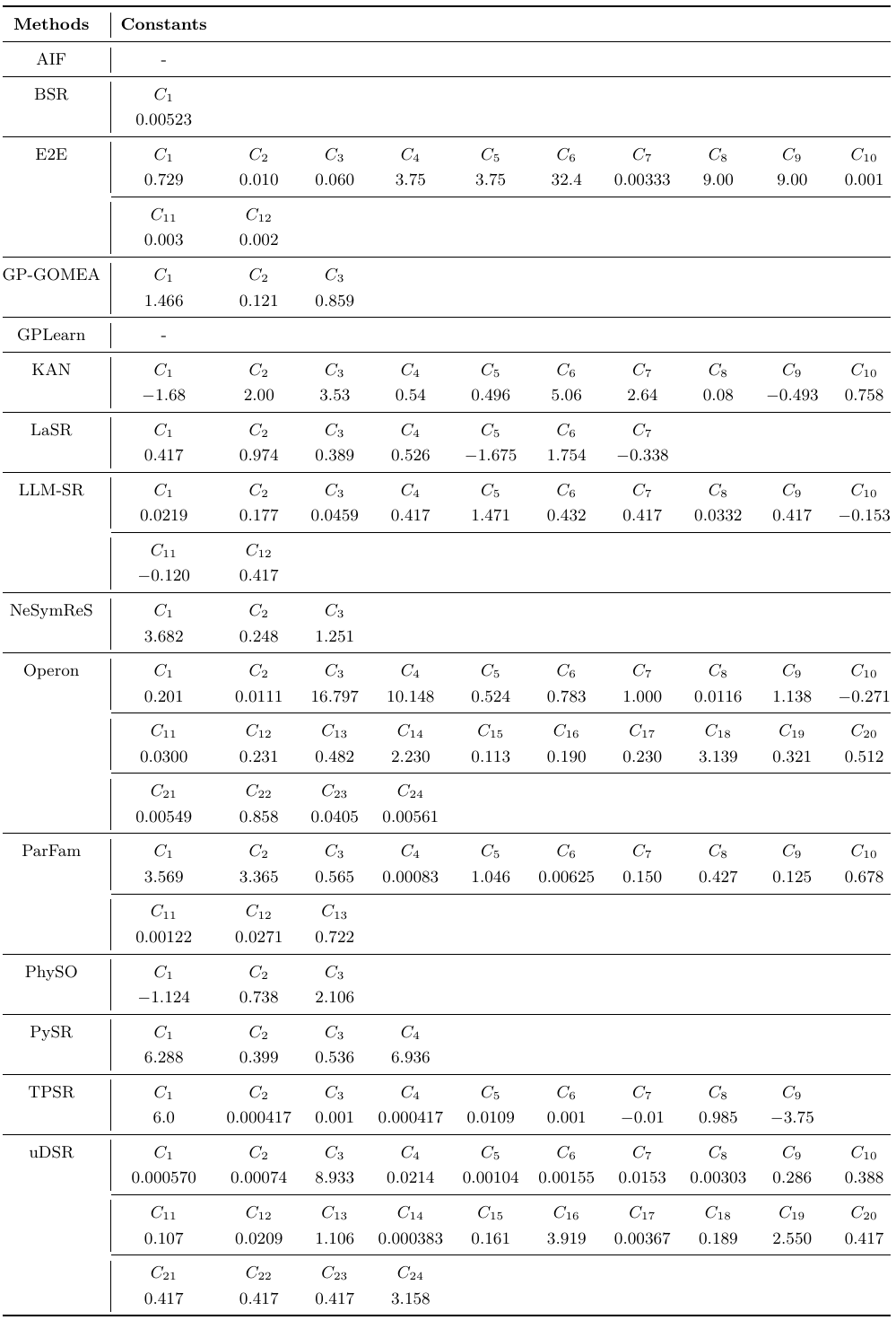}
\end{table}

\newpage

\section{Details settings of baseline Symbolic Regression methods}
\label{sec:baselines}

\paragraph{AIFeynman}
The code for AI Feynman is available at \url{https://github.com/SJ001/AI-Feynman}.
The parameters are set to be the same as their examples, allowing for 14 possible operators: 
$\{$+, *, -, /, +1, -1, neg, inv, sqrt, $\pi$, $\sin$, $\cos$, $\ln$, $\exp$$\}$, with a maximum time for brute-force search of 30 seconds, a maximum polynomial fitting degree of 3, and 500 training epochs for the neural network.
We run the model using the command \texttt{run\_aifeynman("../data/", "data.txt", 30, "14ops.txt")}.

\paragraph{BSR}
The Bayesian Symbolic Regression (BSR) is available at \url{https://github.com/ying531/MCMC-SymReg}, which includes a Bayesian framework and an efficient Markov Chain Monte Carlo (MCMC) algorithm. 
We use the default parameter, allowing for a total of 50 iterations and 3 output formula trees.
We run the model using the command \texttt{bsr.fit(X, y)}.

\paragraph{EndtoEnd}
The code for EndtoEnd model can be downloaded from the official repository at \url{https://github.com/facebookresearch/symbolicregression}. 
We also download the pre-trained model from \url{https://dl.fbaipublicfiles.com/symbolicregression/model1.pt} and use the default parameters provided in the demo scripts.
The model is allowed to search for operators from the following set: $\{$add, sub, mul, div, abs, inv, sqrt, log, exp, sin, arcsin, cos, arccos, tan, arctan, pow2, pow3$\}$.
We only allow at most 200 datapoints in a single bag to fit for each problem, same as default setting, while the remaining of data are split into at most 10 bags, each of which is fitted independently.
We run the model using the command \texttt{est.fit(X, y)}.

\paragraph{GP-GOMEA}
The Gene-pool Optimal Mixing Evolutionary Algorithm for Genetic Programming (GP-GOMEA) is available at \url{https://github.com/marcovirgolin/gpg}. The operators are allowed to search from the following set: $\{$add, sub, mul, div, sqrt, log$\}$.
Other default parameters follow their default configuration, including popsize=64, batchsize=64, time\_limit=100.
We run the model using the command \texttt{gpg.fit(X, y)}.

\paragraph{GPLearn}
The GPLearn model is also a GP-based method available at \url{https://github.com/trevorstephens/gplearn}. We use a function set of $\{$add, sub, mul, div, sin, cos, tan, sqrt, log, exp$\}$. Other parameters follow their default configuration, which include population\_size=5000, generations=20, a crossover probability $p_{\text{crossover}} = 0.7$, a subtree mutation probability $p_{\text{subtree mutation}} = 0.1$, a hoist mutation probability $p_{\text{hoist mutation}} = 0.05$, a point mutation probability $p_{\text{point mutation}} = 0.1$, and a maximum sample fraction of 0.9.
We run the model using the command \texttt{est\_gp.fit(X, y)}.

\paragraph{KAN}
The Kolmogorov-Arnold Networks (KAN) model is available at \url{https://github.com/KindXiaoming/pykan}.
The KAN trains a neural network with learnable activation functions, which is then fitted to the most likely operators selected from the following set: 
$\{x, x^2, x^3, x^4, 1/x, 1/x^2, $
$1/x^3, 1/x^4, \sqrt{x}, 1/\sqrt{x}, \exp(x), \log(x), \text{abs}(x), \sin(x), \tan(x), \tanh(x), \text{sigmoid}(x), \arcsin(x), \arctan(x), $
$\text{arctanh}(x), 0, \cosh(x), \text{gaussian}(x)\}$ to derive its symbolic representation.
Proper design on KAN's architecture is crucial when addressing practical problems.To balance its computation complexity and accuracy, we train the KAN model three times for each problem, using different network widths: $(nv, 1)$, $(nv, 5, 1)$, and $(nv, 3, 3, 1)$, respectively. These widths are among the most popular choices for solving Feynman equations. Other parameters follow their default configuration, including the number of grid intervals = 5, the order of piecewise polynomial = 3.
We run the model using the command \texttt{formula = learning(kan, dataset, verbose)}.

\paragraph{LaSR}
The Library Augmented Symbolic Regression (LaSR), built upon PySR and leveraging the strength of Large Language Models (LLMs), is available at 
\url{https://github.com/trishullab/LibraryAugmentedSymbolicRegression.jl}.
The Llama-3-8B model is used as local LLM engine. We employ a function set of $\{$add, sub, mul, div, pow, exp, log, sin, cos, sqrt$\}$ with constraints and nested constraints as specified in the default configuration. No additional hints are provided, except for the units of each variable, as we believe it is impractical to derive more meaningful task-specific hints for symbolic regression. We allow a maximum iterations of 100 and run the model using the command \texttt{hall\_of\_fame = equation\_search(X, y; niterations=100)}.

\paragraph{LLM-SR}
The LLM-SR, a novel approach leveraging the powers of LLMs for symbolic regression, is available at \url{https://github.com/deep-symbolic-mathematics/LLM-SR}.
We employ the Mixtral-8x7B model as our local LLM engine, setting \texttt{max\_sample\_num} to be 1000 to accommodate the large evaluation dataset. The maximum number of generated tokens is limited to 1024, with a search temperature of 0.8. The \texttt{top\_K} and \texttt{top\_p} parameters are set to be 30 and 0.9 respectively, following the default configuration.
For each different symbolic regression problem, we modify the general specification as: \texttt{Find the mathematical function skeleton that represents the relationship between input and output variables, given the\\ input variables with their physical units, and output variables with their physical units.}
We run the model using the command \texttt{pipeline.main(specification,
            inputs=dataset,
            config=config,
            max\_sample\_nums=max\_sample\_num,    \\ 
            class\_config=class\_config,
            log\_dir=args.log\_path)}.

\paragraph{NeSymRes}
The Neural Symbolic Regression that Scale (NeSymReS) is available at \url{https://github.com/SymposiumOrganization/NeuralSymbolicRegressionThatScales}. We use the 100M pre-trained model downloaded from \url{https://drive.google.com/drive/folders/1LTKUX-KhoUbW-WOx-ZJ8KitxK7Nov41G?usp=sharing}.
The model is allowed to search for operators from the following set: $\{$abs, acos, add, asin, atan, cos, cosh, coth, div, exp, ln, mul, pow, sin, sinh, sqrt, tan, tanh$\}$, same as default configuration.
We run the model using the command \texttt{output = fitfunc(X,y)}.

\paragraph{Operon}
The Operon is another GP-based symbolic regression methods available at \url{https://github.com/heal-research/operon}.
We use a function set of $\{$add, sub, mul, aq, sin, cos, tan, log, exp$\}$. Other parameters follow their default configuration.
We run the model using the command \texttt{reg.fit(X, y)}.

\paragraph{ParFam}
The ParFam model is available at \url{https://github.com/Philipp238/parfam}.
We use the default configuration, as suggested by the authors and run the model using the command \texttt{ParFamWrapper(config\_name='big', iterate=True).fit(X, y, time\_limit=500)}.

\paragraph{PhySO}
The Physical Symbolic Optimizer (PhySO) is available at \url{https://github.com/WassimTenachi/PhySO}.
We use the default \texttt{config0} for our experiment for the efficiency consideration.
The operators are allowed from the following set: $\{$add, sub, mul, div, inv, n2, sqrt, neg, log, exp, sin, cos, tan$\}$.
Each problem is assigned a fixed constant "1" and three free constants, all without physical units. The physical units of the input-output variables are also provided to the model.
We run the model using the command \texttt{expression, logs = physo.SR(X, y,
                                X\_units=x\_units, 
                                y\_units=y\_units, 
                                fixed\_consts=[1.],
                                fixed\_consts\_units = np.zeros((1, 5)),
                                free\_consts\_units =  np.zeros((3, 5)),
                                run\_config = config,
                                op\_names=["mul", "add", "sub", "div", "inv", "n2", "sqrt", "neg", "exp", "log", "sin", "cos", "tan"])}.

\paragraph{PySR}
The PySR model is available at \url{https://github.com/MilesCranmer/PySR}.
We use a function set of $\{$add, sub, mul, div, square, cube, exp, log, sin, cos, sqrt$\}$, with maximum iteration of 400, as specified in the default configuration.
We also incorporate the physical units of variables into the searching process, penalizing the incorrect units by adding a loss term with a coefficient of $10^5$, as provided in the default configuration for toy examples with dimensional constraints in PySR.
The model is executed using the command \texttt{model.fit(X, y, X\_units=x\_units, y\_units=y\_units)}.

\paragraph{TPSR}
The Transformer-based Planning for Symbolic Regression (TPSR) is available at \url{https://github.com/deep-symbolic-mathematics/TPSR}.
We use the EndToEnd model as backbone model for MCTS, hence adopting the same function set as the EndToEnd model. Other parameters are set to their default configuration. 
We run the model using the command \texttt{python tpsr.py --backbone\_model e2e --no\_seq\_cache True --no\_prefix\_cache True}.

\paragraph{uDSR}
The unified Deep Symbolic Regression (uDSO) is available at \url{https://github.com/dso-org/deep-symbolic-optimization}.
We use a function set of $\{$add, sub, mul, div, sin, cos, exp, log, poly$\}$.
Other parameters follow their default configuration.
We run the model using the command \texttt{model.fit(X, y)}.

\section{Detailed comparisons with PySR under different configurations}

In our initial experiments, we evaluated both PhyE2E and PySR using default parameters with physical units. To rigorously demonstrate PhyE2E's performance on both the Synthetic and Feynman datasets, we further conducted a comprehensive hyperparameter analysis for PySR based on the Tuning and Workflow Tips from its official documentation (\url{https://ai.damtp.cam.ac.uk/pysr/tuning/}) by systematically categorizing all the hyperparameters into three key dimensions, as follows,

\paragraph{Operator Sets}
To assess whether alternative operator sets could improve the performance of PySR, we explored five different unary operator sets for PySR during the search process. The Default set \{sin, cos, tan, exp, log, sqrt, square, cube\} includes all the operators and it was suggested by PySR to address most circumstances.  We also evaluated several subsets of the default operator sets categorized by operator types as the Trigonometry set \{sin, cos, tan\}, the Power set \{sqrt, square, cube\}, the Exponential set \{exp, log\}, Trigonometry+Power set, Trigonometry+Exponential set, Power+Exponential set. All the configurations contain the standard binary operator set \{add, sub, mul, div\}. 

We statistically analyzed the accuracy of formulas generated by PySR using different operator sets across datasets of varying complexity and difficulty. High- and low-complexity were defined based on whether the complexity of a formula was above or below the average complexity in the dataset. In addition, each formula was categorized into simple, medium, or hard difficulty levels, according to their similarity to the overall formula dataset (see Methods Section 4.4). 

On the Synthetic Dataset (Supplementary Figure ~\ref{fig:fig-S2}a), \ourModel demonstrated superior performance on formulas with high complexity, outperforming the best PySR configuration (using the Default unary set with 1000 iterations) by 27.61\%. For low-complexity formulas, only the PySR variant with the Trigonometry unary set and the Exponential+Trigonometry unary set using 1000 iterations surpassed the standard \ourModel, but it was still outperformed by \ourModelDCM. When evaluating across difficulty levels, both the standard PhyE2E and \ourModelDCM consistently outperformed all PySR configurations. Specifically, \ourModelDCM achieved improvements of 15.20\%, 20.15\%, and 17.50\% over the best PySR configuration on simple, medium, and hard formulas, respectively.

On the Feynman Dataset (Supplementary Figure ~\ref{fig:fig-S2}b), none of the PySR configurations outperformed \ourModelDCM. For both low-complexity and high-complexity formulas, several PySR configurations achieved higher symbolic accuracy than the standard \ourModel model. However, they were still outperformed by \ourModelDCM by 2.85\% and 9.70\%, respectively. Similar trends were also observed across different difficulty levels: \ourModelDCM outperformed the best PySR configuration by 6.81\%, 8.69\%, and 6.00\% on simple, medium, and hard formulas, respectively.

As a result, we found that different operator configurations for PySR can indeed improve  the symbolic accuracy of the detected formulas. However, the best operator sets vary on different complexity and difficulty levels. Our evaluation also showed that PySR with default operators and 1000 iterations generally delivers the optimal performance. Nevertheless, our \ourModelDCM model still outperformed this configuration, achieving symbolic accuracy improvement of 18.67\% and 9.00\% on the synthetic and Feynman datasets, respectively.

\paragraph{Computational Cost and Search Time}
We test the performance of PySR using different values of hyperparameters `niterations', `populations', `population\_size', and `ncycles\_per\_iteration'. The aim is to check whether additional computational effort could improve the performance of PySR. The operator set was set as the default value: {sin, cos, tan, exp, sqrt, square, cube, add, sub, mul, div}. We performed additional trials by adjusting the populations to \{10, 30, 50, 70\}, the population\_size to \{30, 45, 60\}, the ncycles\_per\_iteration to \{200, 380, 560\}, and the niterations to \{40, 100, 400, 1000\}. We selected the values of hyperparameters based on the default configuration and scaled within the range of 0.1× to 2.5× original values, all of which terminate before the 300 seconds time limit for \ourModel and PySR. All experiments were carried out on both the synthetic dataset and the AI Feynman dataset. 

All models showed improvements with increased computational effort (Supplementary Figure ~\ref{fig:fig-S3}a,b). On the synthetic dataset, the symbolic accuracy of PySR improved 8.00\% by varying populations from 10 to 70, improved 4.25\% by varying population\_size from 30 to 60, improved 2.33\% by varying ncycles\_per\_iterations from 200 to 560 and improved 13.83\% by varying niterations from 40 to 1000. And on the Feynman dataset, the symbolic accuracy of PySR improved 9.40\% by varying populations from 10 to 70, improved 0.00\% by varying population\_size from 30 to 60, improved 4.00\% by varying ncycles\_per\_iterations from 200 to 560 and improved 14.00\% by varying niterations from 40 to 1000.

Among all the hyperparameters, we found that increasing the niteartions yielded the best improvement. However, \ourModelDCM still outperformed all the PySR configurations by at least 18.25\% and 6.60\% improvement on the synthetic dataset and Feynman dataset, respectively. We did not further categorize the datasets into different complexity and difficulty levels, since all PySR configurations in this experiment exhibited consistent performance trends across varying levels.

Next, we analyzed the performance of both \ourModel configurations and PySR configuration when using the same computational time. 

All models showed improvements with increased computational costs (Supplementary Figure ~\ref{fig:fig-S3}c). On the Synthetic Dataset, the \ourModel family outperformed the PySR variants, leading by at least 15\% in symbolic accuracy across different levels of elapsed times. On the Feynman Dataset, \ourModel achieved a symbolic accuracy of 73.05\% with an elapsed time of 4.91s, which was surpassed by several PySR variants that required more than 50 seconds of computation. However, none of these PySR variants could outperform the performance of \ourModelDCM, which achieved higher accuracy with similar computational time resources. 

As a result, we found that increasing the search time by changing the configurations on the computational effort for PySR could indeed improve  the accuracy of the detected formulas. However, the best PySR configuration still could not match the performance of our \ourModelDCM model. Additionally, to achieve comparable symbolic accuracy, PySR required approximately 100× more search time. On the other hand, our \ourModel framework incorporates the MCTS and Genetic Programming (GP) refinement modules, which could also take advantage of increased computational time to further improve its accuracy.

\paragraph{Search Constraints}
We evaluated the performance of PySR using different values of ‘constraint’ and ‘nested\_constraint’. The ‘constraints’ parameter controls the complexity of the sub-formulas used within unary and binary operators, and the ‘nested\_constraints’ is used to limit the occurrence of nested unary operators to reduce the likelihood of deeply nested expressions.

Specifically, starting with the default parameters described above, we performed additional experiments by adding constraints=\{"sin": 10, "cos": 10, "tan": 10, "exp": 10, "log": 10\} to to restrict the complexity of sub-formulas used in unary operators to below 10, and adding nested\_constraints=\{"sin": \{"sin":0, "cos":0, "tan":0, "exp":0, "log":0\}, "cos": \{"sin":0, "cos":0, "tan":0, "exp":0, "log":0\}, "tan": \{"sin":0, "cos":0, "tan":0, "exp":0, "log":0\}, "exp": \{"sin":1, "cos":1, "tan":0, "exp":0, "log":1\}, "log": \{"sin":1, "cos":1, "tan":0, "exp":1, "log":0\}\} to prevent deeply nested unary expressions, allowing only a few simple compositions, such as exp(sin(A)) or exp(log(A)), while disallowing more complicated ones like exp(exp(A)) or sin(cos(A)). These constraints were reasonable as we verified that all the formulas in both datasets satisfied them. All the other hyperparameters were set to the default value.

We conducted additional experiments under four PySR settings: (a) without any constraints, (b) with only `constraints', (c) with only `nested\_constraints', and (d) with both constraints simultaneously. These configurations were evaluated on both the synthetic dataset and the AI Feynman dataset. We reported the results of four different PySR constraint configurations in terms of symbolic accuracy (Supplementary Figure ~\ref{fig:fig-S4}).

We found that none of the constrained variants of PySR showed comparable improvement over the unconstrained versions, yielding only a symbolic accuracy increase of 2.23\% and 3.00\% on the synthetic and Feynman datasets, respectively. Notably, both were achieved using the (c)nested\_constraints variant, while applying both types of constraints simultaneously did not lead to better performance. These results suggested that imposing varying levels of constraints offered limited benefit, and tighter constraints did not always lead to improved performance. We did not categorize the datasets into different complexity and difficulty levels either, since all PySR configurations in this experiment showed a similar performance trend across varying levels.

\paragraph{\quad}
\hspace{4pt} To summarize the above three classes of experiments, after an exhaustive investigation of all possible hyperparameter configurations in PySR, we identified only two factors that could meaningfully enhance performance. Firstly, employing a task-specific operator set, could help reduce the search space and improve efficiency for a specific symbolic regression task. However, there was not a single universal operator set that performs the best across all formula subclasses within the datasets, and the default operator set achieved competitive performance among all candidate operator sets in general. Secondly, allocating more search iterations could also improve performance. However, this would substantially increase the search time, while additional search time could also be utilized by our \ourModel using MCTS and GP refinement modules to further improve its accuracy. We also found that imposing PySR search constraints did not yield much improvement for both datasets.

\section{Divide-and-Conquer Algorithm}
\label{sec:algorism}

\begin{algorithm}
\caption{Construction of the possible $\sigma$-divisions}
\label{algorism1}
\textbf{Input:} the oracle neural network $\tilde{f}_\theta(\bm x)$, the uni-variate operator $\sigma$\\
\textbf{Output:} the set of possible $\sigma$-divisions \(\ \mathcal{B} = \{\mathcal{A}_k\}_{k=1}^s\)

\tcp{estimation of inner-variable relationship}
\For{each distinct features \(i\) and \(j\)}{
    calculate $J_{i, j}(\tilde{f}_\theta, \sigma) = \text{median}_{1\leq k\leq N}(|\frac{\partial^2 \sigma^{-1}\circ \tilde{f}_\theta}{\partial x_i\partial x_j}(\bm x_k)|)$\;
}
\tcp{adaptive threshold strategy}
Initialize the set of possible classes of $\sigma$-separable feature pars $\mathcal{S} = \{\}$\;
Initialize the class of $\sigma$-separable features pairs $S = \{\}$\;
Sort $J_{i, j}(\tilde{f}_\theta, \sigma)$ in non-decreasing order and calculate $\epsilon_0, \epsilon_1, \epsilon_2$\;
\For{each $J_{i, j}(\tilde{f}_\theta, \sigma)$ in the sorted order}{
     $S\gets S\cup \{(i, j)\}$\;
    \If{$\epsilon_1 \leq J_{i, j}(\tilde{f}_\theta, \sigma)  \leq \epsilon_2$}{
        $\mathcal{S} \gets \mathcal{S}\cup \{S\}$\;
    }
}
\tcp{construction of the possible $\sigma$-divisions}
Initialize the set of possible $\sigma$-divisions $\mathcal{B} = \{\}$\;
\For{each set of $\sigma$-separable features $S\in \mathcal{S}$}{
    Initialize $\sigma$-division $\mathcal{A} = \{\{1, 2, ..., n\}\}$\;
    \For{any feature pairs $(i, j)\in S$}{
        \For{each \(A \in \mathcal{A}\)}{
            \If{\(i \in A\) and \(j \in A\)}{
                $\mathcal{A} \gets (\mathcal{A} - \{A\}) \cup \{A-\{i\}, A-\{j\}\}$\;
            }
        }
    }
    $\mathcal{B} \gets \mathcal{B} \cup \{\mathcal{A}\}$\;
}
\For{each $\mathcal{A}_k \in \mathcal{B}$}{
    \For{each two element $A_i, A_j \in \mathcal{A}$}{
        \If{$A_i\subseteq A_j$}{
            $\mathcal{A}_k \gets \mathcal{A}_k - \{A_i\}$\;
        }
    }
}
\textbf{Return} \(\mathcal{B}\)\;
\end{algorithm}

\section{Proofs for the divide-and-conquer strategy}

\subsection{Proofs of the decomposition step}
\label{sec:proof-decomposition}

\begin{proof}[Proof of Lemma 1]
By the Inverse Function Theorem, the uni-variate operator $\sigma$ has an inverse $\sigma^{-1}$, such that $\sigma\circ\sigma^{-1} = \mathrm{Id}$ and $\sigma^{-1}\circ\sigma = \mathrm{Id}$.

The ``only if'' part is straightforward. Suppose two features $i$ and $j$ are $\sigma$-separable. By the definition of being $\sigma$-separable, we have that
\[
\sigma^{-1}\circ f(\bm x) = f_1(\bm x_{-i}) + f_2(\bm x_{-j}).
\]
Straightforward calculation shows that for each $\bm{x} \in \mathbb{R}^n$,
\[
\frac{\partial^2 \sigma^{-1}\circ f(\bm x)}{\partial x_i \partial x_j} = 0 .
\]

Now we turn to the ``if'' part. Suppose we have that
\[
\frac{\partial^2 \sigma^{-1}\circ f(\bm x)}{\partial x_i \partial x_j} = 0.
\]
Integrating both sides over $x_j$, we get that there exists $g_2(\bm{x}_{-j})$ such that
\[
\frac{\partial \sigma^{-1} \circ f(\bm{x})}{\partial x_i} = g_2(\bm{x}_{-j}) .
\]
Now integrate both sides over $x_i$. We have that there exists $g_1(\bm{x}_{-i})$ such that
\[
\sigma^{-1} \circ f(\bm{x}) = g_1(\bm{x}_{-i}) + \int g_2(\bm{x}_{-j}) \mathrm{d} x_i .
\]
Let $f_1(\bm{x}_{-i}) = g_1(\bm{x}_{-i})$ and $f_2(\bm{x}_{-j}) = \int g_2(\bm{x}_{-j}) \mathrm{d} x_i$, we conclude that features $i$ and $j$ are $\sigma$-separable.
\end{proof}

\begin{proof}[Proof of Lemma 2]
We prove the lemma by induction on the number of iterations $\ell$.

\medskip
\noindent \underline{\it Induction basis.} When $\ell = 1$, the lemma reduces to the definition of being $\sigma$-separable.

\medskip
\noindent \underline{\it Inductive step.} Suppose the induction hypothesis holds for any $\ell \geq 1$. Now we consider the iteration $(\ell+1)$. By the induction hypothesis, we can express $f(\bm x)$ as 
\begin{align}
f(\bm x)=\sigma\left(\sum_{k=1}^{m_\ell}f_k^{(l)}(\bm x_{A_k^{\ell}})\right) . \label{eq:proof-decomposition-lemma-1}
\end{align}

Suppose we select feature pair $(i, j)$ that is $\sigma$-separable at the $(\ell+1)$-th iteration. By the definition of being $\sigma$-separable, we can also express $f(\bm x)$ as 
\begin{align}
f(\bm x)=\sigma\left(f_1(\bm x_{-i})+f_2(\bm x_{-j})\right),    \label{eq:proof-decomposition-lemma-2}
\end{align}
where $\bm x_{-i}$ is the $(n - 1)$-dimensional vector obtained by removing $x_i$ from $\bm x$. We further define $\bm x_{(-i,-j)}$ as the $(n - 2)$-dimensional vector obtained by removing $x_i$ and $x_j$ from $\bm x$.

Combining Eqs.~(\ref{eq:proof-decomposition-lemma-1},\ref{eq:proof-decomposition-lemma-2}) and using the fact that $\sigma^{-1}$ exists, for any fixed $\alpha$ and $\beta$ (e.g., $\alpha = \beta = 0$),  we have 
\begin{align}
f_1(\bm x_{-i})+f_2(\bm x_{(-i,-j)},x_i=\alpha)&=\sum_{k=1}^{m_\ell}f_k^{(\ell)}\left((\bm x_{A_k^\ell-\{i\}}, x_i=\alpha)_{A_k^\ell}\right), \label{eq:proof-decomposition-lemma-3}\\
f_1(\bm x_{(-i,-j)},x_j=\beta)+f_2(\bm x_{-j})&=\sum_{k=1}^{m_\ell}f_k^{(\ell)}\left((\bm x_{A_k^\ell-\{j\}}, x_j=\beta)_{A_k^\ell}\right), \label{eq:proof-decomposition-lemma-4}\\
f_1(\bm x_{(-i,-j)},x_j=\beta)+f_2(\bm x_{(-i,-j)},x_i=\alpha)&=\sum_{k=1}^{m_\ell}f_k^{(\ell)}\left((\bm x_{A_k^\ell-\{i,j\}}, x_i=\alpha, x_j=\beta)_{A_k^\ell}\right). \label{eq:proof-decomposition-lemma-5}
\end{align}

Combining Eqs.~(\ref{eq:proof-decomposition-lemma-2},\ref{eq:proof-decomposition-lemma-3},\ref{eq:proof-decomposition-lemma-4},\ref{eq:proof-decomposition-lemma-5}), we further have
\begin{align}
f(\bm x)&=\sigma\left(f_1(\bm x_{-i})+f_2(\bm x_{-j})\right) \notag \\
&=\sigma\Big(
\sum_{k=1}^{m_\ell}f_k^{(\ell)}\left((\bm x_{A_k^\ell-\{i\}}, x_i=\alpha)_{A_k^\ell}\right)
+f_k^{(\ell)}\left((\bm x_{A_k^\ell-\{j\}}, x_j=\beta)_{A_k^\ell}\right) \notag \\
&\qquad\qquad\qquad\qquad\qquad\qquad\qquad\qquad
-f_k^{(\ell)}\left((\bm x_{A_k^\ell-\{i,j\}}, x_i=\alpha, x_j=\beta)_{A_k^\ell}\Big)
\right) . \label{eq:proof-decomposition-lemma-6}
\end{align}
Note that $f_k^{(\ell)}\left((\bm x_{A_k^\ell-\{i\}}, x_i=\alpha)_{A_k^\ell}\right)$ is a function of $\bm{x}_{A_k^\ell -\{i\}}$ and $f_k^{(\ell)}\left((\bm x_{A_k^\ell-\{j\}}, x_j=\beta)_{A_k^\ell}\right) - f_k^{(\ell)}\left((\bm x_{A_k^\ell-\{i,j\}}, x_i=\alpha, x_j=\beta)_{A_k^\ell}\right)$ is a function of $\bm{x}_{A_k^\ell - \{j\}}$. Therefore, Eq.~\eqref{eq:proof-decomposition-lemma-6} implies that $\{A_k^{\ell} - \{i\}, A_k^{\ell} - \{j\}\}$, after removing multiplicities in the class, forms a $\sigma$-division, which proves the lemma for iteration $(\ell + 1)$.
\end{proof}

\subsection{Proof of the aggregation theorem}
\label{sec:proof-aggregation}

Before proving Theorem 3, we first prove two useful lemmas as follows.

\begin{lemma}\label{lem:const-evaluation}
Suppose the uni-variate operator $\sigma$ : $\mathbb{R} \to \mathbb{R}$ is strictly monotonic.
If $\{A_i\}_{i=1}^{m}$ is a $\sigma$-partition of the target formula $f : \mathbb{R}^n \to \mathbb{R}$, i.e., $f$ can be expressed as $f(\bm x) = \sigma(f_1(\bm x_{A_1}) + \dots + f_m(\bm x_{A_m})) $. 
Then for any $\bm{\chi} \in \mathbb{R}^n$, and each $i \in \{1, 2, 3, \dots, m\}$, we have that
\[
f_i(\bm{\chi}_{A_i}) = \sum_{\emptyset \subsetneq I \subseteq [m]} (-1)^{|I|-1}f_i(\bm x_{A_i\cap \mathcal{A}_{I}}=\bm{\chi}_{A_i\cap \mathcal{A}_{I}}, \bm x_{A_i-\mathcal{A}_{I}}=\bm z_{A_i-\mathcal{A}_{I}}).
\]
\end{lemma}
\begin{proof}
For any $\bm{\chi} \in \mathbb{R}^n$, and each $i \in \{1, 2, 3, \dots, m\}$, we have that
\begin{align}
    & \sum_{\emptyset \subsetneq I \subseteq [m]} (-1)^{|I|-1}f_i(\bm x_{A_i\cap \mathcal{A}_{I}}=\bm{\chi}_{A_m\cap \mathcal{A}_{I}}, \bm x_{A_m-\mathcal{A}_{I}}=\bm z_{A_i-\mathcal{A}_{I}})\notag\\
    &\qquad \qquad =  \sum_{\{i\} \subseteq I \subseteq [m]} (-1)^{|I|-1}f_i(\bm x_{A_i\cap \mathcal{A}_{I}}=\bm{\chi}_{A_i\cap \mathcal{A}_{I}}, \bm x_{A_i-\mathcal{A}_{I}}=\bm z_{A_i-\mathcal{A}_{I}})\notag\\
    & \qquad \qquad \qquad \qquad +\sum_{\emptyset \subsetneq I \subseteq [m] - \{i\}} (-1)^{|I|-1}f_i(\bm x_{A_i\cap \mathcal{A}_{I}}=\bm{\chi}_{A_i\cap \mathcal{A}_{I}}, \bm x_{A_i-\mathcal{A}_{I}}=\bm z_{A_i-\mathcal{A}_{I}}) . \notag
\end{align}
Also note that
\begin{align*}
& \sum_{\{i\} \subseteq I \subseteq [m]}  (-1)^{|I|-1}f_i(\bm x_{A_i\cap \mathcal{A}_{I}}=\bm{\chi}_{A_i\cap \mathcal{A}_{I}}, \bm x_{A_i-\mathcal{A}_{I}}=\bm z_{A_i-\mathcal{A}_{I}}) \\
& \qquad \qquad  = f_i(\bm{\chi}_{A_i}) +  \sum_{\{i\} \subsetneq I \subseteq [m]} (-1)^{|I|-1}f_i(\bm x_{A_i\cap \mathcal{A}_{I}}=\bm{\chi}_{A_i\cap \mathcal{A}_{I}}, \bm x_{A_i-\mathcal{A}_{I}}=\bm z_{A_i-\mathcal{A}_{I}}) \\
& \qquad \qquad  = f_i(\bm{\chi}_{A_i}) +  \sum_{\{i\} \subsetneq I \subseteq [m]} (-1)^{|I|-1}f_i(\bm x_{A_i\cap \mathcal{A}_{I - \{i\}}} = \bm{\chi}_{A_i \cap \mathcal{A}_{I - \{i\}}}, \bm x_{A_i-\mathcal{A}_{I - \{i\}}}=\bm z_{A_i-\mathcal{A}_{I - \{i\}}}) \\
& \qquad \qquad  = f_i(\bm{\chi}_{A_i}) -  \sum_{\emptyset \subsetneq I \subseteq [m] - \{i\}} (-1)^{|I|-1}f_i(\bm x_{A_i\cap \mathcal{A}_{I}}=\bm{\chi}_{A_i\cap \mathcal{A}_{I}}, \bm x_{A_i-\mathcal{A}_{I}}=\bm z_{A_i-\mathcal{A}_{I}}) .
\end{align*}
Altogether, we prove the desired equality .
\end{proof}

\begin{lemma}\label{lem:lem-variance}
Under the condition of Theorem 3, for each $I \subseteq [m]$ and every $k \in I$, we have that 
\[
g_{k}( \bm x_{\mathcal{A}_{I}}, \bm{x}_{A_{k}-\mathcal{A}_{I}}=\bm z_{A_{k}-\mathcal{A}_I})
= \sigma\left(\sum_{i=1}^{m}f_i(\bm x_{A_i\cap \mathcal{A}_I}, \bm x_{A_i - \mathcal{A}_I} = \bm z_{A_i - \mathcal{A}_I})\right).
\]
\end{lemma}

\begin{proof}
This lemma follows directly from the definition of $g_k(\bm x_{A_k})$:
\[
g_k(\bm x_{A_k})=f(\bm x_{A_k}, \bm x_{\overline{A_k}}=\bm z_{\overline{A_k}})=\sigma\left(\sum_{i=1}^mf_i(\bm x_{A_i\cap A_k}, \bm x_{A_i-A_k}=\bm z_{A_i-A_k})\right).
\]
Thus, noting that $\mathcal{A}_I \subseteq A_k$, we get
\[
g_{k}(\bm x_{\mathcal{A}_{I}}, \bm{x}_{A_{k}-\mathcal{A}_{I}}=\bm z_{A_{k}-\mathcal{A}_I}) =\sigma\left(\sum_{i=1}^mf_i(\bm x_{A_i\cap \mathcal{A}_I}, \bm x_{A_i-\mathcal{A}_I}=\bm z_{A_i-\mathcal{A}_I})\right).
\]
\end{proof}

Now, we are ready to prove Theorem 3.

\begin{proof}[Proof of Theorem 3]
By Lemma~\ref{lem:const-evaluation}, for any $\bm{\chi} \in \mathbb{R}^n$, and each $i \in \{1, 2, 3, \dots, m\}$, we have
\[
f_i(\bm{\chi}_{A_i}) = \sum_{\emptyset \subsetneq I \subseteq [m]} (-1)^{|I|-1}f_i(\bm x_{A_i\cap \mathcal{A}_{I}}=\bm{\chi}_{A_i\cap \mathcal{A}_{I}}, \bm x_{A_i-\mathcal{A}_{I}}=\bm z_{A_i-\mathcal{A}_{I}}) .
\]
Summing the above equality over $1\leq i\leq m$, we have
\begin{align}
\sum_{i=1}^m f_i(\bm{\chi}_{A_i}) &= \sum_{i=1}^{m}\sum_{\emptyset \subsetneq I \subseteq [m]} (-1)^{|I|-1}f_i(\bm x_{A_i\cap \mathcal{A}_{I}}=\bm{\chi}_{A_i\cap \mathcal{A}_{I}}, \bm x_{A_i-\mathcal{A}_{I}}=\bm z_{A_i-\mathcal{A}_{I}})\notag\\
&= \sum_{\emptyset \subsetneq I\subseteq [m]} \frac{(-1)^{|I|-1}}{|I|}\sum_{t\in I}\sum_{i=1}^{m}f_i(\bm x_{A_i\cap \mathcal{A}_{I}}=\bm{\chi}_{A_i\cap \mathcal{A}_{I}}, \bm x_{A_i-\mathcal{A}_{I}}=\bm z_{A_i-\mathcal{A}_{I}})\notag\\
&= \sum_{\emptyset \subsetneq I\subseteq [m]} \frac{(-1)^{|I|-1}}{|I|}\sum_{t\in I}\sigma^{-1}\circ g_t(\bm x_{\mathcal{A}_{I}}=\bm{\chi}_{\mathcal{A}_{I}}, \bm x_{A_t-\mathcal{A}_{I}}=\bm z_{A_t-\mathcal{A}_{I}}), \notag
\end{align}
where the last equality is due to Lemma~\ref{lem:lem-variance}. Applying $\sigma$ to both sides of the equality above, we prove the theorem.
\end{proof}

\end{document}